\PassOptionsToPackage{numbers,sort&compress}{natbib}
\documentclass[preprint,11pt,numbers]{elsarticle}
\PassOptionsToPackage{unicode}{hyperref}
\PassOptionsToPackage{hyphens}{url}
\PassOptionsToPackage{dvipsnames,svgnames,x11names}{xcolor}

\usepackage{amsmath,amssymb,bm}
\usepackage{dsfont}
\usepackage{iftex}
\ifPDFTeX
  \usepackage[T1]{fontenc}
  \usepackage[utf8]{inputenc}
  \usepackage{textcomp}
\else
  \usepackage{unicode-math}
  \defaultfontfeatures{Scale=MatchLowercase}
  \defaultfontfeatures[\rmfamily]{Ligatures=TeX,Scale=1}
\fi
\usepackage{lmodern}
\usepackage[british]{babel}
\IfFileExists{microtype.sty}{%
  \usepackage[]{microtype}
  \UseMicrotypeSet[protrusion]{basicmath}
}{}
\IfFileExists{upquote.sty}{\usepackage{upquote}}{}
\IfFileExists{xurl.sty}{\usepackage{xurl}}{}

\usepackage{enumitem}
\usepackage{indentfirst}
\usepackage{longtable,booktabs,array}
\usepackage{calc}
\usepackage{graphicx}
\usepackage{xcolor}
\usepackage{amsthm}
\usepackage[]{natbib}
\usepackage{setspace}
\usepackage{bookmark}

\setcitestyle{numbers,square}

\makeatletter
\def\maxwidth{\ifdim\Gin@nat@width>\linewidth\linewidth\else\Gin@nat@width\fi}
\def\maxheight{\ifdim\Gin@nat@height>\textheight\textheight\else\Gin@nat@height\fi}
\def\fps@figure{htbp}
\@ifpackageloaded{caption}{}{\usepackage{caption}}
\@ifpackageloaded{subcaption}{}{\usepackage{subcaption}}
\makeatother
\setkeys{Gin}{width=\maxwidth,height=\maxheight,keepaspectratio}
\hypersetup{
  colorlinks=true,
  linkcolor=Blue,
  filecolor=Blue,
  citecolor=Blue,
  urlcolor=Blue
}

\DeclareMathOperator{\DC}{DC}
\DeclareMathOperator{\cov}{Cov}

\newcommand{\tr}{\mathrm{tr}}

\newcommand{\bmf}{\mathbf{f}}
\newcommand{\bmx}{\mathbf{x}}
\newcommand{\bmy}{\mathbf{y}}

\newcommand{\bmu}{\mathbf{u}}

\newcommand{\bmv}{\mathbf{v}}

\newcommand{\bmk}{\mathbf{k}}
\newcommand{\bmK}{\mathbf{K}}
\newcommand{\bmSigma}{\bm{\Sigma}}

\newcommand{\bmLambda}{\bm{\Lambda}}
\newcommand{\bmtheta}{\bm{\theta}}
\newcommand{\bmnu}{\bm{\nu}}
\newcommand{\bmvartheta}{\bm{\vartheta}}
\newcommand{\bmzero}{\mathbf{0}}
\newcommand{\bmOmega}{\bm{\Omega}}

\newcommand{\bmA}{\mathbf{A}}
\newcommand{\bmI}{\mathbf{I}}

\newcommand{\bmeta}{\bm{\eta}}

\newcommand{\bmPi}{\bm{\Pi}}

\newcommand{\mgp}{\mathcal{MGP}}

\newtheorem{definition}{Definition}[section]

\newtheorem{theorem}{Theorem}[section]
\newtheorem{proposition}{Proposition}[section]

\newtheorem{remark}{Remark}[section]

\usepackage{placeins}

\usepackage{needspace}
\newcommand{\TaskScreenSaving}{14.01}
\newcommand{\TaskScreenSavingLow}{12.55}
\newcommand{\TaskScreenSavingHigh}{15.67}

\usepackage{geometry}
\hypersetup{
  pdftitle={When is multivariate kriging worthwhile? Design geometry, parameter sharing and response coupling under heterotopic sampling},
  pdfauthor={Zexun Chen, Jun Fan, Kuo Wang},
  pdfkeywords={Multivariate statistics, kriging, multi-output Gaussian process, Design of experiments, heterotopic sampling}
}

\makeatletter
\let\ps@pprintTitle\ps@plain
\makeatother
\let\tmptocnumberline\numberline
\begin{document}

\begin{frontmatter}

\title{{When is multivariate kriging worthwhile?
Design geometry, parameter sharing and response coupling
under heterotopic sampling}}

\author[edinburgh]{Zexun Chen\corref{cor1}}
\ead{Zexun.Chen@ed.ac.uk}
\author[nottingham]{Jun Fan}
\author[jiaxing]{Kuo Wang}

\cortext[cor1]{Corresponding author.}

\affiliation[edinburgh]{
  organization={University of Edinburgh Business School},
  city={Edinburgh},
  country={United Kingdom}
}
\affiliation[nottingham]{
  organization={Faculty of Science and Engineering, University of Nottingham Ningbo China},
  city={Ningbo},
  country={China}
}
\affiliation[jiaxing]{
  organization={College of Data Science, Jiaxing University},
  city={Jiaxing},
  country={China}
}

\begin{abstract}
Simulation experiments, multi-fidelity models and monitoring networks often produce related outputs at different input locations, a pattern known as heterotopic sampling.
{We investigate when multivariate kriging improves on separate prediction, using separable multi-output Gaussian processes as the principal theoretical benchmark. We introduce directed coverage, directed proximity and borrowing potential indices to describe design complementarity beyond exact overlap without response values. Under common radial kernels, 
we derive finite-design bounds linking geometry to prediction gains with known covariance parameters and information for estimating cross-output dependence. 
An exact risk decomposition under correct specification separates oracle gains from differences in estimation penalties; we give a mean-square expansion under explicit conditions. 
We distinguish parameter sharing from response coupling: auxiliary data can improve prediction through shared-scale estimation even when direct borrowing is negligible.
This distinction guides a calibrated screening-and-validation procedure comparing separate marginal models, a shared-scale independent model and a coupled model, retaining the shared-scale option when coupling is excluded. In simulations with sparse target data and auxiliary observations retained during validation, the procedure reduces mean latent squared-error risk by 6.66\% relative to separate fitting. After calibration, screening reduces runtime by 14.01\% relative to direct validation of the same candidates, with similar mean risk. Controlled experiments show that covariance misspecification can reverse fitted gains, while queueing and monitoring applications demonstrate the importance of auxiliary coverage and the prediction task. Improved mean prediction can coexist with predictive-interval undercoverage. These findings support task-matched validation and separate assessments of predictive accuracy, uncertainty calibration and computational cost.}

\end{abstract}

\begin{keyword}
Multivariate statistics \sep Kriging \sep Multi-output Gaussian process \sep Design of experiments \sep Heterotopic sampling
\end{keyword}

\end{frontmatter}

\clearpage
\setcounter{tocdepth}{2}
\tableofcontents
\clearpage

\section{Introduction}
\label{sec:intro}

Many systems studied in operational research produce several related outputs at once. 
A discrete-event simulation reports throughput, waiting time and utilisation. 
A multi-fidelity experiment couples a cheap, approximate solver with an expensive, accurate one. 
An environmental or infrastructure monitoring programme operates separate sensor networks for related quantities. 
Kriging, or Gaussian process (GP) regression, is now a standard metamodelling tool for such systems \citep{kleijnen2009kriging,ankenman2010stochastic,kleijnen2017regression,erickson2018comparison}, and a natural question follows immediately: should the outputs be metamodelled jointly, with a multi-output GP model that borrows strength across outputs, or separately, with independent univariate models?

The empirical record on this question is unsettled. {In a Monte Carlo study satisfying the assumed kriging models, \citet{kleijnen2014multivariate} report that univariate kriging often performs at least as well as multivariate kriging in their experimental settings.}
\citet{fricker2013multivariate} report similarly mixed experience with multi-output emulators. Yet multivariate spatial statistics and machine learning offer an extensive literature in which joint models do deliver substantial gains \citep{goulard1992linear,bonilla2007multi,alvarez2012kernels,wackernagel2013multivariate,genton2015cross}, and the multi-fidelity literature is built on the premise that a cheap auxiliary output improves prediction of an expensive one \citep{kennedy2000predicting,forrester2007multi,legratiet2014recursive}. 
{This paper shows that}\label{passage:intro_grammar}, within the separable multivariate kriging setting, these apparently conflicting findings are partly reconciled by the geometry of the output-specific designs, a variable that comparisons rarely control.\footnote{Nonseparable and autoregressive multi-fidelity models introduce an additional cross-structure mechanism, but our focus is the design-geometry mechanism that is already present in the separable case.}

The key distinction is between \emph{isotopic} sampling, in which all outputs are observed on a common design, and \emph{heterotopic} sampling, in which each output has its own input locations \citep{wackernagel2013multivariate}. Most simulation comparisons, including that of \citet{kleijnen2014multivariate}, are isotopic: every replication evaluates all responses at the same design points. 
{Under separable covariance models with a fixed common kernel and known mean, joint and marginal posterior means coincide for isotopic noise-free observations, a classical property known as autokrigeability \citep{wackernagel2013multivariate,chiles2012geostatistics}.
Estimating the cross-output covariance coefficients alone therefore leaves the posterior mean and its squared-error risk unchanged.}
Heterotopic sampling, by contrast, is ubiquitous whenever outputs are costed differently or collected by different instruments. Multi-fidelity experiments deliberately run the cheap code on a dense design and the expensive code on a sparse one \citep{kennedy2000predicting,qian2009nested}. Legacy simulation archives accumulate evaluations of different responses at different inputs \citep{kaminski2015method}. 
Monitoring networks for related pollutants rarely share sites. 
{Heterotopic sampling creates additional opportunities for auxiliary observations to support prediction; the question is when these opportunities are realised after fitting.}

Heterotopic designs are, however, not all alike, and the standard summary, the share of exactly co-located points, cannot tell them apart. 
Two designs with zero overlap may interleave, so that most locations of each output have nearby auxiliary support, or they may occupy disjoint regions of the input space. 
{The overlap index assigns zero to both configurations, although their opportunities for borrowing information can differ substantially.}

{We develop a theoretical account of how heterotopic design geometry affects prediction gains and information for estimating cross-output dependence, using separable multi-output GPs as the principal benchmark.
Auxiliary observations can affect fitted prediction through shared covariance-parameter estimation or through their direct contribution to the conditional predictor.
We refer to the latter as response coupling.
We distinguish these mechanisms when comparing gains available with known covariance parameters with performance after fitting. }{Practical value is assessed through predictive accuracy, uncertainty calibration and computational cost for specified target outputs, prediction regions,
observation-availability rules and fitting procedures.}
{We distinguish three levels of evidence: exact identities and bounds under their stated assumptions, an estimation-penalty expansion under explicit mean-square conditions, and empirical evidence from richer or misspecified models.}

The contributions are as follows.
\begin{enumerate}[label=(\roman*),nosep]
\item {We introduce directed coverage, directed proximity and borrowing potential indices to describe how output-specific designs cover one another and provide local auxiliary support.
These summaries go beyond exact overlap and, for a specified metric and prediction region, require no response values.}

\item {We derive finite-design bounds linking heterotopic geometry to oracle prediction gain. For a target--auxiliary pair, an explicit BPI-based lower bound provides sufficient conditions for positive borrowing, incorporating the target-design gap, cross-output dependence and observation noise.}

\item {We relate design geometry to the information available for estimating cross-output dependence. We bound conditional Fisher information through kernel-weighted cross-design
interactions and obtain an exact expression for efficient information at independence. The bounds also extend component-wise to the linear model of coregionalisation.}

\item {We explain how parameter estimation changes the comparison between joint and separate prediction. Under correct specification, an exact risk decomposition separates oracle gain
from the difference in their estimation penalties. We develop an expansion under explicit mean-square conditions and assess an information-based penalty approximation against refitting and predictive validation.}

\item {For target--auxiliary pairs, we develop and evaluate a conditional-gain screen that retains shared-scale estimation as an alternative to response coupling.
The assessment distinguishes the predictive value of the candidate catalogue from the computational effect of screening within that catalogue. Controlled simulations and queueing and monitoring applications examine how auxiliary coverage, target scarcity and covariance specification affect fitted
prediction, including both gains and losses.}

\end{enumerate}

The remainder of the paper is organised as follows. Section~\ref{sec:related} reviews related work in operational research and multivariate statistics. Section~\ref{sec:geometry} introduces the heterotopic setting and the geometric diagnostics. Sections~\ref{sec:oracle} and~\ref{sec:estimation} develop the theory, first for oracle prediction gain and then for the estimability of cross-output dependence. {Section~\ref{sec:net_benefit} presents the risk decomposition and penalty approximations. Section~\ref{sec:screening} separates shared-scale estimation from response coupling and gives a calibrated screening and validation procedure.} Section~\ref{sec:experiments} reports the numerical experiments and the case study, and Section~\ref{sec:conclusions} concludes.

\section{Related work}
\label{sec:related}

\subsection{Kriging metamodelling in operational research}
\label{subsec:related_or}

Kriging entered the simulation literature as a flexible interpolator of expensive input--output behaviour and is now established as a core metamodelling tool. \citet{kleijnen2009kriging} and \citet{kleijnen2017regression} review the methodology and its experimental designs, \citet{ankenman2010stochastic} extend it to stochastic simulation with heterogeneous intrinsic noise, and \citet{erickson2018comparison} compare implementations. Within this stream, the design of the experiment is treated as a central consideration. \citet{crombecq2011efficient} develop space-filling and non-collapsing sequential designs, \citet{vanbeers2008customized} customise designs for random simulation, and \citet{chen2017sequential} adapt sequential designs to stochastic kriging. Kriging-based optimisation under noise is studied by \citet{jalali2017comparison}, and \citet{kaminski2015method} considers updating fitted metamodels as new evaluations accumulate. 

The design criteria in this stream are constructive. They produce a single design, or a nested hierarchy, with good space-filling, non-collapsing or projection properties. The diagnostics introduced below are instead evaluative and directional. They take the output-specific designs as given, possibly as legacy artefacts of separate data-collection processes, and quantify the relative geometry between them. Directionality is essential because borrowing is inherently asymmetric. A sparsely observed output can gain from a densely observed neighbour without the converse holding.
{Separate fitting remains a competitive baseline in the simulation settings studied by \citet{kleijnen2014multivariate}. Recent work also uses auxiliary responses to improve stochastic simulation prediction. \citet{wang2020enhancing} jointly model a noisy response of interest with a correlated, less noisy auxiliary response and assess prediction against a single-response model.}

\subsection{Multi-output and multi-fidelity GP models}
\label{subsec:related_mogp}

Joint GP models for several outputs have a long history under the names co-kriging and the linear model of coregionalisation (LMC) in geostatistics \citep{goulard1992linear,wackernagel2013multivariate,genton2015cross} and multi-task or multi-output GP regression in machine learning \citep{bonilla2007multi,alvarez2011computationally,alvarez2012kernels,bonilla2019generic,liu2018remarks,chen2020multivariate,chen2023multivariate}. Nonseparable constructions allow output-specific smoothness \citep{gelfand2004nonstationary,majumdar2007convolved,gneiting2010matern,apanasovich2012matern,fricker2013multivariate,peruzzi2026insideout}. The multi-fidelity literature provides the closest operational research analogue to our setting: an autoregressive or co-kriging structure links a cheap low-fidelity code, evaluated on a dense design, to an expensive high-fidelity code evaluated on a sparse one \citep{kennedy2000predicting,forrester2007multi,legratiet2014recursive}, and nested designs are constructed expressly to control the cross-fidelity geometry \citep{qian2009nested}. 

{Recent multi-fidelity methods also accommodate non-nested designs and nonlinear cross-fidelity relationships. \citet{konomi2021bayesian} combine treed co-kriging with imputation for non-nested designs. \citet{heo2025active} develop a recursive non-additive emulator and active-learning strategies for choosing fidelity levels and inputs. These contributions expand the models and designs available for combining related simulations.}

\subsection{When is joint modelling worthwhile?}
\label{subsec:related_when}

In multivariate geostatistics, the limits of co-kriging under proportional cross-covariances are classical: autokrigeability implies that, with isotopic noise-free data, co-kriging reduces to kriging output by output \citep{wackernagel2013multivariate,chiles2012geostatistics}, and practical guidance has long emphasised that secondary variables help most where the primary variable is undersampled \citep{myers1982matrix,verhoef1993multivariable,cressie2015statistics}. 
{Quantitative comparisons of kriging and cokriging under specified covariance and sampling conditions include \citet{lim2022conditions}.}

{Recent work also addresses negative transfer, in which sharing information worsens prediction. \citet{li2022negative} relate this risk to the number and structure of latent functions in convolution multi-output GPs and develop regularised structures for selecting related outputs. \citet{wang2026local} use latent regularisation to identify input regions where information transfer is beneficial or detrimental. These approaches adapt the transfer model to the relationships among responses.}

{For dependence estimation}, the relevant benchmark is the Fisher information for covariance parameters \citep{coxreid1987,vandervaart1998}.
Spatial sampling design for estimating covariance parameters has been studied for single-output processes \citep{zhu2005spatial,zhu2006spatial,muller2007collecting}, where the design is chosen to make variogram or kernel parameters well identified. {\citet{li2015modelbased} study optimal multivariate sampling designs for covariance and cross-covariance estimation and for cokriging with known or estimated covariance parameters.}

{We examine how finite heterotopic designs shape borrowing opportunities and the information available for estimating dependence, using separable multi-output GPs as the principal theoretical benchmark. For common radial kernels, finite-design bounds link design geometry to oracle prediction gain and dependence information. An exact risk decomposition under correct specification then separates oracle gain from differences in estimation penalties. For a prespecified target--auxiliary pair, the operational procedure retains shared-scale independent fitting alongside response coupling and evaluates screening by its predictive effect and computational cost.}

\section{Geometry of heterotopic designs}
\label{sec:geometry}

We consider a multi-output problem over a common input domain \(\mathcal{X}\subseteq\mathbb{R}^{T}\). In the simulation context, \(\mathcal{X}\) is the space of design factors and each output is a response of interest. In the monitoring context, \(\mathcal{X}\) is geographical space and each output is a monitored quantity. For each output \(j\in\{1,\ldots,D\}\), let \(\mathcal{X}_j=\{\bmx^{(j)}_1,\ldots,\bmx^{(j)}_{N_j}\}\subseteq\mathcal{X}\) denote its \(N_j\) observed input locations, and let \(N_{\mathrm{tot}}:=\sum_{j=1}^{D}N_j\). 
\begingroup 
All geometric summaries use a specified metric and coordinate scaling. Unless stated otherwise, observation locations are distinct within each output, so \(N_j\) counts both sites and observations. For replicated data, covariance blocks are indexed by observations, retaining repeated locations; the geometric summaries use unique sites and their counts. The displayed coverage and proximity bounds use the distinct-within-output-site convention.
\endgroup 
The design is \emph{isotopic} if \(\mathcal{X}_1=\cdots=\mathcal{X}_D\), and \emph{heterotopic} otherwise \citep{wackernagel2013multivariate}.
In such problems, the value of joint modelling depends on both cross-output dependence and the geometric arrangement of the output-specific designs. A natural benchmark for that arrangement is the share of exact co-locations.

\begin{definition}[Overlap index]
\label{def:overlap_index}
For outputs \(p\) and \(q\), define the pairwise overlap index by
\begin{equation}
\omega_{pq}
:=
\frac{\lvert \mathcal{X}_p\cap\mathcal{X}_q\rvert}
{\min(N_p,N_q)}.
\label{eq:omega}
\end{equation}
A global average overlap summary is
\[
\omega
:=
\binom{D}{2}^{-1}\sum_{p<q}\omega_{pq}.
\]
\end{definition}

The overlap index is a useful benchmark, but it is intrinsically combinatorial because it records exact coincidences while ignoring nearby, non-coincident observations. 
Two design pairs can therefore share the same overlap summary yet imply different inferential regimes. We instead use diagnostics that retain distance information.

We represent heterotopy through distance-based summaries. In scattered-data approximation and kriging theory, local distance and spacing control interpolation quality \citep{stein1999interpolation,wendland2004scattered}, and the same logic applies here. 
The objects below describe cross-design coverage, directional asymmetry and local borrowing opportunity without reference to any stochastic model.

\subsection[Cross-design distances, coverage and proximity]{{Cross-design distances, coverage and proximity}}
\label{subsec:cross_design_distances}
\label{subsec:global_design_coverage}

For outputs \(p\) and \(q\), and for a point \(\bmx\in\mathcal{X}_p\), define the \emph{cross-design nearest-neighbour distance}
\begin{equation}
\delta_{p\to q}(\bmx)
:=
\min_{\bmx'\in\mathcal{X}_q}\|\bmx-\bmx'\|.
\label{eq:cross_design_distance}
\end{equation}
The collection \(\{\delta_{p\to q}(\bmx_a^{(p)}):a=1,\ldots,N_p\}\) records how the design \(\mathcal{X}_q\) covers \(\mathcal{X}_p\). Small values indicate interleaving and large values indicate separation, and the construction is directional because \(\mathcal{X}_q\) need not cover \(\mathcal{X}_p\) in the same way that \(\mathcal{X}_p\) covers \(\mathcal{X}_q\).

For later scale comparisons, it is useful to define the within-design nearest-neighbour distance and its average:
\[
\begin{aligned}
s_p(\bmx)
&:= \min_{\bmx'\in\mathcal{X}_p\setminus\{\bmx\}}\|\bmx-\bmx'\|, &
\bar{s}_p
&:= \frac{1}{N_p}\sum_{a=1}^{N_p} s_p\!\bigl(\bmx_a^{(p)}\bigr).
\end{aligned}
\]
This quantity summarises the intrinsic resolution of \(\mathcal{X}_p\) and lets us interpret cross-design gaps relative to within-design spacing.

Local cross-design distances can be aggregated to describe design complementarity at the level of full observation sets. A first such object is the directed empirical coverage function.

\begin{definition}[Directed coverage function]
\label{def:coverage_function}
For outputs \(p\) and \(q\), define
\begin{equation}
\DC_{p\to q}(r)
:=
\frac{1}{N_p}\sum_{a=1}^{N_p}
\mathds{1}\!\left\{
\delta_{p\to q}\bigl(\bmx_a^{(p)}\bigr)\le r
\right\},
\qquad r\ge 0.
\label{eq:coverage_function}
\end{equation}
\end{definition}
\(\DC_{p\to q}(r)\) is the proportion of locations in \(\mathcal{X}_p\) that lie within distance \(r\) of \(\mathcal{X}_q\). It describes how quickly output \(q\) covers output \(p\). Equivalently, it is the empirical distribution function of the cross-design nearest-neighbour distances from \(\mathcal{X}_p\) to \(\mathcal{X}_q\), a design analogue of cross nearest-neighbour summaries used in spatial point pattern analysis.

\begin{definition}[Proximity metrics]
\label{def:directed_proximity}
For outputs \(p\) and \(q\), define
\begin{enumerate}[label=(\alph*)]
\item The \emph{directed proximity} and its symmetric version:
\[
\begin{aligned}
\pi_{p\to q}
&:= \frac{1}{N_p}\sum_{a=1}^{N_p}\delta_{p\to q}\bigl(\bmx_a^{(p)}\bigr), &
\pi_{pq}
&:= \frac{1}{2}\bigl(\pi_{p\to q}+\pi_{q\to p}\bigr).
\end{aligned}
\]

\item The \emph{normalised directed proximity}, which adjusts for sampling density:
\[
\begin{aligned}
\tilde{\pi}_{p\to q}
&:= \frac{\pi_{p\to q}}{\bar{s}_q}, &
\tilde{\pi}_{pq}
&:= \frac{1}{2}\bigl(\tilde{\pi}_{p\to q}+\tilde{\pi}_{q\to p}\bigr).
\end{aligned}
\]

\item The \emph{design proximity matrix}, which collects pairwise relations for \(D\) outputs:
\begin{equation}
\bmPi
:=
[\tilde{\pi}_{p\to q}]_{1\le p,q\le D},
\qquad
\tilde{\pi}_{p\to p}:=0.
\label{eq:proximity_matrix}
\end{equation}
\end{enumerate}
\end{definition}

{For \(p\ne q\), the normalised proximity \(\tilde{\pi}_{p\to q}\) is defined only when
\(\mathcal{X}_q\) contains at least two distinct sites and \(\bar{s}_q>0\); otherwise it is undefined.}
The directed proximity \(\pi_{p\to q}\) is the average nearest-neighbour distance from \(\mathcal{X}_p\) to \(\mathcal{X}_q\), an averaged directed Hausdorff-type distance rather than a worst-case separation, and the pair \((\pi_{p\to q},\pi_{q\to p})\) captures directional asymmetry. The normalised version \(\tilde{\pi}_{p\to q}\) compares the typical cross-design gap with the internal spacing of \(\mathcal{X}_q\), giving a dimensionless measure of relative coverage in the role of a mesh-ratio normalisation from scattered-data approximation. 
{When all normalised entries are defined,} the matrix \(\bmPi\) is nonnegative and generally asymmetric and, unlike overlap, retains information about nearby, non-coincident observations.

\subsection{Output-specific borrowing geometry}
\label{subsec:output_specific_borrowing}

The preceding quantities describe the relative geometry of the observed designs at the level of full input sets. For prediction, the relevant geometry is instead local to a target output and a prediction point. Let \(j\in\{1,\ldots,D\}\) be the target output and \(\bmx_\star\in\mathcal{X}\) a prediction location, which need not belong to any observed design. In direct analogy with Eq.~\eqref{eq:cross_design_distance}, define
\[
\begin{aligned}
\delta_j(\bmx_\star)
&:= \min_{\bmx\in\mathcal{X}_j}\|\bmx_\star-\bmx\|, &
\delta_{-j}(\bmx_\star)
&:= \min_{q\neq j}\min_{\bmx\in\mathcal{X}_q}\|\bmx_\star-\bmx\|.
\end{aligned}
\]
Thus \(\delta_j(\bmx_\star)\) measures local support from the target design and \(\delta_{-j}(\bmx_\star)\) the nearest auxiliary support. When \(\delta_{-j}(\bmx_\star)<\delta_j(\bmx_\star)\), at least one auxiliary design is locally closer than the target design itself.

\begin{definition}[Borrowing potential index]
\label{def:bpi}
The \emph{local borrowing potential index} (local BPI) for output \(j\) at \(\bmx_\star\) is
\begin{equation}
b_j(\bmx_\star)
:=
\begin{cases}
\displaystyle
1-\frac{\delta_{-j}(\bmx_\star)}{\delta_j(\bmx_\star)},
& \text{if } \delta_j(\bmx_\star)>0
\text{ and } \delta_{-j}(\bmx_\star)<\delta_j(\bmx_\star),\\[8pt]
0, & \text{otherwise}.
\end{cases}
\label{eq:local_bpi}
\end{equation}
\end{definition}

By construction, \(0\le b_j(\bmx_\star)\le 1\). It is zero when the target output is locally at least as close as every auxiliary output and increases as auxiliary support becomes relatively closer.

To summarise this local opportunity over a region of interest, let \(\mathcal{X}_\star\subseteq\mathcal{X}\) be a reference set, such as an evaluation grid or a designated prediction domain.

\begin{definition}[Global borrowing potential index]
\label{def:global_bpi}
The \emph{global borrowing potential index} (global BPI) for output \(j\) over \(\mathcal{X}_\star\) is
\begin{equation}
B_j
:=
\frac{1}{|\mathcal{X}_\star|}
\sum_{\bmx_\star\in\mathcal{X}_\star}
b_j(\bmx_\star).
\label{eq:global_bpi}
\end{equation}
\end{definition}

{The global index \(B_j\) summarises the average relative nearest-neighbour advantage of auxiliary designs over the prediction region.} Together, \(\bmPi\), the coverage curves and the borrowing potential indices give a design-based description of heterotopic geometry.
{The next section connects these response-free summaries to kernel support and the residual auxiliary channel that determines oracle prediction gain.}

\section{Prediction gain under multi-output Gaussian processes}
\label{sec:oracle}

{This section studies the prediction gain from conditioning on auxiliary observations under a separable multi-output GP model with known covariance parameters.
Fixing the kernel, cross-output covariance and noise variances isolates direct response coupling from parameter-estimation effects.}

\subsection[Model and oracle predictors]{{Model and oracle predictors}}
\label{subsec:mvgp_model}
\label{subsec:oracle_predictors}

Let
$
\mathbf{f}(\bmx)
=
\bigl(f_1(\bmx),\ldots,f_D(\bmx)\bigr)^\top,\bmx\in\mathcal{X},
$
be a \(D\)-variate latent process defined on the common input domain \(\mathcal{X}\subseteq\mathbb{R}^T\). Throughout this section, we assume that
\begin{equation}
\mathbf{f}\sim \mgp\!\left(
\mathbf{0},
\ k(\bmx,\bmx'), \bmLambda
\right),
\label{eq:separable_mvgp}
\end{equation}
where \(k:\mathcal{X}\times\mathcal{X}\to\mathbb{R}\) is a positive-definite kernel and
\(\bmLambda=(\Lambda_{pq})_{1\le p,q\le D}\) is a positive-definite cross-output dependence matrix. {We adopt the separable multivariate Gaussian process construction of \citet{chen2020multivariate,chen2023multivariate} as a particular multi-output GP modelling approach.} In geostatistical terminology, \(\bmLambda\) is the coregionalisation matrix of the intrinsic coregionalisation, or separable covariance, model. Equivalently,
\begin{equation}
\operatorname{Cov}\!\bigl(f_p(\bmx),f_q(\bmx')\bigr)
=
\Lambda_{pq}\,k(\bmx,\bmx'),
\qquad
1\le p,q\le D.
\label{eq:cross_cov_separable}
\end{equation}

{We normalise \(k(x,x)=1\), so diagonal entries of \(\bmLambda\) supply signal variances. A separately free kernel amplitude would be confounded with the scale of \(\bmLambda\).}
For each output \(j\in\{1,\ldots,D\}\), observations are collected at the heterotopic design
\(\mathcal{X}_j=\{\bmx^{(j)}_1,\ldots,\bmx^{(j)}_{N_j}\}\) through
$
y_a^{(j)}
= f_j\!\bigl(\bmx_a^{(j)}\bigr)+\varepsilon_a^{(j)}, a=1,\ldots,N_j,
\varepsilon_a^{(j)}
\stackrel{\mathrm{ind}}{\sim}\mathcal{N}(0,\tau_j^2),
$
{independently across all observations, including replicates}.
Let
$
\bmy^{(j)}
=
\bigl(y_1^{(j)},\ldots,y_{N_j}^{(j)}\bigr)^\top
\in\mathbb{R}^{N_j},
\bmy
=
\bigl(
\bmy^{(1)\top},\ldots,\bmy^{(D)\top}
\bigr)^\top.
$
{Covariance matrices are indexed by observations, including noisy replicates, as specified in Section~\ref{sec:geometry}.}
For outputs \(p\) and \(q\), define the latent covariance block
$
\bmK_{pq}
:=
\left[
\Lambda_{pq}\,
k\!\bigl(\bmx_a^{(p)},\bmx_b^{(q)}\bigr)
\right]_{1\le a\le N_p,\ 1\le b\le N_q},
$
and the corresponding observation covariance block
\begin{equation}
\bmSigma_{pq}
:=
\begin{cases}
\bmK_{pp}+\tau_p^2 \bmI_{N_p}, & p=q,\\[4pt]
\bmK_{pq}, & p\neq q.
\end{cases}
\label{eq:obs_block_cov}
\end{equation}
Thus the full observation covariance matrix is the block matrix
\begin{equation}
\bmSigma
=
\cov(\bmy,\bmy)
=
\bigl[\bmSigma_{pq}\bigr]_{1\le p,q\le D}.
\label{eq:full_obs_cov}
\end{equation}

\begingroup 
We assume \(\bmSigma\succ0\), so its principal blocks and corresponding Schur complements are positive definite.
This is guaranteed by strictly positive independent observation-noise variances, or, in the noise-free case, by distinct within-output sites, a strictly positive-definite kernel on their union and \(\bmLambda\succ0\).
\endgroup 

The separable structure in Eq.~\eqref{eq:separable_mvgp} is used here as a tractable working model for the theory. It is sufficiently rich to represent cross-output borrowing through \(\bmLambda\), while remaining simple enough to reveal how the geometric quantities enter the predictive gain.
Under this Gaussian model, the joint posterior mean coincides with the {simple co-kriging (simple multivariate kriging) predictor with known mean} from multivariate geostatistics \citep{myers1982matrix,verhoef1993multivariable,cressie2015statistics,wackernagel2013multivariate}.
The equivalence is established in Proposition~\ref{prop:si_cokriging_equivalence} (Appendix~\ref{sec:cokriging}).
{An unknown constant mean or trend instead requires the ordinary- or universal-kriging equations, respectively.}
{In the simulation literature, multivariate kriging also accommodates nonseparable covariance structures \citep{kleijnen2014multivariate}. Here, the separable model serves as the principal theoretical benchmark, and diagonal \(\bmLambda\) removes cross-output covariance at fixed nuisance parameters.}

\begingroup 
\begin{remark}[Autokrigeability and the isotopic benchmark]
\label{rem:autokrigeability}
With a common fixed kernel, known zero mean, positive-definite \(\bmLambda\), distinct common observation sites and zero noise, the separable joint mean equals the marginal mean for every output and target \citep{wackernagel2013multivariate,chiles2012geostatistics}. Indeed, \((\bmLambda\otimes\bmK)^{-1}=\bmLambda^{-1}\otimes\bmK^{-1}\), so the weights for output \(j\) reduce to \(\mathbf e_j^\top\otimes\bmk_\star^\top\bmK^{-1}\). They do not depend on \(\bmLambda\). 
The oracle gain is therefore zero, and estimating \(\bmLambda\) alone leaves the predictor unchanged and incurs no squared-error prediction penalty. 
Heterotopy or observation noise can create a residual borrowing channel, but does not guarantee gain. Fitting shared nuisance parameters can change prediction and must be distinguished from dependence-only estimation.
\end{remark}
\endgroup 

Fix an output \(j\in\{1,\ldots,D\}\) and a prediction location \(\bmx_\star\in\mathcal{X}\), and consider prediction of the latent quantity \(f_j(\bmx_\star)\). Two predictors are of primary interest. The first uses only the observations from output \(j\), thereby ignoring cross-output dependence. The second uses the full multivariate observation vector and therefore exploits cross-output borrowing under the multi-output GP model. Their predictive variances are
\begin{equation*}
V_{\mathrm{ind}}^{(j)}(\bmx_\star)
:= \operatorname{Var}\!\bigl(f_j(\bmx_\star)\mid \bmy^{(j)}\bigr),
\quad
V_{\mathrm{joint}}^{(j)}(\bmx_\star)
:= \operatorname{Var}\!\bigl(f_j(\bmx_\star)\mid \bmy\bigr).
\end{equation*}

\begin{definition}[Prediction gain]
\label{def:oracle_gain}
The prediction gain from joint modelling for output \(j\) at \(\bmx_\star\) is
$
\Delta_j(\bmx_\star)
:=
V_{\mathrm{ind}}^{(j)}(\bmx_\star)
-
V_{\mathrm{joint}}^{(j)}(\bmx_\star).
$
\end{definition}

The word oracle is used in the statistical sense, meaning that the covariance kernel, the cross-output dependence matrix and the noise variances are treated as known. It does not refer to the simulation oracle, namely the black-box evaluator queried in simulation optimisation. Under this oracle Gaussian model, \(\Delta_j(\bmx_\star)\) measures the reduction in predictive uncertainty attributable to cross-output borrowing when the dependence structure is known. Although fitted metamodels are usually compared through predictive means and squared-error loss, the posterior mean is the squared-error optimal predictor and its conditional risk is exactly the posterior variance. Hence \(\Delta_j(\bmx_\star)\) is the expected squared-error reduction of the joint posterior mean relative to the independent posterior mean, while remaining a design-level quantity that does not depend on the realised response values. It is therefore the model-based counterpart of the geometric borrowing question.

For subsequent derivations, let
$
\bmy^{(-j)}
=
\bigl(
\bmy^{(1)\top},\ldots,\bmy^{(j-1)\top},
\bmy^{(j+1)\top},\ldots,\bmy^{(D)\top}
\bigr)^\top
$
denote the stacked observation vector from all outputs other than \(j\), and partition the observation covariance matrix as
\begin{equation}
\bmSigma
=
\begin{pmatrix}
\bmSigma_{jj} & \bmSigma_{j,-j}\\
\bmSigma_{-j,j} & \bmSigma_{-j,-j}
\end{pmatrix},
\label{eq:Sigma_partition}
\end{equation}
where \(\bmSigma_{jj}=\operatorname{Cov}(\bmy^{(j)},\bmy^{(j)})\), \(\bmSigma_{j,-j}=\operatorname{Cov}(\bmy^{(j)},\bmy^{(-j)})\), and similarly for the remaining blocks.

Likewise, define
\[
\begin{aligned}
\mathbf{c}_{j,j}(\bmx_\star)
&:= \operatorname{Cov}\!\bigl(f_j(\bmx_\star),\bmy^{(j)}\bigr)
= \Lambda_{jj}\bigl(k(\bmx_\star,\bmx_1^{(j)}),\ldots,k(\bmx_\star,\bmx_{N_j}^{(j)})\bigr),\\
\mathbf{c}_{j,-j}(\bmx_\star)
&:= \operatorname{Cov}\!\bigl(f_j(\bmx_\star),\bmy^{(-j)}\bigr), \qquad
\mathbf{c}_{j}(\bmx_\star)
= \bigl(\mathbf{c}_{j,j}(\bmx_\star),\,\mathbf{c}_{j,-j}(\bmx_\star)\bigr).
\end{aligned}
\]

\subsection[Prediction gain and geometric bounds]{{Prediction gain and geometric bounds}}
\label{subsec:geometry_borrowing_channel}
\label{subsec:exact_gain_bounds}
\label{subsec:design_level_implications}

The geometric quantities introduced in Section~\ref{sec:geometry} enter the prediction problem through the covariance kernel. Throughout the remainder of this section, assume a radial kernel, that is, an isotropic stationary one:
\begin{equation}
k(\bmx,\bmx')
=
\psi\!\left(\|\bmx-\bmx'\|\right),
\label{eq:radial_kernel}
\end{equation}
where \(\psi:[0,\infty)\to\mathbb{R}_+\) is non-increasing. Under Eq.~\eqref{eq:radial_kernel}, covariance strength is determined by distance alone and decreases monotonically with geometric separation.

For any target output \(j\) and prediction point \(\bmx_\star\in\mathcal{X}\), monotonicity of \(\psi\) gives
\begin{equation}
\max_{1\le a\le N_j}
k\!\bigl(\bmx_\star,\bmx_a^{(j)}\bigr)
\;=\;
\psi\!\bigl(\delta_j(\bmx_\star)\bigr),
\qquad
\max_{q\neq j}\max_{1\le a\le N_q}
k\!\bigl(\bmx_\star,\bmx_a^{(q)}\bigr)
\;=\;
\psi\!\bigl(\delta_{-j}(\bmx_\star)\bigr).
\end{equation}
\begingroup 
When \(b_j(\bmx_\star)>0\), Definition~\ref{def:bpi} gives
\(
\delta_{-j}(\bmx_\star)
=\bigl(1-b_j(\bmx_\star)\bigr)\delta_j(\bmx_\star).
\)
For a fixed positive target gap and a fixed strictly decreasing kernel \(\psi\), a larger \(b_j(\bmx_\star)\) therefore implies stronger auxiliary kernel support. Comparisons across different
target gaps require both distances.

\endgroup \begingroup The global index \(B_j\) averages these local advantages, combining their frequency and magnitude rather than measuring frequency alone. However, positive BPI is neither necessary nor sufficient for positive oracle gain: zero cross-output dependence gives zero gain even with positive BPI, whereas
noisy isotopic observations can yield denoising gain at zero BPI.
\endgroup 

The key prediction object is the auxiliary channel that remains after conditioning on the data for the target output. Define the residual auxiliary covariance matrix
\(\bmSigma_{-j,-j\mid j}:=\bmSigma_{-j,-j}-\bmSigma_{-j,j}\bmSigma_{jj}^{-1}\bmSigma_{j,-j}\)
and the corresponding residual cross-covariance vector
\(\mathbf{c}_{j,-j\mid j}(\bmx_\star):=\mathbf{c}_{j,-j}(\bmx_\star)-\mathbf{c}_{j,j}(\bmx_\star)\bmSigma_{jj}^{-1}\bmSigma_{j,-j}\).

\begin{theorem}[Prediction gain identity]
\label{thm:exact_oracle_gain}
{Suppose the latent targets and observations are jointly
Gaussian, with \(\bmSigma\succ0\).}
Then, for every output \(j\in\{1,\ldots,D\}\) and prediction
point \(\bmx_\star\in\mathcal{X}\),
\[
\Delta_j(\bmx_\star)
=
\mathbf{c}_{j,-j\mid j}(\bmx_\star)\,
\bmSigma_{-j,-j\mid j}^{-1}\,
\mathbf{c}_{j,-j\mid j}(\bmx_\star)^\top.
\]
\end{theorem}

{This identity follows from Gaussian conditioning; see also \citet[Eq.~(2)]{lim2022conditions}. It requires neither separability nor radiality.}
Its nonnegative quadratic form identifies the residual
auxiliary channel \(\mathbf{c}_{j,-j\mid j}(\bmx_\star)\)
as the mechanism underlying prediction gains:
auxiliary observations help only through information
not already supplied by the target data.
{Under the separable model,} if \(\Lambda_{jq}=0\)
for all \(q\neq j\), this channel vanishes and
\(\Delta_j(\bmx_\star)=0\) for all \(\bmx_\star\).
The next result makes the dependence on the local geometric gaps explicit. For a symmetric matrix \(\mathbf{A}\), let \(\lambda_{\max}(\mathbf{A})\) denote its largest eigenvalue.

\begin{theorem}[Local geometric upper bound]
\label{thm:local_geometric_upper_bound}
Assume Eq.~\eqref{eq:radial_kernel}. Define
\[
\begin{aligned}
C_{j,\mathrm{aux}}
&:= \left(\sum_{q\neq j}\Lambda_{jq}^2 N_q\right)^{1/2}, &
C_{j,\mathrm{proj}}
&:= |\Lambda_{jj}|\sqrt{N_j}\,\left\|\bmSigma_{jj}^{-1}\bmSigma_{j,-j}\right\|_2,
\end{aligned}
\]
where \(\|\cdot\|_2\) denotes the Euclidean norm for vectors and the associated spectral norm for matrices. Then, for every \(\bmx_\star\in\mathcal{X}\),
\begin{equation}
\Delta_j(\bmx_\star)
\le
\lambda_{\max}\!\left(
\bmSigma_{-j,-j\mid j}^{-1}
\right)
\left\{
C_{j,\mathrm{aux}}\,\psi\!\bigl(\delta_{-j}(\bmx_\star)\bigr)
+
C_{j,\mathrm{proj}}\,\psi\!\bigl(\delta_j(\bmx_\star)\bigr)
\right\}^{2}.
\label{eq:local_geometric_upper_bound}
\end{equation}
\end{theorem}

{The upper bound distinguishes raw auxiliary support, controlled by \(\delta_{-j}(\bmx_\star)\), from the projection term arising from conditioning on the target observations, involving
\(\delta_j(\bmx_\star)\). Its constants depend on the covariance parameters and the design.}
\begingroup 
\phantomsection\label{passage:bpi_bound}
For a target--auxiliary pair, a complementary BPI-based lower bound is given in Proposition~\ref{prop:si_bpi_certificate}. Consider the two-output separable model with target output \(1\), positive signal variances and independent observation noise with variances \(\tau_1^2,\tau_2^2>0\).
Assume that the common radial kernel is positive semidefinite and satisfies Eq.~\eqref{eq:radial_kernel}, with \(\psi(0)=1\). At a prediction point with target-design gap
\(d=\delta_1(\bmx_\star)>0\) and local BPI
\(b=b_1(\bmx_\star)>0\), we have
\[
\Delta_1(\bmx_\star)\ge
\frac{\Lambda_{12}^2}{\Lambda_{22}+\tau_2^2}
\left[
\psi\!\bigl((1-b)d\bigr)
-\frac{N_1\Lambda_{11}}{\tau_1^2}\psi(d)\psi(bd)
\right]_+^2,
\qquad [z]_+:=\max\{z,0\}.
\]
This conservative bound links local BPI to oracle prediction gain through the target-design gap, cross-output dependence and observation noise. In particular, a positive expression inside the brackets, together with \(\Lambda_{12}\ne 0\), is sufficient for strictly positive
oracle gain.

The proof and a leverage-based version of the bound are provided in Appendix~\ref{sec:si_prediction_results}. The residual-channel bounds of Theorem~\ref{thm:si_two_output_residual_bounds} are compared in Table~\ref{tab:si_s2_bounds}.
For fitted prediction, we combine conditional-opportunity screening with task-matched predictive validation, as described in Section~\ref{sec:screening}.
\endgroup 

The results above are pointwise. The next two propositions aggregate them to the design level through directed coverage and directed proximity.

\begin{proposition}[Directed coverage and kernel support]
\label{prop:coverage_kernel_support}
Assume Eq.~\eqref{eq:radial_kernel}. Fix outputs \(p\) and \(q\), and let \(r>0\). If
\(\delta_{p\to q}\!\bigl(\bmx_a^{(p)}\bigr)\le r\), then
\(\max_{1\le b\le N_q} k\!\bigl(\bmx_a^{(p)},\bmx_b^{(q)}\bigr)\ge \psi(r)\).
Consequently, on a fraction \(\DC_{p\to q}(r)\) of the design points of output \(p\), output \(q\) provides kernel support at least \(\psi(r)\).
\end{proposition}

The next result relates this coverage summary to the directed proximity \(\pi_{p\to q}\).

\begin{proposition}[Directed proximity controls directed coverage]
\label{prop:proximity_controls_coverage}
For every \(r>0\), \(\DC_{p\to q}(r)\ge 1-\pi_{p\to q}/r\).
\end{proposition}

Propositions~\ref{prop:coverage_kernel_support} and \ref{prop:proximity_controls_coverage} together provide a bridge from the model-free geometry of Section~\ref{sec:geometry} to borrowing under the GP model. When \(\pi_{p\to q}\) is small, a large fraction of the design points of output \(p\) lie within distance \(r\) of the design of output \(q\). By Proposition~\ref{prop:coverage_kernel_support}, those points then enjoy kernel support at least \(\psi(r)\). 
{The coverage curves and proximity summaries therefore describe the availability and directionality of cross-design kernel support.}

\begin{remark}[Interleaved and separated designs]
\label{rem:interleaved_separated}
With this interpretation, we use \emph{interleaved} to describe output-specific designs that provide close cross-design neighbours at the scale relevant for prediction or covariance learning; in the diagnostics above, this corresponds to high directed coverage and small directed proximity. We use \emph{separated} to describe designs whose cross-design nearest-neighbour distances remain large relative to that scale, corresponding to low directed coverage and large directed proximity.
\end{remark}

{These results link the geometric summaries to kernel support, while the residual auxiliary channel determines oracle prediction gain. The next section examines how cross-design geometry affects the information available for estimating cross-output dependence.}

\section{Estimability of cross-output dependence under heterotopic sampling}
\label{sec:estimation}

{This section studies how heterotopic design geometry supports estimation of cross-output dependence. The parameters of interest are the off-diagonal entries of \(\bmLambda\), which govern direct response coupling.
The diagonal entries represent marginal signal variances and are treated, together with kernel parameters and noise variances, as nuisance quantities.
This question differs from estimating a common kernel scale from several marginal likelihoods, to which
auxiliary observations can contribute even when cross-design interaction is weak.}

The first results below use a conditional information benchmark in which the kernel \(k\), the diagonal entries of \(\bmLambda\) and the noise variances \(\tau_1^2,\ldots,\tau_D^2\) are held fixed. 
This does not assert that these quantities are known in practice. Rather, it isolates the part of the likelihood that is directly controlled by cross-design geometry.
We then relate this benchmark to the fully fitted case, where uncertainty about the nuisance quantities can reduce the usable information for the off-diagonal dependence parameters. That adjustment is captured by the efficient information \citep{coxreid1987,vandervaart1998}. Write the parameter vector as
\begin{equation}
\bmtheta
=
\begin{pmatrix}
\bmeta\\
\bmnu
\end{pmatrix},
\label{eq:theta_partition}
\end{equation}
where \(\bmeta\) collects the off-diagonal entries of \(\bmLambda\) and \(\bmnu\) contains the remaining covariance parameters, including the diagonal entries of \(\bmLambda\), kernel hyperparameters and noise variances.
If the Fisher information matrix for \(\bmtheta\) is partitioned as
\[
\mathcal{I}(\bmtheta)
=
\begin{pmatrix}
\mathcal{I}_{\eta\eta} & \mathcal{I}_{\eta\nu}\\
\mathcal{I}_{\nu\eta} & \mathcal{I}_{\nu\nu}
\end{pmatrix},
\]
{assuming \(\mathcal I_{\nu\nu}\) is invertible,} the information relevant for joint estimation of \(\bmeta\) is the efficient information
\begin{equation}
\mathcal{I}_{\eta\eta\cdot \nu}
:=
\mathcal{I}_{\eta\eta}
-
\mathcal{I}_{\eta\nu}
\mathcal{I}_{\nu\nu}^{-1}
\mathcal{I}_{\nu\eta},
\label{eq:efficient_information}
\end{equation}
the Schur complement of \(\mathcal{I}_{\nu\nu}\). 

\subsection{Kernel-weighted cross-design information}
\label{subsec:fisher_information}

We first specialise the general notation introduced above to the conditional estimation problem for cross-output dependence, treating the nuisance vector \(\bmnu\) as fixed. Let
$
\mathcal{P}
:=
\{(p,q):1\le p<q\le D\},
m:=|\mathcal{P}|=\binom{D}{2},
$
and let
$
\bmeta
:=
\bigl(\Lambda_{pq}\bigr)_{(p,q)\in\mathcal{P}}
\in\mathbb{R}^{m}
$
collect the off-diagonal entries of \(\bmLambda\). Thus, within this subsection, the covariance matrix of the observation vector \(\bmy\) is viewed as a function of \(\bmeta\) alone, written \(\bmSigma(\bmeta)\).
Up to an additive constant, the resulting conditional log-likelihood is
\begin{equation}
\ell(\bmeta;\bmy)
=
-\frac{1}{2}\log\det\bmSigma(\bmeta)
-\frac{1}{2}\bmy^\top
\bmSigma(\bmeta)^{-1}
\bmy.
\label{eq:loglik_eta}
\end{equation}

For \(u=(p,q)\in\mathcal{P}\), define
\(\dot{\bmSigma}_{u}:=\partial \bmSigma(\bmeta)/\partial \Lambda_{pq}\).
Under the multi-output GP model of Section~\ref{subsec:mvgp_model}, \(\dot{\bmSigma}_{u}\) depends only on the kernel and the output-specific designs. Define the kernel block determined by the designs
\(\bmK^{(0)}_{pq}:=\bigl[k\!\bigl(\bmx^{(p)}_a,\bmx^{(q)}_b\bigr)\bigr]_{1\le a\le N_p,\ 1\le b\le N_q}\).
Then \(\dot{\bmSigma}_{u}\) is the block matrix whose \((p,q)\) block is \(\bmK^{(0)}_{pq}\), whose \((q,p)\) block is \(\bmK^{(0)\top}_{pq}\), and whose remaining blocks are zero.
For this centred Gaussian covariance model, the standard Fisher information formula gives
\[
\mathcal{I}_{uv}(\bmeta)
=
\frac{1}{2}
\tr\!\left(
\bmSigma(\bmeta)^{-1}
\dot{\bmSigma}_{u}
\bmSigma(\bmeta)^{-1}
\dot{\bmSigma}_{v}
\right),
\]
where \(\tr(\cdot)\) denotes the matrix trace. In particular, for \(u=(p,q)\),
\[
\mathcal{I}_{u,u}(\bmeta)
=
\frac{1}{2}
\left\|
\bmSigma(\bmeta)^{-1/2}
\dot{\bmSigma}_{u}
\bmSigma(\bmeta)^{-1/2}
\right\|_F^2,
\]
where \(\|\cdot\|_F\) denotes the Frobenius norm, defined by
\(
\|\mathbf{A}\|_F^2
:=
\tr(\mathbf{A}^\top \mathbf{A}).
\)
Thus the conditional estimability of \(\Lambda_{pq}\) is governed by the magnitude of the covariance perturbation induced by that parameter. 

Under the separable multi-output GP model, a perturbation in \(\Lambda_{pq}\) affects only the \((p,q)\) and \((q,p)\) cross-covariance blocks, and the corresponding derivative matrix is determined by the cross-design kernel block \(\bmK^{(0)}_{pq}\). It is therefore natural to measure the raw size of this perturbation through the squared Frobenius norm of \(\bmK^{(0)}_{pq}\).
For \(p<q\), define the kernel-weighted cross-design interaction mass
\begin{equation}
W_{pq}
:=
\left\|
\bmK^{(0)}_{pq}
\right\|_F^2
=
\sum_{a=1}^{N_p}\sum_{b=1}^{N_q}
k\!\bigl(\bmx_a^{(p)},\bmx_b^{(q)}\bigr)^2.
\label{eq:cross_design_information_mass}
\end{equation}
{When the likelihood retains individual replicates, this sum runs over observation pairs and includes their multiplicities. }
This quantity measures the total squared signal through which \(\Lambda_{pq}\) enters the covariance model. When many cross-design pairs have large kernel values, perturbations in \(\Lambda_{pq}\) produce a substantial change in the covariance structure. When most pairs have negligible kernel interaction, their imprint is weak. {Thus \(W_{pq}\) measures covariance-derivative magnitude before whitening for within-output covariance.} Exact co-location is not required. Disjoint but locally interleaved designs can still have nontrivial \(W_{pq}\), whereas geometric separation suppresses it. 

{\(W_{pq}\) is evaluated for a specified kernel and hyperparameters. Section~\ref{subsec:screen_candidates} describes how independent fits provide the covariance scenarios used by the operational screen.}

\begin{theorem}[Information scaling with cross-design interaction mass]
\label{thm:information_scaling}
Let \(u=(p,q)\in\mathcal{P}\). Then
$
\lambda_{\min}\!\bigl(\bmSigma(\bmeta)^{-1}\bigr)^2
\,W_{pq}
\;\le\;
\mathcal{I}_{u,u}(\bmeta)
\;\le\;
\lambda_{\max}\!\bigl(\bmSigma(\bmeta)^{-1}\bigr)^2
\,W_{pq},
$
where \(\lambda_{\min}(\cdot)\) and \(\lambda_{\max}(\cdot)\) denote the smallest and largest eigenvalues, respectively.
\end{theorem}

{These are finite-design bounds: deriving information rates as sample sizes increase additionally requires control of the spectral factors, which also depend on the design.}

\begin{remark}[Extension to LMC]
\label{rem:lmc_extension}
{The LMC represents multi-output covariance as a sum of separable components, each with its own input
kernel and output covariance matrix \citep{goulard1992linear}. It accommodates nonseparable structures while including separable models as special cases.} The calculation in Theorem~\ref{thm:information_scaling} applies component-wise. For a component kernel \(k_\ell\), the information for the corresponding off-diagonal coefficient is bounded as in Theorem~\ref{thm:information_scaling} with \(W_{pq}\) replaced by \(W^{(\ell)}_{pq}:=\|[k_\ell(\bmx_a^{(p)},\bmx_b^{(q)})]_{a,b}\|_F^2\); see Appendix~\ref{proof:lmc_information_scaling}. The geometric bounds of Section~\ref{subsec:geometric_information_bounds} extend in the same way, with \(\psi\) replaced by \(\psi_\ell\). The geometry relevant to estimation therefore extends to finite sums of separable components, each with its own lengthscale.\footnote{The local prediction bounds of Section~\ref{sec:oracle} remain statements for the separable model.}
{These component-wise bounds do not guarantee joint identifiability: identical component kernels yield identical covariance derivatives and singular joint Fisher information, even when the interaction masses are large.}
\end{remark}

The preceding calculation treats nuisance parameters as fixed. In joint estimation, the relevant quantity is the efficient information \(\mathcal{I}_{\eta\eta\cdot\nu}\) from Eq.~\eqref{eq:efficient_information}. The next result shows that at the independence point this distinction disappears exactly.

\begin{theorem}[Orthogonality and exact efficient information at independence]
\label{thm:orthogonality}
Consider the model of Section~\ref{subsec:mvgp_model} with parameter vector
\(\bmtheta=(\bmeta,\bmnu)\) as in Eq.~\eqref{eq:theta_partition}, where
\(\bmnu\) collects the diagonal entries of \(\bmLambda\), the noise variances and any kernel hyperparameters. 
{Assume that the off-diagonal covariance entries \(\bmeta\) vary freely in an open neighbourhood of \(\bmzero\) and that \(\mathcal I_{\nu\nu}\) is nonsingular. At the independence point \(\bmeta=\bmzero\),}
\begin{enumerate}[label=(\roman*),nosep]
\item the cross-dependence parameters are orthogonal in the information metric to the nuisance parameters:
\[
\mathcal{I}_{\eta\nu}=\bmzero,
\qquad
\mathcal{I}_{\eta\eta\cdot\nu}=\mathcal{I}_{\eta\eta};
\]
\item the cross-dependence information matrix \(\mathcal{I}_{\eta\eta}\) is diagonal. For \(u=(p,q)\in\mathcal{P}\),
\begin{equation}
\mathcal{I}_{u,u}
=
\left\|
\bmSigma_{pp}^{-1/2}\,
\bmK^{(0)}_{pq}\,
\bmSigma_{qq}^{-1/2}
\right\|_F^2,
\label{eq:whitened_interaction_mass}
\end{equation}
where \(\bmSigma_{ii}=\Lambda_{ii}\bmK^{(0)}_{ii}+\tau_i^2\bmI_{N_i}\), so the diagonal information is a whitened interaction mass;
\item if \(\tau_i^2>0\) for all outputs, then the whitened interaction mass is bounded by the raw interaction mass:
\begin{equation}
\frac{W_{pq}}
{
\left(\Lambda_{pp}\lambda_{\max}(\bmK^{(0)}_{pp})+\tau_p^2\right)
\left(\Lambda_{qq}\lambda_{\max}(\bmK^{(0)}_{qq})+\tau_q^2\right)
}
\le
\mathcal{I}_{u,u}
\le
\frac{W_{pq}}{\tau_p^2\tau_q^2}.
\label{eq:orthogonality_explicit_bounds}
\end{equation}
\end{enumerate}
\end{theorem}

{At independence, the efficient information for \(\Lambda_{pq}\) equals the whitened interaction mass in Eq.~\eqref{eq:whitened_interaction_mass}, while the raw mass \(W_{pq}\) bounds this quantity through Eq.~\eqref{eq:orthogonality_explicit_bounds}.
Away from independence, \(\mathcal{I}_{\eta\nu}\) need not vanish, so estimating nuisance parameters can reduce the information available for cross-output dependence.
The independence result therefore provides a local benchmark linking heterotopic design geometry to dependence information.}

\subsection{Interaction mass bounds and overlap insufficiency}
\label{subsec:geometric_information_bounds}

The previous subsection identifies \(W_{pq}\) as the bridge from design geometry
to dependence estimation. We now bound \(W_{pq}\) directly using the geometric
summaries. The point is not that exact overlap
is irrelevant, but that it is too coarse. The same overlap value can correspond
to very different interaction masses.

\begingroup 
\begin{proposition}[Symmetric coverage lower bound]
\label{prop:coverage_information_bound}
Assume Eq.~\eqref{eq:radial_kernel} and distinct sites within each output, so \(N_p,N_q\) are the site counts used by the coverage summaries. For outputs \(p\) and \(q\), and any \(r>0\),
\[
W_{pq}\ge\max\{N_p\DC_{p\to q}(r),N_q\DC_{q\to p}(r)\}\psi(r)^2.
\]
\end{proposition}
Under the same distinct-site convention, applying Proposition~\ref{prop:proximity_controls_coverage} in both directions gives the symmetric scalar proximity certificate
\begin{equation}
W_{pq}\ge\max\left\{
N_p\left(1-\frac{\pi_{p\to q}}r\right)_+,
N_q\left(1-\frac{\pi_{q\to p}}r\right)_+
\right\}\psi(r)^2.
\label{eq:proximity_certificate}
\end{equation}
For available normalised proximities one may substitute \(\pi_{p\to q}=\tilde\pi_{p\to q}\bar s_q\) and \(\pi_{q\to p}=\tilde\pi_{q\to p}\bar s_p\).
\endgroup 
Here \((x)_+:=\max(x,0)\). Thus favourable cross-design proximity yields a nontrivial lower bound on \(W_{pq}\) and, through Theorem~\ref{thm:information_scaling}, on the Fisher information for \(\Lambda_{pq}\). In particular, if \(\pi_{p\to q}\le \alpha r\) for some \(\alpha\in[0,1)\), then \(\DC_{p\to q}(r)\ge 1-\alpha\), so
\(W_{pq}\ge N_p(1-\alpha)\,\psi(r)^2\)
and
\(\mathcal{I}_{u,u}(\bmeta)\ge \lambda_{\min}\!\bigl(\bmSigma(\bmeta)^{-1}\bigr)^2\,N_p(1-\alpha)\,\psi(r)^2\)
for \(u=(p,q)\). Moreover, by Proposition~\ref{prop:coverage_kernel_support}, at least a proportion \(1-\alpha\) of the design points of output \(p\) then receive kernel support at least \(\psi(r)\) from output \(q\). Directed proximity therefore controls both borrowing support and dependence information.

The complementary direction is an upper bound. If every cross-design pair is
geometrically separated, the interaction mass is limited to the tail of the
correlation function.

\begin{proposition}[Upper bound under geometric separation]
\label{prop:separation_upper_bound}
Assume Eq.~\eqref{eq:radial_kernel}. If there exists \(d_0>0\) such that
\(\delta_{p\to q}(\bmx)\ge d_0\) for all \(\bmx\in\mathcal{X}_p\), then, for \(u=(p,q)\),
\[
W_{pq}\le N_pN_q\,\psi(d_0)^2,
\qquad
\mathcal{I}_{u,u}(\bmeta)
\le
\lambda_{\max}\!\bigl(\bmSigma(\bmeta)^{-1}\bigr)^2
\,N_pN_q\,\psi(d_0)^2.
\]
\end{proposition}

{The kernel tail can make these bounds small for a finite design, but a fixed positive gap does not ensure that they vanish as sample sizes increase.}
{The following result shows how zero-overlap designs can differ in their local coverage and cross-design interaction.}

\begin{theorem}[Zero-overlap dichotomy]
\label{thm:zero_overlap_dichotomy}
Assume Eq.~\eqref{eq:radial_kernel}. Consider a sequence of two-output
heterotopic designs \(\bigl(\mathcal{X}_p^{(n)},\mathcal{X}_q^{(n)}\bigr)\)
with design sizes \(N_p^{(n)}\) and \(N_q^{(n)}\), such that
\(\omega_{pq}^{(n)}=0\) for all \(n\). Then:
\begin{enumerate}[label=(\roman*),nosep]
\item \textbf{Interleaved zero-overlap.}
If \(\pi_{p\to q}^{(n)}\to 0\) and \(\pi_{q\to p}^{(n)}\to 0\), then for every fixed \(r>0\),
\[
\DC_{p\to q}^{(n)}(r)\to 1,
\qquad
\DC_{q\to p}^{(n)}(r)\to 1,
\]
and
\[
\liminf_{n\to\infty}\frac{W_{pq}^{(n)}}{N_p^{(n)}} \ge \psi(r)^2,
\qquad
\liminf_{n\to\infty}\frac{W_{pq}^{(n)}}{N_q^{(n)}} \ge \psi(r)^2.
\]
\item \textbf{Separated zero-overlap.}
If there exists \(d_0>0\) such that
\(\min_{\bmx\in\mathcal{X}_p^{(n)},\,\bmx'\in\mathcal{X}_q^{(n)}}\|\bmx-\bmx'\|\ge d_0\)
for all \(n\), then for every \(r<d_0\),
\[
\DC_{p\to q}^{(n)}(r)=0,
\qquad
\DC_{q\to p}^{(n)}(r)=0,
\]
so the lower bounds in Proposition~\ref{prop:coverage_information_bound} and Eq.~\eqref{eq:proximity_certificate} are zero for every \(r<d_0\). Moreover,
\(
W_{pq}^{(n)}
\le
N_p^{(n)}N_q^{(n)}\psi(d_0)^2.
\)
\end{enumerate}
\end{theorem}

In the interleaved regime, in the sense of Remark~\ref{rem:interleaved_separated}, zero overlap coexists with widespread local support and a lower bound on interaction mass proportional to the design size at every radius relevant to the kernel. In the separated regime, all local support below the gap \(d_0\) disappears and the full interaction mass is bounded by the kernel tail at that gap.
{Oracle prediction gain still depends on the residual auxiliary channel, but overlap alone cannot
distinguish the two regimes.}
{Experiment~I-A gives a finite numerical version of this contrast. Interleaved, partial-coverage and remote zero-overlap designs produce very different interaction masses and oracle gains.}
{The following risk comparison places parameter-estimation effects alongside oracle prediction gain.}

\section[Net value of borrowing under joint estimation]{{Net value of borrowing under joint estimation}}
\label{sec:net_benefit}
\begingroup 
The oracle analysis isolates prediction gains available when covariance parameters are known. We now examine how parameter estimation changes the comparison between joint and separate prediction through an exact squared-error risk decomposition.
\endgroup \begingroup For the risk decomposition below, \texttt{IND-SEPARATE} is a marginal GP whose kernel parameters, signal variance and noise variance are estimated using only the target-output observations.
In contrast, \texttt{IND-SHARED} has diagonal \(\bmLambda\) but estimates selected nuisance parameters using observations from several outputs. A diagonal \(\bmLambda\) therefore specifies the covariance structure, not the data used for parameter estimation.

Fix a design, a target output \(j\) and a location \(\bmx_\star\).
Let \(g_j^{\mathrm{mv}}(\bmx_\star;\bmtheta,\bmy)\) and
\(g_j^{\mathrm{ind}}(\bmx_\star;\bmvartheta_j,\bmy^{(j)})\)
be the Gaussian posterior-mean maps using the full and target-output
observations, respectively. Under correct specification, their values
at the generating joint parameter \(\bmtheta_0\) and its corresponding
marginal parameter \(\bmvartheta_{j,0}\) are
\(\mathbb E[f_j(\bmx_\star)\mid\bmy]\) and
\(\mathbb E[f_j(\bmx_\star)\mid\bmy^{(j)}]\).
The estimators \(\widehat{\bmtheta}=\widehat{\bmtheta}(\bmy)\) and
\(\widehat{\bmvartheta}_j=\widehat{\bmvartheta}_j(\bmy^{(j)})\)
use those respective information sets. 
\endgroup 
Define
\[
\begin{aligned}
R_{\mathrm{mv}}^{(j)}
&=\mathbb E\!\left[
\{f_j(\bmx_\star)-g_j^{\mathrm{mv}}(\bmx_\star;
\widehat{\bmtheta},\bmy)\}^2\right],\\
R_{\mathrm{ind}}^{(j)}
&=\mathbb E\!\left[
\{f_j(\bmx_\star)-g_j^{\mathrm{ind}}(\bmx_\star;
\widehat{\bmvartheta}_j,\bmy^{(j)})\}^2\right].
\end{aligned}
\]
\begingroup 
Expectations are over training responses and the latent target,
conditional on the fixed design. Oracle variances retain the generating
parameters, and location arguments are suppressed when unambiguous.

\begin{proposition}[Risk decomposition]
\label{prop:risk_decomp}
Under the correctly specified Gaussian model with positive-definite observation covariance, suppose the fitted joint predictor is measurable with respect to \(\sigma(\bmy)\), the fitted independent predictor is measurable with respect to \(\sigma(\bmy^{(j)})\), and both have finite second moments. Then
\[
R_{\mathrm{mv}}^{(j)}=V_{\mathrm{joint}}^{(j)}+H_{\mathrm{mv}}^{(j)},
\qquad R_{\mathrm{ind}}^{(j)}=V_{\mathrm{ind}}^{(j)}+H_{\mathrm{ind}}^{(j)},
\]
where \(H_{\mathrm{mv}}^{(j)}=\mathbb E[(g_j^{\mathrm{mv}}(\bmx_\star;\widehat{\bmtheta},\bmy)-g_j^{\mathrm{mv}}(\bmx_\star;\bmtheta_0,\bmy))^2]\), and \(H_{\mathrm{ind}}^{(j)}\) is defined analogously using only \(\bmy^{(j)}\).
\end{proposition}
The proof is given in Appendix~\ref{sec:delta_derivation}.
Proposition~\ref{prop:risk_decomp} is the multi-output form of the classical decomposition of prediction error when covariance parameters are estimated \citep{kackar1984approximations,harville1992mean,zimmerman1992mean}; \citet{abt1999estimating} and \citet{zhu2006spatial} study approximations to the penalty for spatial processes. Our interest lies in comparing two such penalties computed under different information sets, which yields Eq.~\eqref{eq:net_gain_decomp}, and in the additional cross-term that arises when an independent predictor estimates nuisance parameters from auxiliary responses (as in \texttt{IND-SHARED}; Appendix~\ref{sec:delta_derivation}).
Consequently,
\begin{equation}
R_{\mathrm{ind}}^{(j)}-R_{\mathrm{mv}}^{(j)}
=\Delta_j-(H_{\mathrm{mv}}^{(j)}-H_{\mathrm{ind}}^{(j)}).
\label{eq:net_gain_decomp}
\end{equation}
Each penalty is nonnegative, but their difference need not be. If joint fitting incurs a smaller estimation
penalty than separate fitting, estimation strengthens rather than offsets the oracle advantage.
These identities compare repeated-sampling risks, not realised losses for each dataset.
Equation~\eqref{eq:net_gain_decomp} compares joint prediction with \texttt{IND-SEPARATE}.

\endgroup \begingroup Equation~\eqref{eq:net_gain_decomp} characterises the squared-error aspect of practical value. The same
decomposition holds after averaging over a specified prediction region; comparisons across outputs require declared weights and loss normalisation.

\endgroup 
\begingroup 
\phantomsection\label{passage:output_normalisation}
Let bars denote averages over each output's prespecified prediction region, using the same fixed normalised location weights for both predictors. For prespecified dimensionless output weights \(w_j\ge0\) with \(\sum_j w_j=1\) and positive \(\overline R_{\mathrm{ind}}^{(j)}\) whenever \(w_j>0\), define
\[
G_{\mathrm{rel}}
=\sum_{j:w_j>0}w_j
\frac{\overline R_{\mathrm{ind}}^{(j)}
      -\overline R_{\mathrm{mv}}^{(j)}}
     {\overline R_{\mathrm{ind}}^{(j)}}.
\]
Each ratio uses region-averaged risks, and Eq.~\eqref{eq:net_gain_decomp} supplies its numerator after averaging. Consistent conversion of an output and its predictions multiplies both risks by the same squared conversion factor, leaving the ratio unchanged. If a positively weighted baseline risk is zero, the relative aggregate is undefined and the output-specific risks and differences are reported separately.

\endgroup 
\begingroup 
\phantomsection\label{passage:aqs_normalisation_link}
The AQS application uses the analogous aggregate of held-out response-MSE reductions with equal output weights \(1/3\) (Appendix~\ref{sec:aqs_revised}). Its matching independent baseline is \texttt{SEP} for the separable comparison and the independent two-kernel mixture for LMC. Response means and scales are estimated within training folds, and predictions return to physical units before assessment.

\endgroup 
\begingroup 
 The numerical studies assess this aspect alongside predictive uncertainty and computational expenditure for the intended task.
These dimensions are reported separately rather than combined into a single measure of whether joint modelling is worthwhile.
\endgroup \begingroup To understand the fitting penalty in Eq.~\eqref{eq:net_gain_decomp}, we next approximate each penalty through the prediction's sensitivity to parameter estimation error.

\begin{theorem}[Mean-square expansion of an estimation penalty]
\label{thm:est_penalty}
For either fitting rule, consider prediction maps \(g_n(\theta,Y_n)\) at a specified target along a sequence of fixed designs, with true parameter vector \(\theta_0\) and estimators \(\widehat\theta_n\). Define the estimation penalty by
\[
H_n:=\mathbb E\!\left[
\{g_n(\widehat\theta_n,Y_n)-g_n(\theta_0,Y_n)\}^2
\right],
\]
where \(Y_n\) contains the responses used for prediction and the expectation is over the response data used for estimation and prediction, conditional on the designs. 
Suppose \(\theta\mapsto g_n(\theta,Y_n)\) is differentiable at \(\theta_0\) almost surely. Write \(e_n=\widehat\theta_n-\theta_0\) for the parameter error and \(a_n(Y_n)=\nabla_\theta g_n(\theta_0,Y_n)\) for prediction sensitivity. Define the remainder \(r_n\) by
\[
g_n(\widehat\theta_n,Y_n)-g_n(\theta_0,Y_n)
=a_n(Y_n)^\top e_n+r_n.
\]
If, for a deterministic rate \(h_n>0\) tending to zero, \(\mathbb E[(a_n(Y_n)^\top e_n)^2]=O(h_n)\) and \(\mathbb E[r_n^2]=o(h_n)\), then
\begin{equation}
H_n=\mathbb E[(a_n(Y_n)^\top e_n)^2]+o(h_n).
\label{eq:penalty_mean_square}
\end{equation}
\end{theorem}
The expansion retains dependence between prediction sensitivity and parameter error when estimation and prediction use the same responses. Its mean-square conditions must be checked for the estimator and design sequence; they do not themselves justify an inverse-information approximation.
The proof, additional information-substitution conditions and a separate independent-calibration benchmark are provided in Appendix~\ref{sec:delta_derivation}.

For either model, the information-based working diagnostic is
\(
\widehat P=\widehat a^\top
\mathcal I(\widehat\theta)^{-1}\widehat a,
\)
where \(\widehat a=\nabla_\theta g(\widehat\theta,Y)\) is the fitted prediction gradient and \(\mathcal I\) is expected Fisher information. \endgroup \begingroup Section~\ref{subsec:e7_results} assesses its calibration and selection performance; Appendix~\ref{sec:revision_penalty_methods} gives the information conventions, observed-information alternative, refitting methods and numerical availability checks.

Allowing response coupling can also change
covariance estimates, whereas oracle gain isolates
the conditioning effect at fixed parameters.
Section~\ref{sec:screening} combines conditional-gain
screening with predictive validation of the complete
fitted predictors, retaining shared-scale estimation
as a separate candidate. 

\endgroup 
\space

\section[Screening the additional value of response borrowing]{{Screening the additional value of response borrowing}}
\label{sec:screening}
\begingroup 
Auxiliary observations can affect fitted prediction through covariance-parameter estimation and through direct
response coupling. Shared covariance learning and predictive coupling are established features of multi-task GPs \citep{bonilla2007multi}; here we use their distinction to assess the additional value of coupling while retaining shared-scale estimation as an alternative.
For a target output and one auxiliary output, we compare three fitted candidates:
\begin{enumerate}[label=(\roman*),leftmargin=*,nosep]
\item \texttt{SEP}: separate marginal GPs, each estimating its own covariance
parameters from its own observations;
\item \texttt{DIAG}: independent outputs with separate signal and noise
variances and a common lengthscale, estimated from the sum of the two marginal
log likelihoods;
\item \texttt{JOINT}: a common-scale multi-output GP whose signal variances, noise variances, common lengthscale
and cross-output correlation are estimated from the joint Gaussian likelihood.
\end{enumerate}
At fixed parameters, both independent candidates predict the target using its own responses only. \texttt{SEP} corresponds to \texttt{IND-SEPARATE}; \texttt{DIAG} is the common-scale instance of
\texttt{IND-SHARED}, using auxiliary responses to estimate the shared scale. The screen evaluates coupling under
covariance scenarios derived from \texttt{DIAG}, while validation compares the complete fitted candidates.
The procedure selects a fitting strategy for a prespecified target--auxiliary pair; selection among several auxiliary outputs is a separate problem.

\subsection{A conditional borrowing-opportunity screen}
\label{subsec:screen_candidates}
Fit \texttt{SEP} and \texttt{DIAG} using the observations available at the current training stage.
Let \(\widehat\ell\) denote the fitted \texttt{DIAG} lengthscale and let \(\mathcal A\) be the specified
response-independent prediction grid.
Screening therefore follows independent fitting and precedes coupled likelihood optimisation. The prescribed scale scenarios are working sensitivity choices around \(\widehat\ell\); these scenarios are not confidence bounds for the unknown scale.
At each scale \(s\widehat\ell\), with \(s\in\mathcal S\) for a specified finite set \(\mathcal S\subset(0,\infty)\), hold the diagonal signal and noise estimates fixed.
Let \(V_{D,s}(x)\) be the resulting target marginal latent variance, and let \(\Delta_s(x;\rho_{\max})\)
be the exact additional variance reduction from conditioning on the auxiliary observations at
\(|\rho|=\rho_{\max}\), where \(0\leq\rho_{\max}<1\) is prespecified.
Define
\begin{equation}
U_{\rm cond}=
\max_{s\in\mathcal S}
\frac{\sum_{x\in\mathcal A}\Delta_s(x;\rho_{\max})}
     {\sum_{x\in\mathcal A}V_{D,s}(x)}.
\label{eq:task_opportunity}
\end{equation}
At each scale, the ratio measures the reduction in grid-averaged oracle variance; equivalently, it weights
pointwise relative gains by their marginal variances.
Under the assumptions of Proposition~\ref{prop:task_gain_envelope}, the gain is non-decreasing in squared correlation, so evaluation at \(\rho_{\max}\) gives an envelope for \(|\rho|\leq\rho_{\max}\).
This envelope covers the specified covariance scenarios, not the risk improvement of the subsequently refitted \texttt{JOINT} model. Its use for screening fitted candidates is assessed through calibration and independent evaluation.

For available scores, exclude response coupling precisely when
\begin{equation}
U_{\rm cond}<t_U.
\label{eq:operational_gate}
\end{equation}
The threshold \(t_U\) is selected using separate calibration datasets and held fixed when applying the procedure.
It controls exclusion based on conditional oracle opportunity, rather than the relative improvement
required by predictive validation.
An excluded candidate still permits selection of \texttt{DIAG}; low conditional opportunity
does not discard shared-scale estimation. Failed fits, non-finite or invalid covariance calculations, escalated jitter, or a nearly zero variance denominator return \texttt{unresolved} and retain coupling. 
Finite boundary estimates are included.

BPI describes target-region design complementarity, while correlation
information at independence quantifies information for estimating
cross-output correlation. The latter and its normalisation are detailed
in Appendix~\ref{subsec:diagnostic_details}.
These quantities are descriptive; the calibrated exclusion rule uses
conditional borrowing opportunity.
In the operational procedure, $B_j$ has no retention cutoff: eligibility for coupling follows Eq.~\eqref{eq:operational_gate}, and final selection follows Eq.~\eqref{eq:screening_validation}.
For simulation outputs, supplied training site-mean variances enter each
candidate's observation-noise diagonal alongside its fitted nugget. 
The same conditional-gain calculation uses these diagonals, holding fitted means fixed
(Appendix~\ref{subsec:task_gain_envelope}).
For a common anisotropic kernel, the scalar
multipliers preserve the fitted coordinate ratios and
do not cover a coordinatewise scale rectangle.

\subsection{Task-matched validation and execution}
\label{subsec:screen_validation}
Fix the target, input metric, prediction region, observation-availability rule and validation folds without consulting response values. The procedure below withholds target responses at designated sites while retaining the observed auxiliary data, including measurements at held-out target sites when these would be available for the intended prediction task.

\begin{enumerate}[label=\textbf{Step \arabic*.},leftmargin=*,nosep]
\item \textbf{Fit the independent candidates and screen the full training data.}
Obtain \texttt{SEP}, \texttt{DIAG} and \(U_{\rm cond}\).
A full-data exclusion rules out final coupling, so its validation fits can
be skipped. Validation between the two independent candidates is still required.

\item \textbf{Produce held-out predictions.}
Refit the target marginal and diagonal models within every fold.
The auxiliary marginal fit is reusable because its observations remain
available in every fold. If full-data coupling remains eligible, compute a
fresh fold score using that fold's diagonal fit. Fit \texttt{JOINT} for retained
or unresolved folds; use diagonal predictions for an excluded fold.
Full-data target estimates never initialise a fold fit.

\item \textbf{Select the independent baseline and assess coupling.}
{Let \(0\leq\epsilon_{\mathrm{val}}<1\) be a
prespecified relative-improvement fraction for validation
squared error. First, select \texttt{DIAG} as the independent
baseline only if it reduces pooled validation SSE
by more than this fraction relative to \texttt{SEP};
otherwise select \texttt{SEP}.
Then select coupling only when it remains eligible on
the full training data and
\begin{equation}
\mathrm{SSE}_{\mathrm{adaptive}}
<
(1-\epsilon_{\mathrm{val}})\mathrm{SSE}_{b}.
\label{eq:screening_validation}
\end{equation}
The adaptive predictor uses the permitted joint or
diagonal prediction in each fold.}

\item \textbf{Return the selected full-data fit.}
Reuse the independent fits from Step~1. Fit the full coupled model only if
Step~3 selects it. A failed diagonal fit falls back to separate prediction;
a failed fold joint fit falls back to diagonal prediction. A failed final joint
fit returns the independently selected baseline. A failed target marginal fit
returns the training mean at the current training stage.
\end{enumerate}

For comparison, direct validation of the same three-candidate catalogue
omits both full-data and fold-specific screening. It fits
\texttt{JOINT} in every validation fold while retaining the same
selection thresholds and fitting conventions.

\subsection{Calibration, assessment and practical value}
\label{subsec:screen_assessment}
The threshold \(t_U\) is selected from a prespecified finite grid
using simulated calibration datasets independent of the evaluation data.
Calibration minimises estimated mean application time, obtained by
summing measured fitting, scoring and prediction times, subject to
empirical bounds on excess risk relative to direct validation of the
same candidate catalogue and on positive loss from excluding coupling.
The search grid, risk budgets, resampling and failure conventions
are specified in Appendix~\ref{sec:task_screen_methods}.
The selected threshold is fixed before evaluation, and total runtime
is assessed separately through paired executions.
The calibration population specifies the mean model, simulation-noise
treatment and observation allocation.
At \(t_U=0\), the procedure reduces to direct validation of all three
candidates without score computation. We recommend this setting when
no suitable calibration population is available.

Validation uses squared error on noisy target responses.
For Gaussian simulations, offline calibration and assessment use
latent squared-error risk conditional on the complete training vector,
computed under the generating law. Evaluation-time fitting and
predictor selection use only observed training data and fixed
calibration settings.
A coupled candidate is labelled materially useful when its conditional
latent risk improves on each of \texttt{SEP} and \texttt{DIAG} by more
than a prespecified relative fraction \(\epsilon_{\mathrm{ass}}\).
This differs from \(\epsilon_{\mathrm{val}}\) even when numerically
equal, because observation noise contributes to the response-risk
denominator. Candidate exclusions, final coupling errors and regret
relative to the best of the three fitted candidates are reported
separately.
\endgroup 
\begingroup 
Practical value is assessed separately through predictive accuracy,
uncertainty calibration and total computational cost.
The evaluated implementation and results are reported in
Section~\ref{subsec:revision_e6_results}.
\endgroup 
\space

\section{Numerical experiments}
\label{sec:experiments}

{We examine geometric borrowing opportunity,
covariance-estimation effects, response-coupling screening,
richer covariance families, and robustness and sensitivity
(Experiments I--V), followed by queueing and monitoring applications.
Predictive accuracy, uncertainty calibration and computational cost
are assessed separately. Complete configurations are provided in
Appendices~\ref{sec:si_experiments}--\ref{sec:si_applications},
following the same order.}

{The primary baseline \texttt{SEP} (\texttt{IND-SEPARATE})
fits each marginal GP using only that output's observations;
\texttt{DIAG} shares a fitted kernel scale without response coupling;
\texttt{JOINT} also estimates cross-output dependence.
Oracle comparisons use true marginal parameters
(\texttt{ORACLE-IND}), whereas correlation-only controls fix nuisance
parameters (\texttt{IND-FIXED}).
Richer joint models use matching independent covariance families.}

{Oracle gains are deterministic at fixed designs and covariance.
Fitted-risk reductions compare mean losses on the same datasets,
with positive values indicating improvement over the stated baseline.
Confidence intervals are reported at the 95\% level; captions specify
the loss, baseline and uncertainty unit.
Design terminology and assessment conventions are detailed in
Appendix~\ref{subsec:si_simulation}.}

\subsection[Geometric borrowing opportunities]{{Geometric borrowing opportunities}}
\label{subsec:exp_simulation}

\subsubsection{Zero-overlap geometry and overlap insufficiency}
\label{subsec:exp_zero_overlap}

\begin{figure}[htbp]
\centering
\includegraphics[width=0.32\textwidth]{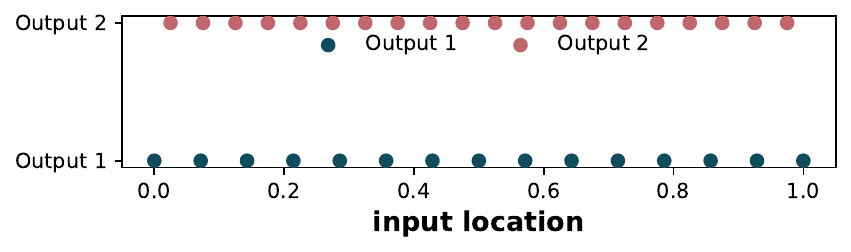}\hfill
\includegraphics[width=0.32\textwidth]{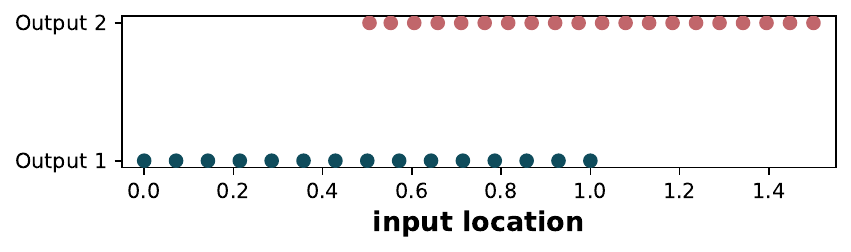}\hfill
\includegraphics[width=0.32\textwidth]{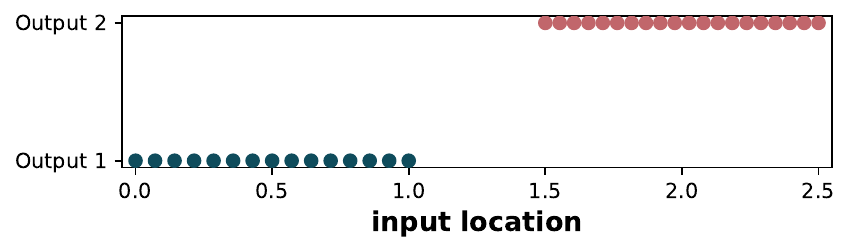}
\caption{{Zero-overlap designs for Experiment~I-A. Left: interleaved. Middle: partial coverage. Right: remote. Output~1 is fixed on \([0,1]\), output~2 varies across the three regimes, and all three designs have \(\omega_{12}=0\).}}
\label{fig:s1_designs}
\end{figure}

{The zero-overlap comparison (Experiment~I-A) holds 15 target sites fixed on \([0,1]\) and places 20 auxiliary sites between target sites (\emph{interleaved}), in a partly overlapping region (\emph{partial coverage}), or beyond a gap (\emph{remote}). With a fixed separable Mat\'ern~\(3/2\) covariance and no exact overlap (\(\omega_{12}=0\)), it tests how local coverage changes oracle variance reduction, illustrating Theorem~\ref{thm:zero_overlap_dichotomy}.}
{Interleaving has the smallest directed proximity, largest \(W_{12}\) and greatest oracle gain: 8.42\% of the independent variance (Figure~\ref{fig:s1_designs}; Table~\ref{tab:s1_summary}). Partial coverage gives intermediate values, whereas the remote design yields near-zero interaction mass and oracle gain for target output~1.
Mean latent posterior variances and relative oracle gains for both outputs, together with directed coverage curves, are provided in Appendix~\ref{sec:si_experiments}.}

\begin{table}[htbp]
\caption{Geometric and oracle summaries for Experiment~I-A, target output~1. {Oracle gain is the reduction in latent predictive variance relative to \texttt{ORACLE-IND}, averaged over the 200-point grid on \([0,1]\). The mean independent variance is \(0.0633\) in every regime. All quantities are deterministic for fixed designs and covariance parameters.}}
\label{tab:s1_summary}
\centering
\begin{tabular}{lrrrr}
\hline
Regime & $\omega$ & $\tilde{\pi}_{1\to2}$ & $W_{12}$ & Oracle gain\\
\hline
Interleaved & 0.00 & 0.27 & 57.9 & 0.0053\\
{Partial coverage} & 0.00 & 2.74 & 32.4 & 0.0028\\
{Remote} & 0.00 & 19.00 & 0.0 & 0.0000\\
\hline
\end{tabular}
\space
\end{table}

{Directionality also matters: on the same $[0,1]$ grid, oracle gain in the remote design is approximately zero for target~1 but 45.90\% for target~2. Target~2 is predicted outside its own training region, where output~1 supplies support under known dependence; learning that dependence still relies on cross-design interaction.}

\subsubsection[Residual borrowing and output-specific geometry]{{Residual borrowing and output-specific geometry}}
\label{subsec:exp_residual_screening}

{Experiment~I-B uses exact Gaussian conditioning at one prediction location and three kernel lengthscales to compare gap-filling, redundant and remote auxiliary support. The additional variance reduction after observing the target follows the residual-channel magnitude and is substantially greater for gap-filling than redundant support (Appendix~\ref{subsec:si_exp_ii}).}
{The three-output control (Experiment~I-C) places one output throughout the prediction region and the other two in restricted subregions. Within this configuration, larger output-specific BPI accompanies larger average oracle variance reduction (Appendix~\ref{subsec:si_exp_iii}).}

\Needspace{6\baselineskip}
\subsection[Parameter estimation and realised prediction gains]{{Parameter estimation and realised prediction gains}}
\label{subsec:exp_estimation}
{Matched designs (Experiment~II-A) and fitting-regime comparisons (Experiment~II-B) assess how oracle opportunity translates into fitted prediction. Penalty calibration and selection (Experiment~II-C) then test the empirical accuracy of the estimation diagnostics in Section~\ref{sec:net_benefit}.}

\subsubsection[Fitted borrowing under matched designs]{{Fitted borrowing under matched designs}}
\label{subsec:exp_information_tradeoff}

{Experiment~II-A examines how geometry affects fitted borrowing in a three-output separable GP. We hold the target design on \([0,0.50]\) and the nearby third-output design fixed, while
repositioning the 12 output-2 observation sites to form \emph{within-region}, \emph{adjacent} and \emph{remote} designs.
Joint and independent predictors of the latent target are compared under correlation-only and full covariance fitting.
Detailed configurations are provided in Appendix~\ref{subsec:si_matched_iv}.}

{Joint fitting improves mean latent risk in the within-region design at all three scales and in both fitting modes (Figure~\ref{fig:s4_tradeoff}). At \(\ell=0.12\), full fitting reduces risk by \(15.60\%\) (95\% paired interval \(11.80\%\)--\(19.23\%\)) in the within-region design and increases it by \(12.08\%\) (\(5.86\%\)--\(19.05\%\)) in the remote design. The adjacent full-fitting comparisons and remote full-fitting scale sensitivities are inconclusive. Thus local support creates an opportunity whose realised value also depends on the covariance-estimation problem.}

\begin{figure}[htbp]
\centering
\includegraphics[width=\textwidth]{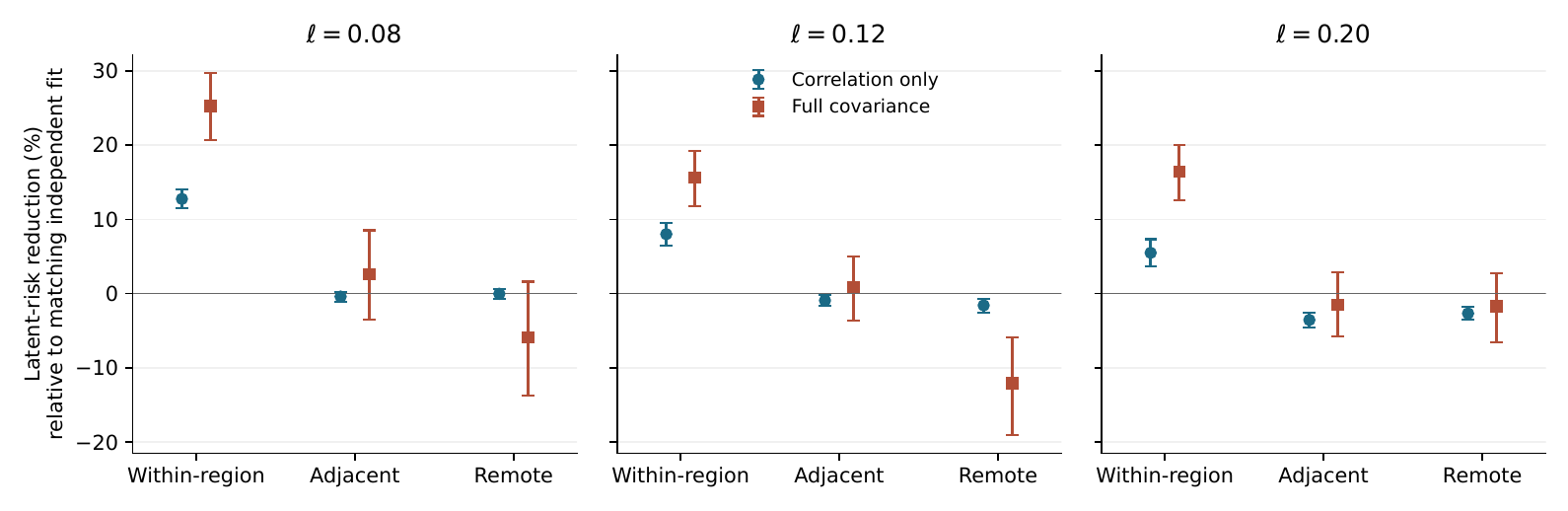}
\caption{{Matched Experiment~II-A: percentage reduction in mean conditional latent risk relative to the mode-matched independent fit, with 95\% paired dataset-bootstrap intervals. Panels use generating lengthscales \(0.08,0.12,0.20\); each geometry has 500 correlation-only and 300 full-fitting datasets. Independent baselines use generating marginal parameters in correlation-only mode and separate marginal estimates in full mode. Zero means are fixed. Positive values favour joint prediction.}}
\label{fig:s4_tradeoff}
\end{figure}

\subsubsection[Univariate versus multivariate kriging revisited]{{Univariate versus multivariate kriging revisited}}
\label{subsec:exp_kleijnen}
\begingroup 
The fitting-regime comparison (Experiment~II-B) revisits \citet{kleijnen2014multivariate}
with 16 fixed target sites and isotopic, interleaved or remote auxiliary
designs. Data follow either a separable Mat\'ern~\(3/2\) model or a two-component
LMC with nearly proportional loadings. Both generators are assessed using the
separable working family.
Full joint fitting is compared with separately fitted marginals, while paired
correlation-only fits isolate the effect of fixing nuisance parameters.

\begin{table}[htbp]\centering\small
\caption{{Experiment~II-B under full covariance fitting: conditional latent-risk reduction (\%) against \texttt{IND-SEPARATE}, with 95\% paired intervals. Each of the six cells uses 300 independent datasets and 2,000 within-cell bootstrap resamples; zero means are fixed. Positive values favour joint prediction. Paired correlation-only comparisons are reported in Appendix~\ref{subsec:e5_extension}.}}
\label{tab:s5_kriging}

\begin{tabular}{llr}
\toprule
Generator & Geometry & Risk reduction (\%) [95\% interval] \\
\midrule
Separable & Isotopic & 0.352 [-9.694, 9.679] \\
Separable & Interleaved & 10.278 [2.950, 17.280] \\
Separable & Remote & 3.585 [-5.285, 11.812] \\
LMC & Isotopic & 44.744 [38.587, 50.490] \\
LMC & Interleaved & 56.258 [52.226, 59.893] \\
LMC & Remote & -10.317 [-19.627, -1.619] \\
\bottomrule
\end{tabular}
\space
\end{table}

Under the separable generator, full joint fitting yields estimated
latent-risk reductions of \(0.35\%\) (95\% interval \(-9.69\%\) to
\(9.68\%\)) for isotopic designs and \(10.28\%\) (\(2.95\%\) to
\(17.28\%\)) for interleaved designs. Both the isotopic and remote
comparisons are inconclusive. This contrast illustrates how design
geometry can contribute to the mixed evidence motivating our comparison
with \citet{kleijnen2014multivariate}. Correlation-only fitting also
favours interleaving, with a small denoising gain under noisy isotopic
sampling and a negligible difference under the remote design.
The noisy-isotopic gain does not contradict exact autokrigeability,
which concerns the noise-free case.
For the LMC generator, full joint fitting reduces latent risk in
isotopic and interleaved designs but increases it in the remote design.
In the interleaved case, however, the gain coexists with lower
noisy-response MLPD and predictive-interval undercoverage.
Improved mean prediction therefore need not imply well-calibrated
predictive uncertainty.
Paired fitting-mode comparisons and calibration results are reported
in Appendix~\ref{subsec:e5_extension}; richer covariance families
are examined in Section~\ref{subsec:e8_robustness}.
\endgroup 

\subsubsection[Estimation penalties and refitting]
{{Estimation penalties and refitting}}
\label{subsec:e7_results}

\begingroup 

Experiment~II-C assesses expected-information, observed-information and parametric-refit penalties against realised estimation losses, and compares the resulting selection rules with predictive validation.
Correlation-only and full covariance fits use the same datasets; each parametric refit uses the same simulated responses for parameter estimation and prediction. Estimated mean penalties show both underestimation and overestimation across scenarios.

Under full covariance fitting, expected- and observed-information selection yield latent-risk reductions that are \(1.70\) and \(2.85\) percentage points smaller than those achieved by predictive validation, respectively (Table~\ref{tab:e7_selection_summary}).
The parametric-refit comparison, based on a smaller subset, is inconclusive. The information-based rules perform more favourably under correlation-only fitting, highlighting the distinction between fixed-nuisance controls and full covariance estimation.
These findings support predictive validation for operational selection, with estimated penalties serving as explanatory diagnostics.
Detailed calibration and selection results, numerical-availability frequencies and computational comparisons are reported in Appendix~\ref{sec:revision_penalty_methods}.

\begin{table}[htbp]
\centering\small
\caption[Full-fitting penalty selections versus validation]
{{Experiment~II-C under full covariance fitting:
paired difference from validation in latent-risk reduction,
in percentage points.
The difference is validation minus diagnostic mean risk,
divided by mean \texttt{IND-SEPARATE} risk on the same
dataset subset and multiplied by 100;
positive values favour the diagnostic.
Expected/observed-information comparisons use 500 independent datasets in
each of six scenarios; refits use the first 30 per scenario (180 unique datasets),
with 200 parametric draws each. Intervals use 2,000 within-scenario paired
bootstrap resamples of outer datasets.
Unavailable diagnostics use validation fallback.}}
\label{tab:e7_selection_summary}

\begin{tabular}{lrr}\toprule
Diagnostic & Unique outer datasets & Difference [95\% interval] \\ \midrule
Expected information & 3,000 & -1.70 [-2.69, -0.76] \\
Observed information & 3,000 & -2.85 [-4.09, -1.72] \\
Parametric refits & 180 (matched subset) & 1.00 [-1.97, 4.36] \\
\bottomrule\end{tabular}
\space
\end{table}

\endgroup 
\space

\subsection[Sharing parameters and screening response coupling]{{Sharing parameters and screening response coupling}}
\label{subsec:revision_e6_results}
\begingroup 
Experiment~III evaluates the screening-and-validation procedure in Section~\ref{sec:screening} using two-output separable GPs with 12 or 24 target observations and four times as many auxiliary observations. 
Isotopic, infill, partial-coverage and remote geometries are crossed with dependence and noise levels. Isotopic sampling uses four auxiliary replicates at each target site; infill interleaves a denser auxiliary grid with the target sites over the same interval.
All covariance parameters are estimated. The screen uses \(\rho_{\max}=0.995\) and \(\mathcal S=\{0.5,1,2\}\).
Screened and direct validation both compare \texttt{SEP}, \texttt{DIAG} and \texttt{JOINT}, using \(\epsilon_{\mathrm{val}}=0.02\) and four folds that withhold target responses while retaining auxiliary observations.
Independent calibration fixes the screening threshold at \(t_U=0.20\).

Screened validation reduces mean conditional latent risk by \(6.66\%\) (95\% paired interval \(5.37\%\) to \(7.86\%\)) relative to \texttt{SEP} (Table~\ref{tab:e6_core}).
Across 240 paired complete executions, mean runtime falls from 0.441 to 0.379 seconds per dataset, a reduction of
\TaskScreenSaving\% (95\% paired interval \TaskScreenSavingLow--\TaskScreenSavingHigh\%) relative to direct
validation of the same three candidates (Appendix~\ref{subsec:task_screen_results}).
The absolute saving is small for these designs, which contain 60 or 120 observations in total.

Each skipped coupled fit avoids repeated dense covariance factorisations during optimisation, each costing \(O(m^3)\) for \(m\) training observations. Here \(m=N_1+N_q\) for a full-data fit and is smaller within validation folds. Net runtime savings depend on the number of excluded fits and the costs of screening and optimisation. 
Across the evaluation datasets, screening avoids all coupled fits in remote designs but none in infill designs
(Figure~\ref{fig:e6_task_geometry}(c)).

\begin{table}[htbp]\centering\small
\caption[Task-specific borrowing and selection]{{Experiment~III:
mean conditional latent risk, percentage reduction relative to \texttt{SEP}
and final coupling frequency. Independent calibration and evaluation use
4,800 and 9,600 datasets, respectively; 95\% intervals use 2,000 paired resamples
within the 48 scenarios. Reductions are ratios of mean losses.
All covariance parameters are fitted with zero means fixed.}}
\label{tab:e6_core}
{%
\begin{tabular}{@{}lrrr@{}}
\toprule
Procedure & Mean risk & Reduction (\%) [95\% CI] & JOINT (\%) \\
\midrule
Always SEP & 0.0645 & 0.00 [0.00, 0.00] & 0.0 \\
Always DIAG & 0.0607 & 5.95 [4.27, 7.60] & 0.0 \\
Always JOINT & 0.0609 & 5.62 [3.93, 7.24] & 100.0 \\
Direct validation & 0.0602 & 6.62 [5.33, 7.82] & 33.3 \\
Screened validation & 0.0602 & 6.66 [5.37, 7.86] & 29.5 \\
\bottomrule
\end{tabular}
}
\end{table}

Gains are greatest under strong dependence and local auxiliary coverage
(Figure~\ref{fig:e6_task_geometry}); at independence, the mean-risk reduction
is inconclusive. Remote auxiliaries still improve common-scale estimation:
\texttt{DIAG} reduces risk by 4.52\% (1.01--7.88\%) despite negligible direct
borrowing (Appendix~\ref{subsec:screening_implementation}). This supports retaining shared-scale estimation when screening
response coupling.

\begin{figure}[htbp]\centering
\includegraphics[width=\linewidth]{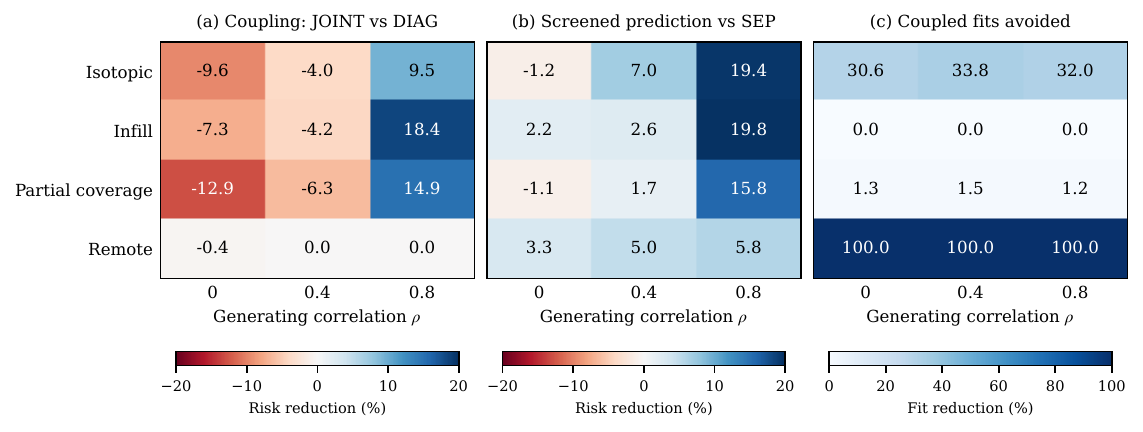}
\caption[Prediction gains and avoidable coupled fits]{{Experiment~III:
800 datasets per geometry/correlation cell, pooling target counts and noise
levels. Panels (a) and (b) show risk reductions from ratios of mean conditional
latent losses against the indicated baselines; panel (c) shows coupled fits
avoided relative to direct three-candidate validation.}}
\label{fig:e6_task_geometry}
\end{figure}

Screening falsely excludes 4.26\% (3.52--4.99\%) of joint fits that improve
latent risk by more than 2\% over both independent candidates. After validation,
final coupling has a false-positive rate (FPR) of 18.55\% (17.69--19.39\%)
and a false-negative rate (FNR) of 40.38\% (38.56--42.21\%).
Nominal 95\% response intervals attain 80.79\% coverage, compared with 80.80\%
for \texttt{SEP}; undercoverage is also present under separate fitting.
Configurations, calibration, decision errors, predictive uncertainty and
runtime comparisons are provided in Appendix~\ref{sec:task_screen_methods}.
\endgroup 
\space

\subsection[Borrowing beyond a common kernel]{{Borrowing beyond a common kernel}}
\label{subsec:e8_robustness}
\begingroup 
The richer-family study (Experiment~IV-A) compares fully fitted
two-component LMC and autoregressive (AR) models with matching independent
marginal families. Isotopic, interleaved and remote designs use 20 or 40
observations per output; validation withholds all outputs at held-out sites.
Cancelling-LMC validation reduces interleaved risk by 5.74\% (1.47--9.93\%)
and 5.42\% (3.60--7.10\%) with 20 and 40 observations per output, respectively,
despite zero same-location cross-covariance. Positive-LMC and fidelity family
averages are adverse, and remote designs give adverse or inconclusive results.

The allocation study (Experiment~IV-B) compares oracle and fully fitted
borrowing under the positive-LMC and AR covariances, crossing
8 or 16 target observations, one or four auxiliary observations per target and
interleaved or remote designs. With four auxiliaries per target, interleaved
LMC fitting reduces risk by 10.60\% (7.12--14.00\%) at $N_1=8$, $N_q=32$
and 12.32\% (6.23--18.15\%) at $N_1=16$, $N_q=64$, against matching independent
two-kernel fits; remote LMC fitting is adverse. At the interleaved allocation
$N_1=16$, $N_q=64$, LMC and AR oracle opportunities are 15.05\% and 0.90\%,
respectively, and fitted AR risk increases by 6.61\% (2.88--10.37\%).

These comparisons connect covariance-specific borrowing opportunity with
estimation effects. The covariance models, distinct validation protocols and all paired results,
including the complete 16-configuration allocation study, are provided in Appendix~\ref{sec:si_kernel_extensions}.
\endgroup 

\subsection[Robustness and sensitivity]{{Robustness and sensitivity}}
\label{subsec:exp_robustness}

\subsubsection[Working-model misspecification]{{Working-model misspecification}}
\label{subsec:e8_misspecification}
\begingroup 
Experiment~V-A examines working-model misspecification with 20 or 40 observations per output under isotopic, interleaved and remote
designs. 
The six generating families comprise positive and cancelling LMCs, AR fidelity, linear trends, varying observation noise and
Student-$t$ errors. Screened and direct validation use the same three-candidate separable catalogue. 
The screening procedure retains the threshold calibrated in Experiment~III without recalibration.

Relative to \texttt{SEP}, screened validation increases mean risk by
\(19.49\%\) (95\% paired interval \(14.86\%\) to \(24.70\%\)),
\(7.12\%\) (\(2.63\%\) to \(11.70\%\)) and
\(21.23\%\) (\(17.02\%\) to \(26.02\%\)) under positive LMC,
cancelling LMC and AR fidelity, respectively.
It reduces risk by \(5.22\%\) (\(3.81\%\) to \(6.57\%\)) under
varying noise and \(4.10\%\) (\(2.52\%\) to \(5.62\%\)) under
Student-$t$ errors; the linear-trend comparison is inconclusive.
\endgroup 

\begingroup 
Direct validation shows the same pattern, indicating limitations of the common-kernel working catalogue.
In Experiments~III and V-A, with 12--40 target observations, four-fold validation uses noisy responses and smaller training sets than full-data assessment. 
Selection can therefore favour \texttt{DIAG} or \texttt{JOINT} even when \texttt{SEP} has lower
full-data latent risk.
For Experiment~III, regret relative to the lowest-risk fitted candidate for each dataset under the generating law is reported in Appendix~\ref{subsec:task_screen_results}.
Generating models, prediction-risk comparisons and candidate-exclusion results for Experiment~V-A are provided in Appendix~\ref{sec:e8_robustness}.
\endgroup 

\subsubsection[Dimension and scale sensitivity]{{Dimension and scale sensitivity}}
\label{subsec:e9_scales}
\begingroup 
Experiment~V-B examines dimension and scale sensitivity of the screening-and-validation procedure.
It uses 20 or 50 target observations and twice as many auxiliary observations in 2, 5 and 10 dimensions, under within-region, half-region and remote designs. 
Isotropic and anisotropic (ARD) lengthscale specifications are compared across fitted covariance
scenarios.
Both scale families improve on their own SEP baselines. ARD screening gives
12.07\% (11.07--13.07\%), 11.82\% (10.91--12.72\%) and 10.88\%
(10.21--11.52\%) reductions across the three dimensions, while
isotropic screening has lower absolute risk in five and ten dimensions.
ARD false exclusion rises from 0.99\% (0.29--1.71\%) to 16.15\%
(12.26--20.13\%). Enlarging the full-data
scale scenarios changes two of 360 sampled exclusion decisions; the finite
scenario envelope remains distinct from a continuous-box guarantee.
\space

In Experiment~V-B, always selecting \texttt{JOINT} yields larger
mean risk reductions than screened validation for both lengthscale
specifications at every tested dimension
(Appendix~\ref{sec:e9_scales}).
The same pattern holds under varying observation noise and
Student-$t$ errors in Experiment~V-A
(Appendix~\ref{sec:e8_robustness}).
By contrast, under positive-LMC misspecification in Experiment~V-A,
always selecting \texttt{JOINT} increases mean risk relative to
\texttt{SEP} by \(81.61\%\) (95\% interval \(71.15\%\) to
\(93.27\%\)), compared with \(19.49\%\) (\(14.86\%\) to
\(24.70\%\)) for screened validation.
Configurations, paired intervals and coordinatewise sensitivity
checks for Experiment~V-B are provided in
Appendix~\ref{sec:e9_scales}.

\endgroup 
\space

\Needspace{5\baselineskip}
\subsection[Queueing and monitoring applications]{{Queueing and monitoring applications}}
\label{subsec:exp_applications}

\subsubsection[Tandem queue with scarce target observations]{{Tandem queue with scarce target observations}}
\label{subsec:q2_opportunity}
\begingroup 
The tandem queueing application (TQA) tests whether conditional-opportunity
screening reduces the cost of selecting a predictor while retaining the
accuracy of direct validation. \endgroup \begingroup Cheap single-station simulations supply
auxiliary observations for predicting mean total sojourn time in a two-station
queue with Poisson arrivals and lognormal services. Arrival rate and service
variability are the inputs.
Random designs combine 8, 12 or 20 target sites with 60 auxiliaries covering the
target region, its lower-load half, or a remote region separated by a positive
gap. All designs have zero exact overlap.

\endgroup \begingroup Calibration on 900 case datasets selects centred means and $t_U=0.20$ before
generating 1,800 fresh evaluation case datasets. The screening procedure compares SEP, shared-scale DIAG and separable JOINT
using four target-response folds that retain auxiliaries. Transformations,
simulation variances and covariance fits use training data only.

\begin{table}[htbp]\centering\small\setlength{\tabcolsep}{4pt}
\caption{{TQA: conditional opportunity, screening and prediction under
the calibrated centred mean. Each row contains 200 macro-datasets.
$U_{\rm cond}$ is the median full-data score. Exclude and JOINT are the
percentages with full-data coupling excluded and with JOINT finally selected,
respectively. MSE reductions compare the
screened procedure with SEP; 95\% paired intervals propagate macro-dataset and
independent-reference uncertainty.}}
\label{tab:q2_opportunity}

\begin{tabular}{@{}rlrrrr@{}}
\toprule
$N_1$ & Auxiliary region & $U_{\rm cond}$ & Exclude (\%) & JOINT (\%) & MSE reduction (\%) \\
\midrule
8 & Within-region & 0.856 & 0.0 & 91.0 & 49.61 [43.70, 55.04] \\
8 & Half-region & 0.548 & 0.0 & 62.5 & 1.17 [-3.06, 4.93] \\
8 & Remote & 0.161 & 68.5 & 5.0 & 4.98 [3.16, 6.92] \\
12 & Within-region & 0.771 & 0.0 & 82.0 & 12.12 [3.49, 20.13] \\
12 & Half-region & 0.469 & 0.0 & 43.5 & -9.14 [-13.26, -5.64] \\
12 & Remote & 0.118 & 86.0 & 0.5 & -0.15 [-0.45, 0.18] \\
20 & Within-region & 0.632 & 0.0 & 69.5 & 4.45 [-0.60, 9.30] \\
20 & Half-region & 0.374 & 1.0 & 15.5 & -2.54 [-3.72, -1.46] \\
20 & Remote & 0.075 & 97.0 & 0.0 & 0.00 [0.00, 0.00] \\
\bottomrule
\end{tabular}
\space
\end{table}

With eight targets, within-region auxiliaries give median
$U_{\rm cond}=0.856$; JOINT is selected in 91\% of datasets and the screened procedure reduces MSE by 49.61\%.
As the target count increases, median conditional opportunity declines; the within-region gain falls to 12.12\% at 12 targets and becomes inconclusive at 20 targets. Remote designs have median $U_{\rm cond}$ between 0.075 and 0.161, and the screen excludes 68.5--97.0\% of joint candidates. 
In the remote design with eight targets, the procedure selects DIAG in 77.5\% of datasets
and JOINT in 5.0\%, retaining shared-scale estimation when conditional borrowing opportunity is weak. Higher conditional opportunity retains JOINT for validation, but half-region designs increase
MSE at 12 or 20 targets even after validation.

Across all nine configurations, direct validation reduces SEP risk by 12.00\%;
screening gives 12.18\% (9.87--14.31\%).
Relative to direct validation, screening changes risk by \(-0.18\)
percentage points of SEP risk (95\% interval \(-0.43\) to
\(0.10\)) and reduces execution time by 8.61\%
(6.51--10.52\%).
\space
False exclusion is 14.92\% (95\% interval 12.48--17.44\%) among joint fits
improving MSE by more than 2\% on both independent candidates. The simulation budgets, calibration, exclusion losses, linear-mean sensitivity and complete comparisons with direct validation are reported in Appendix~\ref{subsec:q2_calibrated}.

\endgroup 
\space
\space

\subsubsection[Monitoring networks and prediction tasks]{{Monitoring networks and prediction tasks}}
\label{subsec:exp_real}
\begingroup 
The AQS study assesses cross-pollutant borrowing for missing-pollutant
completion and spatial prediction at unobserved sites or regions. It predicts
2024 annual PM$_{2.5}$, ozone and NO$_2$ concentrations in California, Colorado
and Texas. Each pollutant uses a defined annual measure and one instrument
per site, selected using quality and completeness criteria. Distances are
calculated in kilometres using a state-centred azimuthal equidistant projection.
We compare separable joint GPs (\texttt{JOINT}) with
separately fitted marginals (\texttt{SEP}), and two-component LMCs with matching independent
two-kernel mixtures. The two-component LMC allows nonseparable covariance
(Remark~\ref{rem:lmc_extension}). Each comparison controls for marginal covariance flexibility.
Task~A withholds a target pollutant while retaining observed auxiliaries at
its site; Task~B uses site-blocked cross-validation, withholding every
pollutant at each held-out site; Task~C withholds every pollutant in a
contiguous region. Instruments and folds are chosen without concentrations;
response transformations and covariance parameters are estimated within
training folds.

\begin{table}[htbp]
\centering\small\setlength{\tabcolsep}{4pt}
\caption[Matched borrowing comparisons across AQS tasks and states]{{
AQS with primary instruments and zero spatial buffer.
Separable and LMC columns give equal-output MSE reductions (\%)
against their matching independent covariance families.
The 95\% intervals use 2,000 paired spatial-block resamples after
averaging site losses over ten partitions; they condition on fitted
predictions and omit refitting uncertainty.}}
\label{tab:aqs_summary}

\begin{tabular}{@{}llrr@{}}
\toprule
State & Task & Separable & LMC \\
\midrule
CA & A & 5.48 [-1.92, 12.65] & 13.58 [10.00, 21.39] \\
CA & B & 0.97 [-3.44, 5.09] & 6.29 [1.66, 13.95] \\
CA & C & -4.29 [-15.23, 4.09] & 7.34 [0.86, 16.89] \\
CO & A & -3.93 [-21.61, 15.60] & -0.13 [-13.48, 8.36] \\
CO & B & -2.04 [-6.84, 0.76] & 3.10 [-6.06, 15.78] \\
CO & C & 0.09 [-0.74, 0.52] & 0.70 [-0.03, 3.30] \\
TX & A & -8.93 [-31.99, 0.52] & 7.43 [-0.92, 18.72] \\
TX & B & -8.79 [-32.57, 1.36] & 3.93 [-1.82, 14.43] \\
TX & C & -90.79 [-120.57, -37.29] & -23.35 [-53.93, 7.39] \\
\bottomrule
\end{tabular}
\space
\end{table}

In California, LMC reduces equal-output MSE by \(13.58\%\), \(6.29\%\) and \(7.34\%\) in Tasks~A, B and C, respectively, with all paired 95\% intervals above zero (Table~\ref{tab:aqs_summary}).
These gains persist when all pollutants at the prediction site or within the held-out region are unavailable.
California's aggregate separable comparisons and the aggregate LMC comparisons for Colorado and Texas are inconclusive.
In Texas, regional prediction under the separable model increases equal-output MSE by \(90.79\%\); the LMC comparison for NO$_2$ is also adverse.

Comparisons with the training mean, nearest neighbour, inverse-distance weighting and a fitted spatial plane show that simple baselines remain competitive for regional prediction.
In Task~C, the spatial plane attains lower RMSE than all three separable GP candidates for PM$_{2.5}$ in all three states and for all three pollutants in California.
LMC has lower RMSE than the plane for all three California pollutants, whereas nearest neighbour has lower RMSE than every GP candidate, including LMC, for Colorado NO$_2$ in the same task.
Appendix~\ref{sec:aqs_revised} provides full methodological details,
pollutant-specific and baseline comparisons, and predictive calibration,
and summarises separable-family instrument-definition and Task C buffer sensitivities.

\endgroup 
\space

\section{Conclusions}
\label{sec:conclusions}

\begingroup 
This paper examines how observation geometry and covariance estimation
jointly shape the value of multivariate kriging under heterotopic
sampling. Directed coverage, directed proximity and borrowing potential
indices reveal differences in local auxiliary support that exact overlap
misses. For separable Gaussian models with common radial kernels,
finite-design bounds link this geometry to oracle prediction gains and
to information for estimating cross-output dependence. Under correct
specification, an exact risk decomposition building on classical
prediction-error results shows how differing estimation penalties alter
the advantage over separate fitting; a mean-square expansion
approximates these penalties under explicit conditions. Performance
under richer or misspecified models is assessed empirically.
\endgroup 

\begingroup 
Controlled experiments show that local auxiliary support can improve
fitted prediction without exact overlap and help explain the mixed
evidence on multivariate kriging: in our revisit of
\citet{kleijnen2014multivariate}, separable joint fitting gives an
inconclusive comparison at shared sites but an estimated \(10.28\%\)
risk reduction at interleaved sites. The size and direction of the
benefit depend on the covariance family and observation allocation.
\endgroup \begingroup 
Under full covariance fitting, predictive validation outperforms
information-based selection, so estimated penalties serve as
diagnostics rather than selection rules.
\endgroup \begingroup 
Under misspecification, however, both screened and direct validation
can select models that perform worse than separate fitting.
\endgroup 

\begingroup 
Auxiliary data help through parameter sharing, using several outputs
to estimate common covariance parameters, and response coupling, where
auxiliary responses enter the target prediction directly. Geometry,
with dependence and noise, shapes the opportunity for coupling, whereas
remote observations can still aid shared-scale estimation. The
screening procedure therefore retains a shared-scale independent model
while assessing coupling. With sparse targets and plentiful
auxiliaries, the complete procedure reduces mean latent risk by
\(6.66\%\) relative to separate fitting (\(5.95\%\) for always
selecting the shared-scale model), and screening cuts runtime by
\TaskScreenSaving\% relative to direct validation with similar risk.
In practice, analysts can fit the separate and shared-scale models
first, screen coupling where calibration is available and select by
task-matched validation.
\endgroup 

\begingroup 
The applications show the importance of observation availability at
prediction time. In the tandem queue, screening reduces execution time
and retains gains when targets are scarce and auxiliaries cover the
target region, but restricted coverage causes some losses.
California's two-component linear model of coregionalisation improves
on matching independent models even when all pollutants at a site or
region are withheld. Colorado and Texas comparisons are inconclusive or
partly adverse, and simple spatial baselines remain competitive for
some regional predictions.
\endgroup 

\begingroup 
\label{passage:limitations}
{Practical value requires separate assessments of accuracy, calibration
and computational cost, as several fitted models improve mean
prediction while their intervals undercover.} Screening can exclude
useful joint fits at rates that vary across allocations, dimensions and
mean specifications, and the diagnostics depend on the input metric and
kernel scales. Mean and covariance specification, including possible
trend and covariance confounding, should be assessed together. Future
work should incorporate covariance-estimation uncertainty in predictive
intervals, extend the geometric analysis to nonstationary
cross-covariance and develop sequential designs that weigh prediction
gains against observation costs in choosing which output to sample and
where.
{Selecting among several auxiliary outputs, and comparing a screened subset
with fully joint modelling of all outputs, requires a catalogue-level loss
and cost objective and remains for future work.}
\endgroup 

\section*{Declaration of generative AI and AI-assisted technologies in the manuscript preparation process}
\begingroup 
During the preparation of this work the authors used Codex (OpenAI; GPT-5.5, GPT-5.6 Sol and GPT-6 Astra) and Claude Code (Anthropic; Claude Opus 4.8, Claude Opus 5 and Claude Opus 5.5) in order to assist with code development and debugging, experiment implementation and execution, brainstorming of experimental directions, consolidation of proof arguments, matching of theoretical results with experimental evidence, reference checking and manuscript preparation. After using these tools, the authors reviewed and edited the content as needed and take full responsibility for the content of the published article. All proofs in the main paper and the supplement were independently checked by the authors without AI assistance.
\endgroup

\clearpage
\section*{Appendices}
\addcontentsline{toc}{section}{Appendices}
\appendix
\renewcommand*{\thesection}{\Alph{section}}
\renewcommand*{\thesubsection}{\thesection.\arabic{subsection}}
\renewcommand*{\theHsection}{appendix.\Alph{section}}
\makeatletter
\renewcommand*{\p@subsection}{}
\makeatother
\numberwithin{equation}{section}
\numberwithin{figure}{section}
\numberwithin{table}{section}
\captionsetup[figure]{name=Fig.}

\section{Model Equivalence}
\label{sec:cokriging}

\begin{proposition}
\label{prop:si_cokriging_equivalence}
{For known zero mean and positive-definite observation covariance,} under the separable {multi-output} GP with known covariance kernel, \(\bmLambda\) and noise variances, and a Gaussian likelihood, the posterior mean \(\widehat{f}_{\mathrm{mv}}^{(j)}(\bmx_\star)\) is the best linear unbiased predictor of \(f_j(\bmx_\star)\) given all observations.
\end{proposition}

\begin{proof}
\leavevmode\phantomsection\label{proof:cokriging_equivalence}%
Under the Gaussian model, the posterior mean equals the conditional expectation \(\mathbb{E}[f_j(\bmx_\star)\mid \bmy]\), which is linear in \(\bmy\) and minimises mean squared prediction error among all measurable predictors. It is therefore also the best linear unbiased predictor. To make the equivalence explicit, write \(\bmSigma:=\operatorname{Cov}(\bmy)\) for the full observation covariance and \(\mathbf{c}_{j}(\bmx_\star):=\operatorname{Cov}\!\bigl(f_j(\bmx_\star),\bmy\bigr)\) for the cross-covariance between the target and all observations. {Because \(\bmSigma\) is positive definite by hypothesis}, the posterior mean is
\[
\widehat{f}_{\mathrm{mv}}^{(j)}(\bmx_\star)
=
\mathbf{c}_{j}(\bmx_\star)\,\bmSigma^{-1}\bmy .
\]
The classical co-kriging equations determine weights \(\mathbf{w}\in\mathbb{R}^{N_{\mathrm{tot}}}\) by solving the same block linear system \(\bmSigma\,\mathbf{w}=\mathbf{c}_{j}(\bmx_\star)^\top\) \citep[see][Ch.~3]{wackernagel2013multivariate}, so the co-kriging predictor \(\mathbf{w}^\top\bmy\) coincides with \(\widehat{f}_{\mathrm{mv}}^{(j)}(\bmx_\star)\).
\end{proof}

\section{Proofs and supporting results}
\label{sec:aux_proofs}

\begingroup 
\paragraph*{Statement-to-proof index}
The following index locates the proofs and supporting arguments.
\begin{longtable}{@{}p{0.53\linewidth}p{0.40\linewidth}@{}}
\toprule
Result & Proof or discussion \\
\midrule
\endhead
Main Theorem~\ref{thm:exact_oracle_gain} & Appendix~\ref{proof:exact_oracle_gain}, p.~\pageref{proof:exact_oracle_gain} \\
Main Theorem~\ref{thm:local_geometric_upper_bound} & Appendix~\ref{proof:local_geometric_upper_bound}, p.~\pageref{proof:local_geometric_upper_bound} \\
Main Proposition~\ref{prop:coverage_kernel_support} & Appendix~\ref{proof:coverage_kernel_support}, p.~\pageref{proof:coverage_kernel_support} \\
Main Proposition~\ref{prop:proximity_controls_coverage} & Appendix~\ref{proof:proximity_controls_coverage}, p.~\pageref{proof:proximity_controls_coverage} \\
Main Theorem~\ref{thm:information_scaling} & Appendix~\ref{proof:information_scaling}, p.~\pageref{proof:information_scaling} \\
Main Remark~\ref{rem:lmc_extension} & Appendix~\ref{proof:lmc_information_scaling}, p.~\pageref{proof:lmc_information_scaling} \\
Main Theorem~\ref{thm:orthogonality} & Appendix~\ref{proof:orthogonality}, p.~\pageref{proof:orthogonality} \\
Main Proposition~\ref{prop:coverage_information_bound} & Appendix~\ref{proof:coverage_information_bound}, p.~\pageref{proof:coverage_information_bound} \\
Main Proposition~\ref{prop:separation_upper_bound} & Appendix~\ref{proof:separation_upper_bound}, p.~\pageref{proof:separation_upper_bound} \\
Main Theorem~\ref{thm:zero_overlap_dichotomy} & Appendix~\ref{proof:zero_overlap_dichotomy}, p.~\pageref{proof:zero_overlap_dichotomy} \\
Main Proposition~\ref{prop:risk_decomp} & Appendix~\ref{proof:risk_decomp}, p.~\pageref{proof:risk_decomp} \\
Main Theorem~\ref{thm:est_penalty} & Appendix~\ref{proof:est_penalty}, p.~\pageref{proof:est_penalty} \\
Eq.~\eqref{eq:penalty_block_identity} & Appendix~\ref{proof:penalty_block_identity}, p.~\pageref{proof:penalty_block_identity} \\
Remark~\ref{rem:anisotropic_scales} & Discussion, p.~\pageref{rem:anisotropic_scales} \\
Proposition~\ref{prop:anisotropic_mass_corners} & Appendix~\ref{proof:anisotropic_mass_corners}, p.~\pageref{proof:anisotropic_mass_corners} \\
Proposition~\ref{prop:si_cokriging_equivalence} & Appendix~\ref{proof:cokriging_equivalence}, p.~\pageref{proof:cokriging_equivalence} \\
Proposition~\ref{prop:independent_calibration_penalty} & Appendix~\ref{proof:independent_calibration_penalty}, p.~\pageref{proof:independent_calibration_penalty} \\
Theorem~\ref{thm:si_two_output_residual_bounds} & Appendix~\ref{proof:two_output_residual_bounds}, p.~\pageref{proof:two_output_residual_bounds} \\
Proposition~\ref{prop:si_bpi_certificate} & Appendix~\ref{proof:bpi_certificate}, p.~\pageref{proof:bpi_certificate} \\
Proposition~\ref{prop:task_gain_envelope} & Appendix~\ref{proof:task_gain_envelope}, p.~\pageref{proof:task_gain_envelope} \\
Information ceiling (unnumbered) & Appendix~\ref{proof:information_ceiling}, p.~\pageref{proof:information_ceiling} \\
\bottomrule
\end{longtable}
\endgroup 
\space

\subsection{Prediction-gain results}

\begin{proof}[Proof of Theorem~\ref{thm:exact_oracle_gain}]
\leavevmode\phantomsection\label{proof:exact_oracle_gain}%
Apply Gaussian conditioning to the joint distribution of \(f_j(\bmx_\star)\), \(\bmy^{(j)}\), and \(\bmy^{(-j)}\). Conditioning only on \(\bmy^{(j)}\) gives
\[
V_{\mathrm{ind}}^{(j)}(\bmx_\star)
=
{\operatorname{Var}(f_j(\bmx_\star))}
-
\mathbf{c}_{j,j}(\bmx_\star)\,
\bmSigma_{jj}^{-1}\,
\mathbf{c}_{j,j}(\bmx_\star)^\top.
\]
Conditioning on the full observation vector gives
\[
V_{\mathrm{joint}}^{(j)}(\bmx_\star)
=
V_{\mathrm{ind}}^{(j)}(\bmx_\star)
-
\mathbf{c}_{j,-j\mid j}(\bmx_\star)\,
\bmSigma_{-j,-j\mid j}^{-1}\,
\mathbf{c}_{j,-j\mid j}(\bmx_\star)^\top.
\]
Subtracting the two variances yields the identity in
Theorem~\ref{thm:exact_oracle_gain}. Since
\(\bmSigma_{-j,-j\mid j}\) is positive definite whenever the full observation
covariance is nonsingular, the quadratic form is nonnegative.
\end{proof}

\begin{proof}[Proof of Theorem~\ref{thm:local_geometric_upper_bound}]
By \phantomsection\label{proof:local_geometric_upper_bound}Theorem~\ref{thm:exact_oracle_gain},
\[
\Delta_j(\bmx_\star)
=
\mathbf{c}_{j,-j\mid j}(\bmx_\star)\,
\bmSigma_{-j,-j\mid j}^{-1}\,
\mathbf{c}_{j,-j\mid j}(\bmx_\star)^\top
\le
\lambda_{\max}\!\bigl(\bmSigma_{-j,-j\mid j}^{-1}\bigr)
\left\|
\mathbf{c}_{j,-j\mid j}(\bmx_\star)
\right\|_2^2.
\]
Using the definition of \(\mathbf{c}_{j,-j\mid j}(\bmx_\star)\) from the
main text and the triangle inequality,
\[
\left\|
\mathbf{c}_{j,-j\mid j}(\bmx_\star)
\right\|_2
\le
\left\|
\mathbf{c}_{j,-j}(\bmx_\star)
\right\|_2
+
\left\|
\mathbf{c}_{j,j}(\bmx_\star)\,
\bmSigma_{jj}^{-1}\bmSigma_{j,-j}
\right\|_2.
\]
For the first term, each entry of \(\mathbf{c}_{j,-j}(\bmx_\star)\) equals \(\Lambda_{jq}\psi(\|\bmx_\star-\bmx_a^{(q)}\|)\), and \(\|\bmx_\star-\bmx_a^{(q)}\|\ge \delta_{-j}(\bmx_\star)\). Because \(\psi\) is non-increasing,
\[
\left\|
\mathbf{c}_{j,-j}(\bmx_\star)
\right\|_2
\le
\left(
\sum_{q\neq j}\Lambda_{jq}^2N_q
\right)^{1/2}
\psi\!\bigl(\delta_{-j}(\bmx_\star)\bigr)
=
C_{j,\mathrm{aux}}\psi\!\bigl(\delta_{-j}(\bmx_\star)\bigr).
\]
For the second term, each entry of \(\mathbf{c}_{j,j}(\bmx_\star)\) is bounded by \(|\Lambda_{jj}|\psi(\delta_j(\bmx_\star))\), so
\[
\left\|
\mathbf{c}_{j,j}(\bmx_\star)
\right\|_2
\le
|\Lambda_{jj}|\sqrt{N_j}\,\psi\!\bigl(\delta_j(\bmx_\star)\bigr).
\]
Hence
\[
\left\|
\mathbf{c}_{j,j}(\bmx_\star)\,
\bmSigma_{jj}^{-1}\bmSigma_{j,-j}
\right\|_2
\le
|\Lambda_{jj}|\sqrt{N_j}
\left\|
\bmSigma_{jj}^{-1}\bmSigma_{j,-j}
\right\|_2
\psi\!\bigl(\delta_j(\bmx_\star)\bigr)
=
C_{j,\mathrm{proj}}\psi\!\bigl(\delta_j(\bmx_\star)\bigr).
\]
Combining the bounds gives Eq.~\eqref{eq:local_geometric_upper_bound}.
\end{proof}

\subsection{Design-geometry and estimability results}
\label{sec:aux_design_estimation_risk_proofs}

\begin{proof}[Proof of Proposition~\ref{prop:coverage_kernel_support}]
\leavevmode\phantomsection\label{proof:coverage_kernel_support}%
If \(\delta_{p\to q}(\bmx_a^{(p)})\le r\), then there exists \(\bmx_b^{(q)}\in\mathcal{X}_q\) such that \(\|\bmx_a^{(p)}-\bmx_b^{(q)}\|\le r\). Since \(\psi\) is non-increasing,
\[
k\!\bigl(\bmx_a^{(p)},\bmx_b^{(q)}\bigr)
=
\psi\!\bigl(\|\bmx_a^{(p)}-\bmx_b^{(q)}\|\bigr)
\ge
\psi(r).
\]
This proves Proposition~\ref{prop:coverage_kernel_support}. The
final statement follows from the definition of \(\DC_{p\to q}(r)\).
\end{proof}

\begin{proof}[Proof of Proposition~\ref{prop:proximity_controls_coverage}]
\leavevmode\phantomsection\label{proof:proximity_controls_coverage}%
Let
\[
Z_a:=\delta_{p\to q}\!\bigl(\bmx_a^{(p)}\bigr),
\qquad
a=1,\ldots,N_p.
\]
Then \(Z_a\ge 0\) and
\[
\pi_{p\to q}
=
\frac{1}{N_p}\sum_{a=1}^{N_p} Z_a.
\]
Applying Markov's inequality to the empirical distribution of \(\{Z_a\}_{a=1}^{N_p}\) gives
\[
\frac{1}{N_p}\sum_{a=1}^{N_p}\mathbf{1}\{Z_a>r\}
\le
\frac{\pi_{p\to q}}{r}.
\]
Since
\[
\DC_{p\to q}(r)
=
1-\frac{1}{N_p}\sum_{a=1}^{N_p}\mathbf{1}\{Z_a>r\},
\]
the result follows.
\end{proof}

\begin{proof}[Proof of Theorem~\ref{thm:information_scaling}]
\leavevmode\phantomsection\label{proof:information_scaling}%
Fix \(u=(p,q)\). The standard covariance information formula for a centred Gaussian model gives
\[
\mathcal{I}_{u,u}(\bmeta)
=
\frac{1}{2}
\left\|
\bmSigma^{-1/2}
\dot{\bmSigma}_{u}
\bmSigma^{-1/2}
\right\|_F^2.
\]
Since \(\dot{\bmSigma}_{u}\) has only the \((p,q)\) and \((q,p)\) blocks nonzero, with values \(\bmK^{(0)}_{pq}\) and \(\bmK_{pq}^{(0)\top}\),
\[
\left\|
\dot{\bmSigma}_{u}
\right\|_F^2
=
2\left\|
\bmK^{(0)}_{pq}
\right\|_F^2
=
2W_{pq}.
\]
For any positive-definite \(\mathbf{A}\) and symmetric \(\mathbf{B}\),
\[
\lambda_{\min}(\mathbf{A})^2\|\mathbf{B}\|_F^2
\le
\|\mathbf{A}^{1/2}\mathbf{B}\mathbf{A}^{1/2}\|_F^2
\le
\lambda_{\max}(\mathbf{A})^2\|\mathbf{B}\|_F^2.
\]
Apply this with \(\mathbf{A}=\bmSigma^{-1}\) and
\(\mathbf{B}=\dot{\bmSigma}_{u}\), then multiply by \(1/2\), to obtain
the conclusion of Theorem~\ref{thm:information_scaling}.
\end{proof}

\paragraph*{Component-wise LMC extension}
\phantomsection\label{proof:lmc_information_scaling}
For the LMC, let \(\mathbf B_\ell\) be the output covariance matrix of component \(\ell\), and fix \(u=(p,q)\), \(p<q\). Differentiating the full observation covariance matrix with respect to \([\mathbf{B}_\ell]_{pq}\) gives a symmetric derivative matrix \(\dot{\bmSigma}_{\ell,u}\) whose only nonzero blocks are \((p,q)\) and \((q,p)\), equal to \(\bmK^{(\ell)}_{pq}\) and \(\bmK^{(\ell)\top}_{pq}\). Therefore
\[
\left\|
\dot{\bmSigma}_{\ell,u}
\right\|_F^2
=
2
\left\|
\bmK^{(\ell)}_{pq}
\right\|_F^2
=
2W^{(\ell)}_{pq}.
\]
The standard covariance information formula for a centred Gaussian model gives
\[
\mathcal{I}_{(\ell,p,q),(\ell,p,q)}
=
\frac{1}{2}
\left\|
\bmSigma^{-1/2}
\dot{\bmSigma}_{\ell,u}
\bmSigma^{-1/2}
\right\|_F^2.
\]
Applying the same spectral inequality used in the proof of Theorem~\ref{thm:information_scaling}, with \(\mathbf{A}=\bmSigma^{-1}\) and \(\mathbf{B}=\dot{\bmSigma}_{\ell,u}\), gives the two-sided bound stated in Remark~\ref{rem:lmc_extension}.

If \(k_\ell(\bmx,\bmx')=\psi_\ell(\|\bmx-\bmx'\|)\), then
\[
W^{(\ell)}_{pq}
=
\sum_{a=1}^{N_p}
\sum_{b=1}^{N_q}
\psi_\ell\!\left(
\left\|\bmx_a^{(p)}-\bmx_b^{(q)}\right\|
\right)^2.
\]
The proof of Proposition~\ref{prop:coverage_information_bound}, the subsequent bounds based on proximity, Theorem~\ref{thm:zero_overlap_dichotomy} and Proposition~\ref{prop:separation_upper_bound} use only this radial expression and monotonicity. Replacing \(W_{pq}\) and \(\psi\) by \(W^{(\ell)}_{pq}\) and \(\psi_\ell\) therefore gives the component-wise geometric conclusions claimed in Remark~\ref{rem:lmc_extension}.

\begin{proof}[Proof of Theorem~\ref{thm:orthogonality}]
\leavevmode\phantomsection\label{proof:orthogonality}%
At an independence point \(\bmeta=\bmzero\), the observation covariance matrix is block diagonal,
\[
\bmSigma
=
\operatorname{blockdiag}(\bmSigma_{11},\ldots,\bmSigma_{DD}),
\qquad
\bmSigma_{ii}
=
\Lambda_{ii}\bmK^{(0)}_{ii}+\tau_i^2\bmI_{N_i}.
\]
For \(u=(p,q)\in\mathcal{P}\), the derivative
\(\dot{\bmSigma}_{u}=\partial\bmSigma/\partial\Lambda_{pq}\) has only two nonzero off-diagonal blocks, \((p,q)\) and \((q,p)\), equal to \(\bmK^{(0)}_{pq}\) and \(\bmK^{(0)\top}_{pq}\). By contrast, at \(\bmeta=\bmzero\), the derivative with respect to any nuisance parameter in \(\bmnu\) is block diagonal: diagonal entries of \(\bmLambda\), noise variances and kernel hyperparameters alter only the within-output covariance blocks. Because each within-output block \(\bmSigma_{ii}=\Lambda_{ii}\bmK^{(0)}_{ii}+\tau_i^2\bmI_{N_i}\) {is positive definite as a principal block of \(\bmSigma\)}, the inverse \(\bmSigma^{-1}\) is also block diagonal, so the product
\[
\bmSigma^{-1}\dot{\bmSigma}_{u}\bmSigma^{-1}\dot{\bmSigma}_{\nu}
\]
is off-diagonal for every nuisance derivative \(\dot{\bmSigma}_{\nu}\), and hence has trace zero. The standard covariance information formula for a centred Gaussian model therefore gives \(\mathcal{I}_{\eta\nu}=\bmzero\), so the Schur complement in Eq.~\eqref{eq:efficient_information} reduces to \(\mathcal{I}_{\eta\eta}\).

The same block calculation shows that distinct cross-output parameters are orthogonal at independence. If \(u=(p,q)\) and \(v=(r,s)\) are distinct, then
\(\bmSigma^{-1}\dot{\bmSigma}_{u}\bmSigma^{-1}\dot{\bmSigma}_{v}\) has zero diagonal blocks, regardless of whether the pairs share an output or are disjoint. Hence \(\mathcal{I}_{u,v}=0\). For \(u=(p,q)\), the same information formula gives
\[
\begin{aligned}
\mathcal{I}_{u,u}
&=
\frac{1}{2}
\tr\!\left(
\bmSigma^{-1}
\dot{\bmSigma}_{u}
\bmSigma^{-1}
\dot{\bmSigma}_{u}
\right)\\
&=
\frac{1}{2}
\left\{
\tr\!\left(
\bmSigma_{pp}^{-1}\bmK^{(0)}_{pq}
\bmSigma_{qq}^{-1}\bmK^{(0)\top}_{pq}
\right)
+
\tr\!\left(
\bmSigma_{qq}^{-1}\bmK^{(0)\top}_{pq}
\bmSigma_{pp}^{-1}\bmK^{(0)}_{pq}
\right)
\right\}\\
&=
\tr\!\left(
\bmSigma_{pp}^{-1}\bmK^{(0)}_{pq}
\bmSigma_{qq}^{-1}\bmK^{(0)\top}_{pq}
\right)
=
\left\|
\bmSigma_{pp}^{-1/2}
\bmK^{(0)}_{pq}
\bmSigma_{qq}^{-1/2}
\right\|_F^2.
\end{aligned}
\]
This proves the diagonal and whitened-mass claims.

For the explicit bounds, use positive semidefiniteness of \(\bmK^{(0)}_{ii}\) and \(\tau_i^2>0\) to obtain
\[
\tau_i^2\bmI_{N_i}
\preceq
\bmSigma_{ii}
\preceq
\left\{
\Lambda_{ii}\lambda_{\max}(\bmK^{(0)}_{ii})+\tau_i^2
\right\}\bmI_{N_i}.
\]
Taking inverses reverses the Loewner order. Therefore
\[
\frac{1}{
\Lambda_{ii}\lambda_{\max}(\bmK^{(0)}_{ii})+\tau_i^2
}\bmI_{N_i}
\preceq
\bmSigma_{ii}^{-1}
\preceq
\frac{1}{\tau_i^2}\bmI_{N_i}.
\]
Substituting these two bounds into
\(\tr(\bmSigma_{pp}^{-1}\bmK^{(0)}_{pq}\bmSigma_{qq}^{-1}\bmK^{(0)\top}_{pq})\)
gives the lower and upper bounds in Eq.~\eqref{eq:orthogonality_explicit_bounds}, because
\(\|\bmK^{(0)}_{pq}\|_F^2=W_{pq}\).
\end{proof}

\begin{proof}[Proof of Proposition~\ref{prop:coverage_information_bound}]
\leavevmode\phantomsection\label{proof:coverage_information_bound}%
If \(\delta_{p\to q}(\bmx_a^{(p)})\le r\), then there exists \(\bmx_b^{(q)}\in\mathcal{X}_q\) such that \(\|\bmx_a^{(p)}-\bmx_b^{(q)}\|\le r\). Since \(\psi\) is non-increasing,
\[
k\!\bigl(\bmx_a^{(p)},\bmx_b^{(q)}\bigr)^2
=
\psi\!\bigl(\|\bmx_a^{(p)}-\bmx_b^{(q)}\|\bigr)^2
\ge
\psi(r)^2.
\]
Summing over all \(a\) with \(\delta_{p\to q}(\bmx_a^{(p)})\le r\) yields
\[
W_{pq}
\ge
\#\{a:\delta_{p\to q}(\bmx_a^{(p)})\le r\}\,\psi(r)^2
=
N_p\,\DC_{p\to q}(r)\,\psi(r)^2.
\]
{Interchanging \(p\) and \(q\) proves the second directional lower bound. Since \(W_{pq}=W_{qp}\), both bounds hold for the same mass and their maximum is valid. The maximum avoids double-counting shared cross-design interactions. Applying the proximity-to-coverage bound separately in each direction gives Eq.~\eqref{eq:proximity_certificate}.}
\end{proof}

\begin{proof}[Proof of Proposition~\ref{prop:separation_upper_bound}]
\leavevmode\phantomsection\label{proof:separation_upper_bound}%
If \(\delta_{p\to q}(\bmx)\ge d_0\) for all \(\bmx\in\mathcal{X}_p\), then \(\|\bmx_a^{(p)}-\bmx_b^{(q)}\|\ge d_0\) for every cross-design pair \((a,b)\). Since \(k(\bmx,\bmx')=\psi(\|\bmx-\bmx'\|)\) with \(\psi\) non-increasing, \(k(\bmx_a^{(p)},\bmx_b^{(q)})\le \psi(d_0)\) for all \(a\) and \(b\). Therefore
\[
W_{pq}
=
\sum_{a=1}^{N_p}\sum_{b=1}^{N_q}
k\!\bigl(\bmx_a^{(p)},\bmx_b^{(q)}\bigr)^2
\le
N_pN_q\,\psi(d_0)^2.
\]
The Fisher information bound then follows from the upper bound in Theorem~\ref{thm:information_scaling}.
\end{proof}

\begingroup 
\begin{proof}[Proof of Theorem~\ref{thm:zero_overlap_dichotomy}]
\leavevmode\phantomsection\label{proof:zero_overlap_dichotomy}%
\textbf{(i)} For fixed \(r>0\), Proposition~\ref{prop:proximity_controls_coverage} gives
\[
1\ge\DC_{p\to q}^{(n)}(r)
\ge1-\pi_{p\to q}^{(n)}/r\longrightarrow1,
\]
and likewise in the reverse direction. By Proposition~\ref{prop:coverage_information_bound},
\[
\frac{W_{pq}^{(n)}}{N_p^{(n)}}
\ge\DC_{p\to q}^{(n)}(r)\psi(r)^2,
\qquad
\frac{W_{pq}^{(n)}}{N_q^{(n)}}
\ge\DC_{q\to p}^{(n)}(r)\psi(r)^2.
\]
Taking lower limits proves both claims.

\textbf{(ii)} A cross-design gap of at least \(d_0\) makes both directed coverages zero for \(r<d_0\) and both proximities at least \(d_0\). Thus the coverage bound and positive-part proximity certificate vanish at those radii. Every term in \(W_{pq}^{(n)}\) is at most \(\psi(d_0)^2\), giving
\[
W_{pq}^{(n)}\le N_p^{(n)}N_q^{(n)}\psi(d_0)^2.
\]
\end{proof}
\endgroup 

\subsection[Normalisation and scale sensitivity]{{Normalisation and scale sensitivity}}
\label{subsec:diagnostic_details}
\begingroup 
\paragraph*{Average interaction mass and correlation information}
For the unit-diagonal kernel in main Section~\ref{subsec:mvgp_model},
\[
\overline W_{pq}=\frac{W_{pq}}{N_pN_q}
=\frac{1}{N_pN_q}\sum_{a=1}^{N_p}\sum_{b=1}^{N_q}
k(\bmx_a^{(p)},\bmx_b^{(q)})^2\in[0,1].
\]
Normalising by the pair count gives average squared kernel interaction at a common metric and scale. Signal variances reside in \(\bmLambda\), so both masses are independent of output units. Experiment~I-B illustrates the additional effects of noise and redundancy, which neither mass accounts for.
The descriptive correlation-information score follows from Theorem~\ref{thm:orthogonality}. Write \(\Lambda_{pq}=\rho_{pq}\sqrt{\Lambda_{pp}\Lambda_{qq}}\). Holding the marginal parameters fixed gives
\[
\frac{\partial\bmSigma}{\partial\rho_{pq}}
=\sqrt{\Lambda_{pp}\Lambda_{qq}}\,
\frac{\partial\bmSigma}{\partial\Lambda_{pq}}.
\]
Substitution into the Gaussian information formula yields, at independence,
\begin{equation}
\mathcal J_{pq}=\Lambda_{pp}\Lambda_{qq}\mathcal I_{u,u}
=\Lambda_{pp}\Lambda_{qq}
\|\bmSigma_{pp}^{-1/2}\bmK^{(0)}_{pq}\bmSigma_{qq}^{-1/2}\|_F^2,
\qquad u=(p,q).
\label{eq:correlation_information}
\end{equation}
At independence, \(\rho_{pq}=0\) for every \(p\ne q\), so
\(\partial\Lambda_{pq}/\partial\Lambda_{pp}
=\tfrac12\rho_{pq}\sqrt{\Lambda_{qq}/\Lambda_{pp}}=0\),
and likewise for \(\Lambda_{qq}\). Nuisance derivatives therefore remain block
diagonal in correlation coordinates, while correlation derivatives occupy
off-diagonal blocks; their information cross-terms remain zero under the
theorem's assumptions.
Rescaling output \(i\) by a nonzero constant \(a_i\) multiplies \(\Lambda_{ii}\) and \(\bmSigma_{ii}\) by \(a_i^2\). The squared whitened norm is therefore divided by \(a_p^2a_q^2\), cancelling the change in \(\Lambda_{pp}\Lambda_{qq}\). Thus \(\mathcal J_{pq}\) is invariant to output units while retaining sample-size, within-output covariance and signal-to-noise effects.
\endgroup 

\begingroup 
\paragraph*{Information ceiling}\label{proof:information_ceiling}
For a common positive-semidefinite kernel, positive signal variances and nonnegative independent diagonal observation noise, assume that the two marginal covariance blocks are invertible. Then
\[
0\leq\mathcal J_{pq}\leq\min(N_p,N_q).
\]
To see this, construct the two latent outputs as scaled copies of one GP, retaining their respective independent observation noises. Their covariance has marginal blocks \(\bmSigma_{pp}\), \(\bmSigma_{qq}\) and cross block \(\sqrt{\Lambda_{pp}\Lambda_{qq}}\bmK^{(0)}_{pq}\). Marginal whitening therefore gives
\[
\begin{pmatrix}I&C\\C^\top&I\end{pmatrix}\succeq0,
\qquad C=\sqrt{\Lambda_{pp}\Lambda_{qq}}\,
\bmSigma_{pp}^{-1/2}\bmK^{(0)}_{pq}\bmSigma_{qq}^{-1/2}.
\]
The Schur complement implies \(C^\top C\preceq I\), so every singular value of \(C\) is at most one and
\[
\mathcal J_{pq}=\|C\|_F^2\leq\operatorname{rank}(C)\leq\min(N_p,N_q).
\]
The perfectly correlated construction proves a covariance inequality; the information score itself is evaluated at independence, under the interior assumptions of main Theorem~\ref{thm:orthogonality}. The bound holds at each common scale. Cholesky factors \(\bmSigma_{ii}=L_iL_i^\top\) give the equivalent formula
\(\mathcal J_{pq}=\|\sqrt{\Lambda_{pp}\Lambda_{qq}}L_p^{-1}\bmK^{(0)}_{pq}L_q^{-\top}\|_F^2\).
Thus the descriptive ratio \(\mathcal J_{pq}/\min(N_p,N_q)\) lies in \([0,1]\).

\endgroup 

\begingroup 
\paragraph*{Common anisotropy and raw-mass sensitivity}
For positive dimension-specific scales \(\boldsymbol\ell=(\ell_1,\ldots,\ell_T)\), define
\[
d_{\boldsymbol\ell}(x,x')=\left\{\sum_{t=1}^T\frac{(x_t-x_t')^2}{\ell_t^2}\right\}^{1/2},
\qquad k_{\boldsymbol\ell}(x,x')=\psi(d_{\boldsymbol\ell}(x,x')).
\]
A single common coordinate transformation \(x_t\mapsto x_t/\ell_t\) brings this kernel into the radial setting. At each fixed scale vector, the coverage and proximity bounds of main Sections~\ref{sec:geometry}--\ref{sec:estimation} apply in that transformed metric. The following result supplies the two-corner calculation in Remark~\ref{rem:anisotropic_scales}. The family \(\psi\), including any smoothness parameter, is held fixed when varying \(\boldsymbol\ell\).

\begin{proposition}[Two-corner bound for anisotropic interaction mass]
\label{prop:anisotropic_mass_corners}
For fixed cross-design pairs, suppose \(\psi\ge0\) is non-increasing and the scale sensitivity set is a rectangle \(\mathcal L=\prod_{t=1}^T[\ell_t^-,\ell_t^+]\), with \(0<\ell_t^-\le\ell_t^+\). Then
\[
W_{pq}(\boldsymbol\ell^-)
\le W_{pq}(\boldsymbol\ell)
\le W_{pq}(\boldsymbol\ell^+),\qquad \boldsymbol\ell\in\mathcal L.
\]
\end{proposition}
\begin{proof}
\leavevmode\phantomsection\label{proof:anisotropic_mass_corners}%
Increasing any coordinate of \(\boldsymbol\ell\) decreases each transformed distance, so it increases each nonnegative kernel value and its square. Sum over the fixed cross-design pairs. The minimum and maximum therefore occur at the two componentwise corners.
\end{proof}
Only two mass evaluations are needed for this rectangular bound. Whitening, residualisation and nuisance adjustment also change with scale, so \(\mathcal J_{pq}\), oracle gain and fitted-model decisions require separate sensitivity assessment. Experiment~V-B examines conditional opportunity over finite scale scenarios and checks their sensitivity against a denser grid in two dimensions (Appendix~\ref{sec:e9_scales}). Rectangles derived from marginal scale estimates describe a chosen range of scales; interpreting them as simultaneous confidence regions would require a coverage calculation.
\endgroup 
\space
\space

\subsection[Net benefit results and mean-square expansion]{{Net benefit results and mean-square expansion}}
\label{sec:delta_derivation}
\begingroup 
\begin{proof}[Proof of Proposition~\ref{prop:risk_decomp}]
\leavevmode\phantomsection\label{proof:risk_decomp}%
Under correct specification the oracle joint mean is \(m=\mathbb E[f_j(\bmx_\star)\mid\bmy]\), and the oracle marginal mean is \(m_j=\mathbb E[f_j(\bmx_\star)\mid\bmy^{(j)}]\). Conditional-expectation orthogonality gives \(\mathbb E[(f_j-m)(\widehat g-m)]=0\) for every square-integrable \(\bmy\)-measurable \(\widehat g\). Applying the same argument in \(\sigma(\bmy^{(j)})\) gives the independent identity. Conditional Gaussian variances depend on the fixed design and parameters, so their expectations are the stated oracle variances. Expanding both squares proves the result.

The measurability condition includes deterministic fitting rules; independent optimiser randomness may be conditioned on, or added to the information set without altering the oracle conditional mean. A marginal predictor using a nuisance estimate learned from all outputs (\texttt{IND-SHARED}) need not satisfy the marginal condition. For that predictor, expanding about \(m_j\) yields
\[
R_{\mathrm{shared}}=V_{\mathrm{ind}}+\mathbb E[(\widehat g_{\mathrm{shared}}-m_j)^2]
-2\mathbb E[(f_j-m_j)(\widehat g_{\mathrm{shared}}-m_j)].
\]
The final term can be nonzero because the nuisance estimate uses auxiliary responses. This explains the main text's information-set restriction.
\end{proof}

\endgroup \begingroup 
\begin{proof}[Proof of Theorem~\ref{thm:est_penalty}]
\leavevmode\phantomsection\label{proof:est_penalty}%
Let \(L_n=a_n(Y_n)^\top e_n\). Taylor's remainder is defined exactly by
\begin{align}
g_n^{\mathrm{mv}}(\widehat\theta_n,Y_n)-g_n^{\mathrm{mv}}(\theta_0,Y_n)&=L_{n,\mathrm{mv}}+r_{n,\mathrm{mv}},\label{eq:si_delta_mv}\\
g_n^{\mathrm{ind}}(\widehat\vartheta_n,Y_n^{(j)})-g_n^{\mathrm{ind}}(\vartheta_0,Y_n^{(j)})&=L_{n,\mathrm{ind}}+r_{n,\mathrm{ind}}.\label{eq:si_delta_ind}
\end{align}
For either rule, the definition of the estimation penalty gives
\[
H_n=\mathbb E[(L_n+r_n)^2]
=\mathbb E[L_n^2]+2\mathbb E[L_nr_n]+\mathbb E[r_n^2].
\]
Cauchy--Schwarz then gives
\[
|H_n-\mathbb E[L_n^2]|
\le 2\sqrt{\mathbb E[L_n^2]\mathbb E[r_n^2]}+\mathbb E[r_n^2]=o(h_n).
\]
Hence the two penalties satisfy
\begin{align}
H_{n,\mathrm{mv}}&=\mathbb E[L_{n,\mathrm{mv}}^2]+o(h_n),\label{eq:si_H_mv}\\
H_{n,\mathrm{ind}}&=\mathbb E[L_{n,\mathrm{ind}}^2]+o(h_n),\label{eq:si_H_ind}
\end{align}
with separate scales if necessary. The second-moment bound controls the expected squared remainder without requiring independence between the gradient and estimation error.
\end{proof}

\paragraph*{Likelihood information and score substitution}
For the centred Gaussian covariance model, \(\mathcal I(\bmtheta)=I_n(\bmtheta)\) denotes the expected covariance-parameter information for the entire observation vector, with \(n=N_{\mathrm{tot}}\):
\begin{equation}
[I_n]_{uv}=\tfrac12\tr(\bmSigma_n^{-1}\bmSigma_{n,u}\bmSigma_n^{-1}\bmSigma_{n,v}),
\qquad J_n(\widehat\theta)=-\nabla^2\ell_n(\widehat\theta;Y_n),
\label{eq:information_convention}
\end{equation}
where \(\bmSigma_{n,u}=\partial\bmSigma_n/\partial\theta_u\) and \(J_n\) is the observed negative Hessian. Its inverse and \(I_n^{-1}\) have full-sample scaling; neither is divided by \(n\) again. The coordinate system must be the same for the gradient and information: if \(\theta=t(\phi)\), then \(a_\phi=Dt(\phi)^\top a_\theta\) and \(I_\phi=Dt(\phi)^\top I_\theta Dt(\phi)\). Observed Hessians have additional transformation terms away from a stationary likelihood point.

Under the conditions of main Theorem~\ref{thm:est_penalty}, let \(s_n=\nabla_\theta\ell_n(\theta_0;Y_n)\) be the score and suppose \(I_n=I_n(\theta_0)\) is invertible. If
\begin{equation}
\mathbb E[\{a_n(Y_n)^\top(e_n-I_n^{-1}s_n)\}^2]=o(h_n),
\label{eq:score_linearisation}
\end{equation}
then the parameter error can be replaced by its score approximation:
\[
H_n=\mathbb E[(a_n(Y_n)^\top I_n^{-1}s_n)^2]+o(h_n).
\]
To see this, set \(d_n=a_n(Y_n)^\top(e_n-I_n^{-1}s_n)\) and \(L_n=a_n(Y_n)^\top e_n\). Then \(\mathbb E[d_n^2]=o(h_n)\) and \(\mathbb E[(L_n-d_n)^2]=O(h_n)\). Applying the same Cauchy--Schwarz bound as in the theorem's proof establishes the replacement. It retains the dependence between the random prediction gradient and score.

\paragraph*{Dependence between the gradient and estimation error}
Replacing the joint second moment by \(\mathbb E[a_n(Y_n)^\top I_n^{-1}a_n(Y_n)]\) requires additional moment and dependence assumptions.
As a simple algebraic counterexample, let \(Y\sim N(0,1)\), \(g_n(\theta,Y)=\theta Y\), \(\theta_0=0\), and \(\widehat\theta_n=Y/\sqrt n\). The estimator is exactly normally distributed with variance \(1/n\), and the Taylor remainder is zero, but
\[
H_n=\mathbb E[Y^4]/n=3/n,\qquad
\mathbb E[a(Y)^2]\operatorname{Var}(\widehat\theta_n)=1/n.
\]
The factor-of-three difference arises from dependence between the gradient and estimation error despite an exactly normal estimator and a zero Taylor remainder.

\paragraph*{Dependence and nuisance parameters}
For a positive-definite information matrix partitioned into dependence \(\eta\) and nuisance \(\nu\), set \(S=I_{\eta\eta}-I_{\eta\nu}I_{\nu\nu}^{-1}I_{\nu\eta}\). Exact block inversion gives, for a realised gradient \(a\),
\begin{equation}
\begin{split}
a^\top I^{-1}a={}&a_\nu^\top I_{\nu\nu}^{-1}a_\nu\\
&+(a_\eta-I_{\eta\nu}I_{\nu\nu}^{-1}a_\nu)^\top
S^{-1}(a_\eta-I_{\eta\nu}I_{\nu\nu}^{-1}a_\nu).
\end{split}
\label{eq:penalty_block_identity}
\end{equation}
\phantomsection\label{proof:penalty_block_identity}
For positive-definite \(I\), its inverse has blocks
\[
I^{-1}=\begin{pmatrix}
S^{-1}&-S^{-1}I_{\eta\nu}I_{\nu\nu}^{-1}\\
-I_{\nu\nu}^{-1}I_{\nu\eta}S^{-1}&I_{\nu\nu}^{-1}+I_{\nu\nu}^{-1}I_{\nu\eta}S^{-1}I_{\eta\nu}I_{\nu\nu}^{-1}
\end{pmatrix}.
\]
Multiplication by \((a_\eta^\top,a_\nu^\top)\) and its transpose, followed by completing the square, proves Eq.~\eqref{eq:penalty_block_identity}. Thus \(a_\eta^\top S^{-1}a_\eta\) is not generally an additive dependence contribution to the full quadratic form. At independence, main Theorem~\ref{thm:orthogonality} supplies \(I_{\eta\nu}=0\), so the adjusted gradient reduces to \(a_\eta\). If nuisance parameters are actually fixed, the relevant information is the conditional block \(I_{\eta\eta}\).

\paragraph*{Independent calibration and asymptotic scope}
\begin{proposition}[Independent-calibration benchmark]
\leavevmode\label{prop:independent_calibration_penalty}%
Fix the distribution of the response vector \(Y\) used for prediction and a finite-dimensional, identified, interior parameter \(\theta_0\). Suppose \(\widehat\theta_M\) is estimated from \(M\) independent replicate processes, independently of \(Y\), and
\(M\mathbb E[(\widehat\theta_M-\theta_0)(\widehat\theta_M-\theta_0)^\top]\to I_1^{-1}\), with \(I_1\) positive definite. Assume \(\mathbb E\|a(Y)\|^2<\infty\) and the prediction Taylor remainder has second moment \(o(M^{-1})\). Then
\[
H_M=M^{-1}\mathbb E[a(Y)^\top I_1^{-1}a(Y)]+o(M^{-1}).
\]
\end{proposition}
\begin{proof}
\leavevmode\phantomsection\label{proof:independent_calibration_penalty}%
Independence gives
\(\mathbb E[(a(Y)^\top e_M)^2]=\tr\{\mathbb E[a(Y)a(Y)^\top]\mathbb E[e_Me_M^\top]\}\).
The assumed second-moment convergence gives the leading term and the \(O(M^{-1})\) bound. Theorem~\ref{thm:est_penalty} with \(h_M=M^{-1}\) controls the remainder. The second-moment hypothesis accounts for both the variance and bias of the estimation error.
\end{proof}
The benchmark applies when parameter estimation uses replicate processes independent of the prediction responses. Estimation from the same responses requires the joint moment conditions in Theorem~\ref{thm:est_penalty}.
The applicable asymptotic regime depends on the sampling design. Fixed-domain Mat\'ern covariance estimation can identify only combinations of variance and range under the relevant conditions; increasing-domain results require design and identifiability assumptions \citep{zhang2004inconsistent,bachoc2020asymptotic}. Applying main Theorem~\ref{thm:est_penalty} to either regime requires its mean-square conditions to hold for the estimator and design; the score substitution additionally requires Eq.~\eqref{eq:score_linearisation}. Independent simulation replications estimate repeated-sampling performance at each fixed design.

\endgroup \begingroup \paragraph*{Working diagnostic and prediction-region comparison}
Using the fitted prediction gradient \(a(\widehat\theta,Y)\) from main Section~\ref{sec:net_benefit}, the working diagnostic is
\begin{equation}
\widehat P=a(\widehat\theta,Y)^\top I_n(\widehat\theta)^{-1}a(\widehat\theta,Y).
\label{eq:working_penalty}
\end{equation}
The observed-information version uses \(J_n\) from Eq.~\eqref{eq:information_convention}. Appendix~\ref{sec:revision_penalty_methods} specifies numerical availability checks and the Experiment~II-C comparison with refitting and validation. Both information and prediction sensitivity enter the quadratic form: under noise-free autokrigeability, for example, the dependence gradient is identically zero regardless of dependence information.

For a finite prediction set \(\mathcal A\), averaging main Eq.~\eqref{eq:net_gain_decomp} gives the exact criterion for positive mean squared-error benefit:
\begin{equation}
\frac1{|\mathcal A|}\sum_{\bmx_\star\in\mathcal A}\Delta_j(\bmx_\star)>
\frac1{|\mathcal A|}\sum_{\bmx_\star\in\mathcal A}(H_{\mathrm{mv}}^{(j)}-H_{\mathrm{ind}}^{(j)})(\bmx_\star).
\label{eq:average_net_benefit}
\end{equation}
For a prescribed relative squared-error improvement fraction \(\epsilon\), the squared-error materiality criterion adds \(\epsilon |\mathcal A|^{-1}\sum_{\bmx_\star\in\mathcal A}R_{\mathrm{ind}}^{(j)}(\bmx_\star)\) to the right-hand side. Equation~\eqref{eq:local_net_benefit} evaluates the corresponding post-fit comparator using estimated gain and penalties on the same grid.
\endgroup 
\space

\section{Additional Prediction Results}
\label{sec:si_prediction_results}

\subsection{Two-output residual-channel characterisation}

\begin{theorem}[Two-sided characterisation of the residual channel]
\label{thm:si_two_output_residual_bounds}
Assume \(D=2\) with target output \(j=1\), suppose Eq.~\eqref{eq:radial_kernel} holds with \(\psi(0)=1\), and assume \(\tau_1^2>0\) and \(\tau_2^2>0\). Let
$
\bmK_{ab}^{(0)}
:=
\left[
k\!\bigl(\bmx_u^{(a)},\bmx_v^{(b)}\bigr)
\right]_{1\le u\le N_a,\ 1\le v\le N_b},
a,b\in\{1,2\},
$
and define
$
\bmk_p(\bmx_\star)
:=
\bigl(
k(\bmx_\star,\bmx_1^{(p)}),\ldots,
k(\bmx_\star,\bmx_{N_p}^{(p)})
\bigr)^\top, p\in\{1,2\}.
$
Set \(\sigma_1^2:=\tau_1^2/\Lambda_{11}\), and define
$
\bmk_2^{\mathrm{res}}(\bmx_\star)
:=
\bmk_2(\bmx_\star)
-
\bmK_{21}^{(0)}
\bigl(\bmK_{11}^{(0)}+\sigma_1^2\bmI\bigr)^{-1}
\bmk_1(\bmx_\star).
$
Define
\[
S_2
:=
\max_{1\le b\le N_2}
\sum_{b'\neq b}
\psi\!\left(\left\|\bmx_b^{(2)}-\bmx_{b'}^{(2)}\right\|\right),
\]
with the convention that an empty sum is zero. Then:
\begin{enumerate}[label=(\roman*),nosep]
\item
\[
\frac{\Lambda_{12}^2}
{\Lambda_{22}(1+S_2)+\tau_2^2}
\bigl\|\bmk_2^{\mathrm{res}}(\bmx_\star)\bigr\|_2^2
\le
\Delta_1(\bmx_\star)
\le
\frac{\Lambda_{12}^2}{\tau_2^2}
\bigl\|\bmk_2^{\mathrm{res}}(\bmx_\star)\bigr\|_2^2.
\]
\item If \(b^\star\in\arg\min_{1\le b\le N_2}
\|\bmx_\star-\bmx_b^{(2)}\|\), then
\(\Delta_1(\bmx_\star)\ge \dfrac{\Lambda_{12}^2}{\Lambda_{22}+\tau_2^2}\,
\bigl[k_{2,b^\star}^{\mathrm{res}}(\bmx_\star)\bigr]^2\),
where \(k_{2,b^\star}^{\mathrm{res}}(\bmx_\star)\) denotes the
\(b^\star\)-th entry of \(\bmk_2^{\mathrm{res}}(\bmx_\star)\).
\end{enumerate}
\end{theorem}

\begingroup 
\begin{proof}[Proof of Theorem~\ref{thm:si_two_output_residual_bounds}]
\leavevmode\phantomsection\label{proof:two_output_residual_bounds}%
Write \(\bmA:=\bmSigma_{22\mid1}
=\bmSigma_{22}-\bmSigma_{21}\bmSigma_{11}^{-1}\bmSigma_{12}\).
Theorem~\ref{thm:exact_oracle_gain} gives
\begin{equation}
\Delta_1(\bmx_\star)
=\mathbf c_{1,2\mid1}(\bmx_\star)\bmA^{-1}
\mathbf c_{1,2\mid1}(\bmx_\star)^\top.
\label{eq:si_two_output_gain_identity}
\end{equation}
Under the separable model,
\(\bmSigma_{11}=\Lambda_{11}(\bmK_{11}^{(0)}+\sigma_1^2\bmI)\)
and \(\bmSigma_{12}=\Lambda_{12}\bmK_{12}^{(0)}\), so
\begin{align}
\mathbf c_{1,2\mid1}(\bmx_\star)
&=\Lambda_{12}\left[
\bmk_2(\bmx_\star)^\top
-\bmk_1(\bmx_\star)^\top
(\bmK_{11}^{(0)}+\sigma_1^2\bmI)^{-1}\bmK_{12}^{(0)}
\right]
\label{eq:si_two_output_cond_cov_expand}\\
&=\Lambda_{12}\bmk_2^{\mathrm{res}}(\bmx_\star)^\top,
\label{eq:si_two_output_cond_cov_residual}
\end{align}
using symmetry of the kernel.

For part (i), independent auxiliary noise and the Schur-complement ordering give
\[
\tau_2^2\bmI
\preceq\bmA
=\operatorname{Cov}(\bmf^{(2)}\mid\bmy^{(1)})+\tau_2^2\bmI
\preceq\bmSigma_{22}
=\Lambda_{22}\bmK_{22}^{(0)}+\tau_2^2\bmI.
\]
The nonnegative entries and unit diagonal of \(\bmK_{22}^{(0)}\)
imply \(\lambda_{\max}(\bmK_{22}^{(0)})\le1+S_2\) by its maximum row sum.
Since \(\tau_2^2>0\), inversion therefore yields
\[
\frac{\bmI}{\Lambda_{22}(1+S_2)+\tau_2^2}
\preceq\bmA^{-1}\preceq\frac{\bmI}{\tau_2^2}.
\]
Substitute these bounds and Eq.~\eqref{eq:si_two_output_cond_cov_residual}
into Eq.~\eqref{eq:si_two_output_gain_identity}.

For part (ii), set \(\bmv=\Lambda_{12}\bmk_2^{\mathrm{res}}(\bmx_\star)\).
For any auxiliary index \(b\), restricting the variational quadratic form
to vectors \(t\mathbf e_b\) gives
\[
\Delta_1(\bmx_\star)
=\max_{\bmu}\{2\bmv^\top\bmu-\bmu^\top\bmA\bmu\}
\ge\max_t\{2tv_b-t^2A_{bb}\}
=\frac{v_b^2}{A_{bb}}.
\]
Because \(A_{bb}=\operatorname{Var}(y_b^{(2)}\mid\bmy^{(1)})
\le\operatorname{Var}(y_b^{(2)})=\Lambda_{22}+\tau_2^2\), choosing
\(b=b^\star\) proves the single-neighbour bound.
\end{proof}
\endgroup 

\subsection{BPI certificate for positive borrowing}

\begin{proposition}[BPI certificate for positive borrowing]
\label{prop:si_bpi_certificate}
In the setting of Theorem~\ref{thm:si_two_output_residual_bounds}, let \(b^\star\in\arg\min_{1\le b\le N_2}\|\bmx_\star-\bmx_b^{(2)}\|\). For \(z\in\mathcal{X}\), write
\[
\bmk_1(z)
:=
\bigl(
k(z,\bmx_1^{(1)}),\ldots,
k(z,\bmx_{N_1}^{(1)})
\bigr)^\top
\]
and define the leverage of the target design for the unit-variance GP
\[
H_1(z)
:=
\bmk_1(z)^\top
\bigl(\bmK_{11}^{(0)}+\sigma_1^2\bmI\bigr)^{-1}
\bmk_1(z).
\]
Then
\[
k_{2,b^\star}^{\mathrm{res}}(\bmx_\star)
=
\operatorname{Cov}\!\left\{
f(\bmx_\star),f(\bmx_{b^\star}^{(2)})
\mid
f(\bmx_a^{(1)})+\epsilon_a,\ a=1,\ldots,N_1
\right\},
\]
where \(f\) is a unit-variance GP with kernel \(k\) and \(\epsilon_a\stackrel{\mathrm{ind}}{\sim}N(0,\sigma_1^2)\). Consequently,
\[
\Delta_1(\bmx_\star)
\ge
\frac{\Lambda_{12}^2}{\Lambda_{22}+\tau_2^2}
\left[
\psi\!\bigl(\delta_{-1}(\bmx_\star)\bigr)
-
\sqrt{H_1(\bmx_\star)H_1(\bmx_{b^\star}^{(2)})}
\right]_+^2.
\]
Moreover,
\begin{equation}
\Delta_1(\bmx_\star)
\ge
\frac{\Lambda_{12}^2}{\Lambda_{22}+\tau_2^2}
\left[
\psi\!\bigl(\delta_{-1}(\bmx_\star)\bigr)
-
\frac{N_1\Lambda_{11}}{\tau_1^2}
\psi\!\bigl(\delta_1(\bmx_\star)\bigr)
\psi\!\bigl(\delta_{2\to 1}(\bmx_{b^\star}^{(2)})\bigr)
\right]_+^2.
\label{eq:si_bpi_certificate_distance}
\end{equation}
If \(b_1(\bmx_\star)>0\), then \(\delta_{-1}(\bmx_\star)=\{1-b_1(\bmx_\star)\}\delta_1(\bmx_\star)\) and
\(\delta_{2\to 1}(\bmx_{b^\star}^{(2)})\ge b_1(\bmx_\star)\delta_1(\bmx_\star)\), so the right-hand side of Eq.~\eqref{eq:si_bpi_certificate_distance} is bounded below by
\[
\frac{\Lambda_{12}^2}{\Lambda_{22}+\tau_2^2}
\left[
\psi\!\bigl(\{1-b_1(\bmx_\star)\}\delta_1(\bmx_\star)\bigr)
-
\frac{N_1\Lambda_{11}}{\tau_1^2}
\psi\!\bigl(\delta_1(\bmx_\star)\bigr)
\psi\!\bigl(b_1(\bmx_\star)\delta_1(\bmx_\star)\bigr)
\right]_+^2.
\]
{Thus a positive bracket and nonzero \(\Lambda_{12}\) give strictly positive oracle gain. At \(\Lambda_{12}=0\), both the prefactor and gain are zero.}
\end{proposition}

	\begin{proof}[Proof of Proposition~\ref{prop:si_bpi_certificate}]
\leavevmode\phantomsection\label{proof:bpi_certificate}%
For a unit-variance GP \(f\) with kernel \(k\), observed with independent noise variance \(\sigma_1^2\) on \(\mathcal{X}_1\), the standard Gaussian conditioning formula gives
\[
\begin{aligned}
&\operatorname{Cov}\!\left\{
f(\bmx_\star),f(\bmx_b^{(2)})
\mid
f(\bmx_a^{(1)})+\epsilon_a,\ a=1,\ldots,N_1
\right\}\\
&\quad=
k(\bmx_\star,\bmx_b^{(2)})
-
\bmk_1(\bmx_\star)^\top
\bigl(\bmK_{11}^{(0)}+\sigma_1^2\bmI\bigr)^{-1}
\bmk_1(\bmx_b^{(2)}).
\end{aligned}
\]
The right-hand side is exactly \(k_{2,b}^{\mathrm{res}}(\bmx_\star)\), the \(b\)-th entry of the residual vector in Theorem~\ref{thm:si_two_output_residual_bounds}. Taking \(b=b^\star\) proves the identity for the conditional covariance.

Let
\[
\bmA_1
:=
\bigl(\bmK_{11}^{(0)}+\sigma_1^2\bmI\bigr)^{-1}.
\]
By Cauchy--Schwarz in the \(\bmA_1\)-inner product,
\[
\left|
\bmk_1(\bmx_\star)^\top
\bmA_1
\bmk_1(\bmx_{b^\star}^{(2)})
\right|
\le
\sqrt{
H_1(\bmx_\star)
H_1(\bmx_{b^\star}^{(2)})
}.
\]
Since \(b^\star\) is a nearest auxiliary point to \(\bmx_\star\),
\[
k(\bmx_\star,\bmx_{b^\star}^{(2)})
=
\psi\!\bigl(\delta_{-1}(\bmx_\star)\bigr).
\]
Therefore
\[
k_{2,b^\star}^{\mathrm{res}}(\bmx_\star)
\ge
\psi\!\bigl(\delta_{-1}(\bmx_\star)\bigr)
-
\sqrt{
H_1(\bmx_\star)
H_1(\bmx_{b^\star}^{(2)})
}.
\]
Combining this with the one-point lower bound in Theorem~\ref{thm:si_two_output_residual_bounds} gives the first displayed lower bound. The positive part is valid because, if the bracket is non-positive, the claimed lower bound is zero; if it is positive, then \(k_{2,b^\star}^{\mathrm{res}}(\bmx_\star)\) is at least the bracket and hence its square is at least the squared bracket.

For the certificate based only on distance, note that
\[
\bmK_{11}^{(0)}+\sigma_1^2\bmI
\succeq
\sigma_1^2\bmI,
\qquad
\bmA_1
\preceq
\sigma_1^{-2}\bmI
=
\frac{\Lambda_{11}}{\tau_1^2}\bmI.
\]
Hence
\[
\begin{aligned}
\left|
\bmk_1(\bmx_\star)^\top
\bmA_1
\bmk_1(\bmx_{b^\star}^{(2)})
\right|
&\le
\frac{\Lambda_{11}}{\tau_1^2}
\left\|\bmk_1(\bmx_\star)\right\|_2
\left\|\bmk_1(\bmx_{b^\star}^{(2)})\right\|_2\\
&\le
\frac{N_1\Lambda_{11}}{\tau_1^2}
\psi\!\bigl(\delta_1(\bmx_\star)\bigr)
\psi\!\bigl(\delta_{2\to1}(\bmx_{b^\star}^{(2)})\bigr).
\end{aligned}
\]
	Substituting this upper bound for the projection term gives Eq.~\eqref{eq:si_bpi_certificate_distance}.

Finally, if \(b_1(\bmx_\star)>0\), Definition~\ref{def:bpi} gives
\[
\delta_{-1}(\bmx_\star)
=
\{1-b_1(\bmx_\star)\}\delta_1(\bmx_\star).
\]
For any \(\bmx_a^{(1)}\in\mathcal{X}_1\), the triangle inequality gives
\[
\left\|
\bmx_{b^\star}^{(2)}-\bmx_a^{(1)}
\right\|
\ge
\left\|
\bmx_\star-\bmx_a^{(1)}
\right\|
-
\left\|
\bmx_\star-\bmx_{b^\star}^{(2)}
\right\|
\ge
\delta_1(\bmx_\star)-\delta_{-1}(\bmx_\star)
=
b_1(\bmx_\star)\delta_1(\bmx_\star).
\]
Taking the minimum over \(a\) gives
\(\delta_{2\to1}(\bmx_{b^\star}^{(2)})\ge b_1(\bmx_\star)\delta_1(\bmx_\star)\). Since \(\psi\) is non-increasing, substituting
\[
\psi\!\bigl(\delta_{2\to1}(\bmx_{b^\star}^{(2)})\bigr)
\le
\psi\!\bigl(b_1(\bmx_\star)\delta_1(\bmx_\star)\bigr)
\]
	into Eq.~\eqref{eq:si_bpi_certificate_distance} gives the BPI-based certificate.
\end{proof}

\section[Geometric borrowing opportunities]{{Geometric borrowing opportunities}}
\label{sec:si_experiments}

\subsection[Generation, fitting and assessment conventions]{{Generation, fitting and assessment conventions}}
\label{subsec:si_simulation}
\begingroup 
\paragraph*{Designs and generation}
Isotopic outputs share distinct sites; interleaved sites provide nearby cross-output neighbours; infill fills target-design gaps. Within-region, partial-coverage and half-region auxiliaries cover all, part, or half of the target region; adjacent regions meet at a boundary and remote regions have a positive gap. Counts and replication are specified separately. Interaction depends on distances relative to kernel scales.

The default model has zero means, latent covariance $\Lambda_{pq}k_\ell(x,x')$
between outputs $p,q$, and independent Gaussian noise of variance $\tau_j^2$, with
\begin{equation}
k_\ell(x,x')=(1+\sqrt3\,|x-x'|/\ell)\exp(-\sqrt3\,|x-x'|/\ell).
\label{eq:si_synth_matern32}
\end{equation}
Processes are defined over every training and prediction input. Latent training and test values are drawn jointly when sampled test responses are assessed. Geometric summaries and oracle gains are deterministic at fixed designs and covariance.

\paragraph*{Losses and uncertainty}
\label{subsec:revision_comparators}
The separate, shared-scale diagonal and joint fits are denoted SEP, DIAG and JOINT as in main Section~\ref{sec:experiments}; richer joint models use matching independent families. For Gaussian experiments, the generating full-data conditional mean $\mu_0$ and variance $V_0$ supply the assessment loss
\[
L_m^*(Y)=|\mathcal A|^{-1}\sum_{x\in\mathcal A}
\{V_0(x\mid Y)+(\mu_0(x\mid Y)-\widehat g_m(x))^2\}.
\]
Its dataset average estimates latent squared-error risk; Student-$t$ experiments use realised squared latent error. Percentage risk reductions are 100 times the paired mean loss difference divided by mean baseline loss. Latent test MSE, noisy-response mean log predictive density (MLPD), and response-interval coverage/width are separate endpoints. {Independent mean-zero future observation noise gives $R_m^y=R_m^f+\overline\tau^2$, where $\overline\tau^2$ is the grid-average noise variance. Absolute risk differences agree in expectation, but the 2\% latent-risk assessment and 2\% noisy-response validation criteria use different denominators. Validation folds also use fewer target observations than final assessment.}
Unless stated otherwise, 95\% confidence intervals use 2,000 paired dataset resamples within configuration, after averaging prediction-point losses within datasets. Queueing resamples whole macro-datasets and AQS resamples spatial blocks. {Folds and prediction points are not independent replications.}
\endgroup %
{Paired complete-procedure timings for Experiment~III, Experiments~V-A/V-B and TQA use one worker and one BLAS/OpenMP thread.}

\subsubsection[Exact fitting]{{Exact fitting}}
\label{subsec:si_common_fitting}
\begingroup 
Exact fits use float64 Gaussian likelihoods and fixed zero means unless specified otherwise. Signal/noise variances and scales use log coordinates; partial correlations use tanh coordinates in a valid correlation Cholesky factor. Joint fits estimate a common scale, marginal variances and dependence; SEP fits marginal scales separately. Bounds are signal variance $[0.0025,25]$, noise variance $[10^{-6},2]$, scale $[0.01,3]$ and absolute partial correlation at most 0.995.
L-BFGS-B uses analytic gradients, at most 1,000 iterations and 40 line-search steps, objective tolerance $10^{-8}$ and projected-gradient tolerance $10^{-5}$. The smallest successful objective across four starts defines the fit. Finite predictions and convergence at either tolerance are required. Risk and decision summaries retain these fits, including parameter-bound solutions and objective-tolerance stops above the gradient tolerance. A bound is active within $10^{-5}$ in optimisation coordinates.

Marginal starts use the training second moment clipped to $[0.05,5]$, noise at one tenth of that value clipped to $[0.001,0.5]$, and scales $0.05,0.15,0.5,1.5$. Joint starts use marginal signal estimates and the observation-weighted geometric mean of marginal scales, multiplied by $0.5,1,1,2$ and paired with target--auxiliary correlations $\rho=-0.7,-0.2,0.2,0.7$; correlations among auxiliaries are $|\rho|$. The first retains marginal noise estimates; the others use $0.03,0.1,0.3$ times the marginal signal estimate. Starts are clipped inside the bounds. DIAG-based starts are specified with Experiment~III.

Cholesky jitter starts at $10^{-10}$ and increases by powers of ten to $10^{-5}$ only after failure; it is distinct from observation noise. Information-based diagnostics use matching parameter coordinates and full-sample information, with conditioning criteria in Appendix~\ref{sec:revision_penalty_methods}.
\endgroup 

\subsection[Experiment I-A: zero-overlap geometry]{{Experiment I-A: zero-overlap geometry}}
\label{subsec:si_exp_i}
\begingroup 
The target sites are $X_1=\{i/14:i=0,\ldots,14\}$. The three 20-site auxiliary designs are
\[
\begin{aligned}
X_{2,\mathrm{int}}&=\{(i+1/2)/20:i=0,\ldots,19\},\\
X_{2,\mathrm{part}}&=\{0.505\}\cup\{0.5+i/19:i=1,\ldots,19\},\\
X_{2,\mathrm{rem}}&=\{1.5+i/19:i=0,\ldots,19\}.
\end{aligned}
\]
All have zero exact overlap. The covariance has $\ell=0.15$, unit signal variances, correlation 0.70 and noise variances $(0.08,0.08)$. Both outputs are evaluated on 200 equally weighted points in $[0,1]$.

\begin{figure}[htbp]\centering
\includegraphics[width=0.32\textwidth]{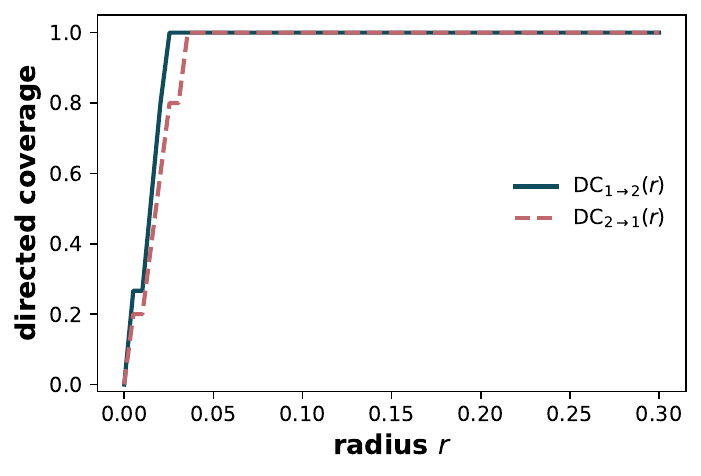}\hfill
\includegraphics[width=0.32\textwidth]{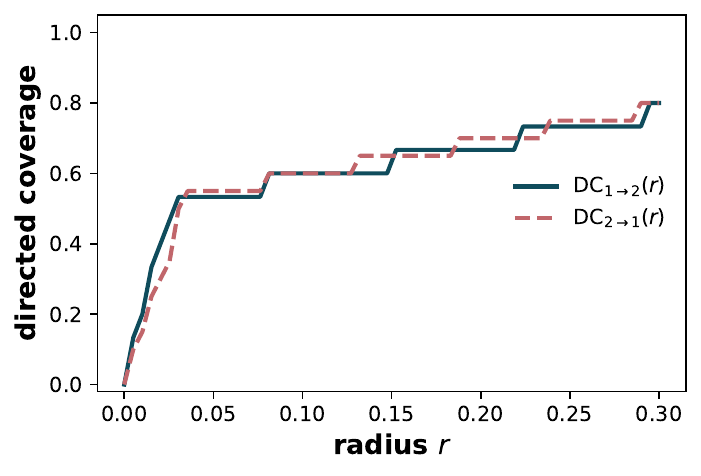}\hfill
\includegraphics[width=0.32\textwidth]{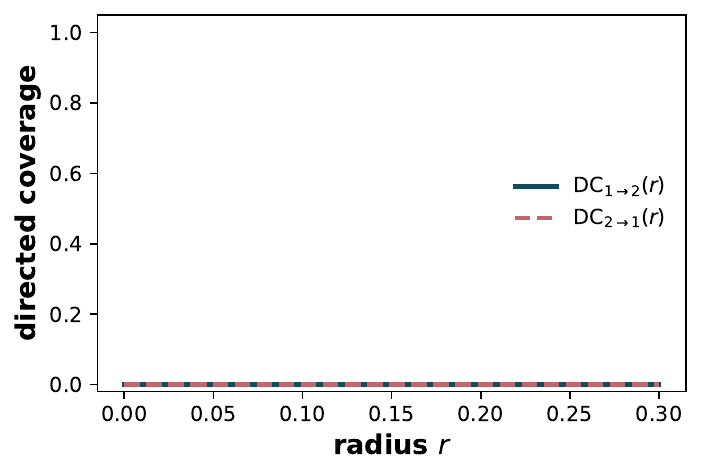}
\caption{{I-A directed coverage: interleaved, partial-coverage and remote designs, respectively.}}
\label{fig:s1_coverage}
\end{figure}
\begin{table}[htbp]\centering\footnotesize\setlength{\tabcolsep}{3pt}
\caption{{I-A mean latent posterior variances and relative oracle gain on the common grid. Relative gain is the reduction in grid-averaged variance against \texttt{ORACLE-IND}. All covariance parameters are fixed; values are deterministic.}}
\label{tab:e1_relative_oracle}

\begin{tabular}{llrrrr}
\toprule
Geometry & Target & Marginal variance & Joint variance & Oracle gain & Relative gain (\%) \\
\midrule
Interleaved & 1 & 0.06333 & 0.05800 & 0.00533 & 8.42 \\
Interleaved & 2 & 0.05053 & 0.04719 & 0.00335 & 6.62 \\
Partial coverage & 1 & 0.06333 & 0.06054 & 0.00279 & 4.41 \\
Partial coverage & 2 & 0.42678 & 0.24822 & 0.17856 & 41.84 \\
Remote & 1 & 0.06333 & 0.06333 & 0.00000 & 0.00 \\
Remote & 2 & 0.99998 & 0.54103 & 0.45895 & 45.90 \\
\bottomrule
\end{tabular}
\space
\end{table}
\endgroup 

\FloatBarrier
\Needspace{5\baselineskip}
\subsection[Experiment I-B: residual support]{{Experiment I-B: residual support}}
\label{subsec:si_exp_ii}
\begingroup 
At $x_\star=0.5$, compare the following target/auxiliary designs:
\[
\begin{array}{lll}
\text{Residual:}&X_1=\{0.1,0.9\},&X_2=\{0.5,0.85\},\\
\text{Redundant:}&X_1=\{0.45,0.50,0.55\},&X_2=\{0.49,0.50,0.51\},\\
\text{Remote:}&X_1=\{0.45,0.55\},&X_2=\{0.05,0.15\}.
\end{array}
\]
Signal variances are one, correlation is 0.65, both noise variances are 0.04, and $\ell\in\{0.06,0.12,0.20\}$. Exact Gaussian conditioning gives Figure~\ref{fig:s2_residual} and the bounds from Theorem~\ref{thm:si_two_output_residual_bounds} in Table~\ref{tab:si_s2_bounds}.
\par
\endgroup \begingroup 
The single-neighbour lower bound is nearly exact for residual support; the upper bound is conservative.
In the residual-support configuration, the local BPI is 1 and the target-design
gap is 0.4. The BPI lower bound in Proposition~\ref{prop:si_bpi_certificate}
is nearly equal to the exact oracle gain at lengthscale 0.06, reaches 95.74\%
of it at 0.12, and 0.06\% at 0.20. It certifies a positive oracle gain at all
three scales, but is quantitatively weak at lengthscale 0.20.
\par
\begin{figure}[htbp]\centering
\includegraphics[width=0.58\textwidth]{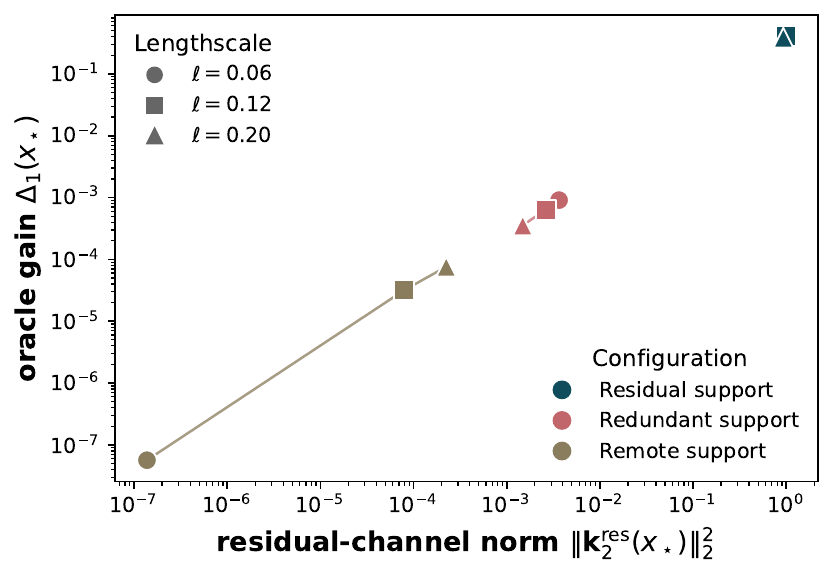}
\caption{{I-B exact oracle gain versus squared residual cross-kernel norm. Colours indicate support configurations and symbols indicate scales; both axes are logarithmic. Values are deterministic.}}
\label{fig:s2_residual}
\end{figure}
\begin{table}[htbp]\centering\small\setlength{\tabcolsep}{5pt}
\caption{{I-B single-neighbour lower and upper bounds divided by exact oracle gain. Each row fixes the target, designs and covariance.}}
\label{tab:si_s2_bounds}
\begin{tabular}{lrrr}
\hline
Configuration & $\ell$ & Lower$/\Delta_1(x_\star)$ & Upper$/\Delta_1(x_\star)$\\
\hline
Residual & 0.06 & 1.000 & 26.0\\
Residual & 0.12 & 1.000 & 26.0\\
Residual & 0.20 & 0.978 & 25.6\\
Redundant & 0.06 & 0.607 & 42.4\\
Redundant & 0.12 & 0.605 & 43.6\\
Redundant & 0.20 & 0.598 & 44.2\\
Remote & 0.06 & 0.982 & 25.7\\
Remote & 0.12 & 0.923 & 26.4\\
Remote & 0.20 & 0.899 & 31.2\\
\hline
\end{tabular}
\space
\end{table}
\endgroup 

\FloatBarrier
\Needspace{5\baselineskip}
\subsection[Experiment I-C: output-specific geometry]{{Experiment I-C: output-specific geometry}}
\label{subsec:si_exp_iii}
\begingroup 
The three designs are $X_1=\{i/34:i=0,\ldots,34\}$,
$X_2=\{0.55i/14:i=0,\ldots,14\}$ and
$X_3=\{0.82+0.18i/11:i=0,\ldots,11\}$, with $\ell=0.18$,
\[
\bmLambda=\begin{pmatrix}1&0.85&0.65\\0.85&1&0.80\\0.65&0.80&1\end{pmatrix},
\qquad (\tau_1^2,\tau_2^2,\tau_3^2)=(0.02,0.03,0.03).
\]
Evaluation uses 200 equally weighted points in $[0,1]$. Table~\ref{tab:si_s3_summary} reports output-specific BPI, oracle gain and latent partial $R_j^2=1-1/(\Lambda_{jj}[\bmLambda^{-1}]_{jj})$.

\begin{table}[htbp]\centering
\caption{{I-C output-specific borrowing. Geometry, latent partial $R_j^2$ and oracle gains are deterministic at the stated designs and covariance. Gains average latent variance reductions on the common prediction grid.}}
\label{tab:si_s3_summary}
\begin{tabular}{rrrr}
\hline
Output & $B_j$ & $R_j^2$ & Avg. gain\\
\hline
1 & 0.172 & 0.725 & 0.001\\
2 & 0.585 & 0.829 & 0.251\\
3 & 0.803 & 0.643 & 0.415\\
\hline
\end{tabular}
\space
\end{table}
\endgroup 
\space
\FloatBarrier
\Needspace{6\baselineskip}
\section[Parameter estimation and realised prediction gains]{{Parameter estimation and realised prediction gains}}
\label{sec:si_fitted_experiments}

\subsection[Experiment II-A: matched designs and fitted risk]{{Experiment II-A: matched designs and fitted risk}}
\label{subsec:si_matched_iv}
\begingroup 
Experiment~II-A uses $X_1=\{0.50i/11:i=0,\ldots,11\}$ and $X_3=\{0.03+0.44i/10:i=0,\ldots,10\}$. The within-region, adjacent and remote output-2 designs are $\{0.48i/11\}$, $\{0.50+0.40i/11\}$ and $\{0.90+0.10i/11\}$, respectively, with $i=0,\ldots,11$. The generating model has lengthscale 0.12, unit signal variances, noise variances $(0.10,0.03,0.03)$, and output correlations $\rho_{12}=0.85$, $\rho_{13}=0.45$, $\rho_{23}=0.25$. Scale sensitivities use 0.08 and 0.20 with the other generating parameters fixed. Each geometry/scale uses 500 correlation-only and 300 full-fitting datasets. Assessment uses a 120-point grid in $[0,0.50]$.
Fitting follows Appendix~\ref{subsec:si_common_fitting}. Correlation-only fits hold nuisance values fixed; full fits estimate scale, signal and noise variances with dependence. Table~\ref{tab:e4c_matched} decomposes the scale-0.12 comparisons in main Figure~\ref{fig:s4_tradeoff} into oracle gain and fitting penalties.
\begin{table}[htbp]\centering\scriptsize\setlength{\tabcolsep}{3pt}
\caption{{Experiment~II-A at scale 0.12, with 500 C and 300 F datasets per geometry. Gain, penalties and risk difference are multiplied by $10^3$. Positive conditional latent-risk difference favours joint fitting; intervals use 2,000 paired dataset resamples. C fixes nuisance parameters; F estimates them. The two finite-sample estimates of the expectation identity need not agree exactly.}}
\label{tab:e4c_matched}%

\begin{tabular}{llrrrr}
\toprule
Geometry & Mode & $W_{12}$ & $G_0$ & $H_{\rm mv}-H_{\rm ind}$ & Risk difference [95\% CI] \\
\midrule
Within-region & C & 41.29 & 14.662 & 9.122 & 4.911 [3.942, 5.853] \\
Within-region & F & 41.29 & 14.662 & 1.059 & 12.900 [9.509, 16.306] \\
Adjacent & C & 5.83 & 2.967 & 3.623 & -0.573 [-1.035, -0.149] \\
Adjacent & F & 5.83 & 2.967 & 2.369 & 0.762 [-2.985, 4.293] \\
Remote & C & 0.00 & 2.171 & 3.123 & -0.982 [-1.563, -0.465] \\
Remote & F & 0.00 & 2.171 & 11.825 & -9.950 [-15.599, -4.911] \\
\bottomrule
\end{tabular}
\space
\end{table}\endgroup 

\subsection[Experiment II-B: designs and fitting regimes]{{Experiment II-B: designs and fitting regimes}}
\label{subsec:e5_extension}
\begingroup 
The target grid is $X_1=\{i/15:i=0,\ldots,15\}$. Auxiliary sites are identical (isotopic), $\{(i+1/2)/16:i=0,\ldots,15\}$ (interleaved), or $\{1.25+i/15:i=0,\ldots,15\}$ (remote). Prediction uses 120 points on $[0,1]$.
The separable Mat\'ern~3/2 generator has scale 0.16, unit signal variances and correlation 0.7. The LMC generator is $f_j=a_{j1}u_1+a_{j2}u_2$, with independent unit-variance Mat\'ern~3/2 components of scales 0.09 and 0.34 and loading matrix
\[
A=\begin{pmatrix}0.78&0.44\\0.58&0.28\end{pmatrix}.
\]
Its nearly proportional loadings give zero-lag correlation approximately 0.9980. Both generators add independent Gaussian observation noise of variance 0.01 at each output. Each generator/geometry cell uses 300 independent datasets; correlation-only and full fits share each dataset. The generating covariance, including $10^{-8}$ latent jitter, supplies the conditional assessment law only.

Both fitting modes use the separable working family and exact settings in Appendix~\ref{subsec:si_common_fitting}. Correlation-only fits fix scale 0.16, signal variances 1 and noise variances 0.01; the fixed scale and signal variances are misspecified under the LMC. Full fits estimate all covariance parameters. Each mode uses its own independent baseline; zero means are fixed.
\endgroup 
\begingroup 
\begin{table}[htbp]\centering\scriptsize\setlength{\tabcolsep}{3pt}
\caption{{II-B conditional latent-risk reduction (\%) against each mode's independent baseline. All six cells use 300 datasets, shared between modes. Brackets give 95\% intervals from 2,000 paired within-cell resamples; reductions divide mean risk differences by mean baseline risks.}}
\label{tab:e5_extension_risk}%

\begin{tabular}{llrr}
\toprule
Generator & Geometry & Correlation-only risk reduction & Full-fit risk reduction \\
\midrule
Separable & Isotopic & 1.477 [1.310, 1.642] & 0.352 [-9.694, 9.679] \\
Separable & Interleaved & 6.234 [5.752, 6.683] & 10.278 [2.950, 17.280] \\
Separable & Remote & -0.010 [-0.016, -0.005] & 3.585 [-5.285, 11.812] \\
LMC & Isotopic & 2.561 [2.219, 2.877] & 44.744 [38.587, 50.490] \\
LMC & Interleaved & 23.619 [22.544, 24.663] & 56.258 [52.226, 59.893] \\
LMC & Remote & -0.010 [-0.013, -0.008] & -10.317 [-19.627, -1.619] \\
\bottomrule
\end{tabular}
\space
\end{table}

Under the interleaved LMC generator, full joint fitting improves latent risk but lowers noisy-response MLPD by 13.9079 nats per test point relative to independent fitting (95\% interval 4.0155--23.4644). Its joint 95\% response intervals cover 77.96\% with mean width 0.543, versus 86.51\% and 1.064 for separate fitting. The target noise reaches its $10^{-6}$ lower bound in 163/300 joint and 127/300 independent fits; joint correlation reaches a bound in 228/300 fits. These results show overconfidence alongside improved mean prediction.
All 7,200 optimised fits succeed; 3,788 reach bounds and 2,601 stop above the gradient tolerance. A further 3,600 evaluations use fixed marginals.
\endgroup 
\space
\subsection[Experiment II-C: penalties and refitting]{{Experiment II-C: penalties and refitting}}
\label{sec:revision_penalty_methods}
\begingroup 
Experiment~II-C uses target sites $i/(n-1)$ and auxiliary sites equal to them (isotopic), $(i+1/2)/n$ (interleaved), or $1.25+(i+1/2)/n$ (remote), $i=0,\ldots,n-1$, with 20 observations per output, unit signal variances, noise variances 0.05, and a common Mat\'ern 3/2 lengthscale 0.15. Dependence is zero or 0.7. Each of the six generating scenarios has 500 independent datasets, each fitted in correlation-only and full modes: 3,000 unique datasets and 6,000 scenario/mode evaluations. The two fitting modes use the same generated data and fold assignments at corresponding design/replication indices. The evaluation grid has 201 midpoints in $[0,1]$. The known-nuisance marginal baseline is \texttt{IND-FIXED}; the full-fitting baseline is \texttt{IND-SEPARATE}.

\paragraph*{Fitting and assessment}
Experiment~II-C uses the exact likelihood, bounds, four starts and tolerances in Appendix~\ref{subsec:si_common_fitting}; nuisance values remain fixed only in the correlation controls. Common-site folds round coordinates to $10^{-10}$, independently shuffle unique target sites and remaining unique sites, and assign each group round-robin to four folds. All observations at a site share its fold. Validation withholds all outputs assigned to that fold; the screening procedure retains auxiliaries. Assessment uses the true full-data conditional Gaussian risk solely to evaluate fitted predictions. Realised penalties compare each fit with the oracle using its own observation information set; their expectation identity need not hold dataset by dataset.

\paragraph*{Information conventions}
Expected and observed information follow Eq.~\eqref{eq:information_convention} at fitted parameters, using each model's full observation vector and all estimated parameters in identified log-variance/log-scale/partial-correlation coordinates. Prediction gradients use the same coordinates, and no additional sample-size factor enters the penalty; numerical jitter has zero derivative. Observed information is computed by central differences of analytic gradients with steps $10^{-4}\max(1,|\widehat\theta_u|)$ and half that value. The smaller step is used; relative discrepancy above $10^{-3}$ makes it unavailable. Either information matrix is unavailable unless its minimum eigenvalue exceeds $10^{-10}\max(1,\lambda_{\max})$.
The plug-in penalty averages $a(x,Y)^\top I(\widehat\theta)^{-1}a(x,Y)$ over the prediction grid. The full-fit block calculation uses the nuisance-adjusted dependence gradient; a fixed marginal has zero estimation penalty.

\paragraph*{Same-response parametric refits}
The first 30 Experiment~II-C dataset indices per scenario/mode, selected before observing performance, each receive 200 bootstrap draws: 360 outer scenario/mode cases on 180 generated datasets, and 72,000 refit draws. Draw $Y_b$ from the zero-mean joint observation law at the outer fitted parameter $\widetilde\theta$, without adding simulation jitter. Refit the identical joint and separate marginal algorithms. Then calculate
\[
\widehat H^{\mathrm{boot}}_{\mathrm{mv}}
=\frac1B\sum_{b=1}^B\frac1{|\mathcal A|}\sum_{x\in\mathcal A}
\{g_{\mathrm{mv}}(x;\widehat\theta_b,Y_b)
 -g_{\mathrm{mv}}(x;\widetilde\theta,Y_b)\}^2.
\]
The marginal reference uses the marginal covariance implied by $\widetilde\theta$ and the target component of the same $Y_b$. In the known-nuisance arm, only dependence is refitted. If a reference joint fit is unavailable, its bootstrap is unavailable. Failed inner refits have unavailable penalties; if more than 5\% fail, the bootstrap estimate is treated as unreliable.

\endgroup \begingroup \paragraph*{Post-fit decision and uncertainty}
Let $\widehat G_0$ and $\widehat V_{\rm ind}$ be the oracle gain and marginal variance under the fitted joint reference. An available penalty estimate selects the joint fit when
\begin{equation}
\widehat G_0-(\widehat H_{\mathrm{mv}}-\widehat H_{\mathrm{ind}})
>\epsilon_{\rm ass}(\widehat V_{\rm ind}+\widehat H_{\mathrm{ind}}).
\label{eq:local_net_benefit}
\end{equation}
Here both reference variances defining \(\widehat G_0\) use the same fitted joint covariance; \(\widehat V_{\rm ind}\) uses its marginal block. The independent fitted predictor and its penalty retain their separately estimated parameters. Equation~\eqref{eq:local_net_benefit} is the evaluated post-fit comparator; the operational procedure uses predictive validation.
The post-fit latent-risk diagnostic uses \(\epsilon_{\rm ass}=0.02\); direct validation uses \(\epsilon_{\rm val}=0.02\) on noisy-response MSE. The two denominators differ as specified in main Section~\ref{sec:screening}. At thresholds of 0\%, 1\% and 5\%, validation retains a larger mean latent-risk reduction than both information-based rules on the 3,000 full-fitting datasets, with every paired 95\% interval for these advantages above zero. Parametric-refit comparisons with validation remain inconclusive on the 180-dataset matched subset. An unavailable diagnostic uses the direct-validation decision. True generating risks label outcomes only and do not enter this decision. Bootstrap comparisons use the same first-30 outer subset for every comparator; all methods are evaluated on this matched subset.

\endgroup \begingroup Uncertainty resamples independent datasets within each fixed generating scenario and fitting mode, with 2,000 resamples. Paired method contrasts reuse identical resampling weights. Information calibration compares estimated and realised mean penalties on matching available outer datasets. The 30 outer datasets per combination determine bootstrap-calibration precision; inner refits and grid points are averaged within each outer dataset. Computational costs distinguish full fits/predictions, validation, information diagnostics and bootstrap refits.
\endgroup 
\space
\subsubsection[Penalty calibration and selection]{{Penalty calibration and selection}}\label{subsec:e7_calibration_selection}
\begingroup 
\begin{table}[htbp]\centering\scriptsize\setlength{\tabcolsep}{2pt}
\caption{{Experiment~II-C estimated-minus-realised mean joint penalty, multiplied by $10^3$, with 95\% paired intervals. Each row gives geometry, true correlation and correlation-only (C) or full (F) fitting. Information comparisons use their matched available subset of 500 outer datasets; refits use the prespecified first 30, with 200 draws each. Intervals use 2,000 paired outer-dataset resamples. Positive error means overestimation.}}
\label{tab:e7_calibration}%

\begin{tabular}{lrrr}
\toprule
Scenario & Expected information & Observed information & Refit benchmark \\
\midrule
Iso, 0.0, C & -0.103 [-0.166, -0.045] & -0.056 [-0.122, 0.002] & -0.151 [-0.617, 0.180] \\
Iso, 0.0, F & -2.259 [-4.459, -0.361] & -0.890 [-2.584, 0.637] & 3.695 [1.553, 5.850] \\
Iso, 0.7, C & -0.092 [-0.172, -0.020] & -0.041 [-0.112, 0.028] & -0.019 [-0.345, 0.215] \\
Iso, 0.7, F & -0.841 [-2.190, 0.241] & -0.920 [-2.042, 0.163] & -4.769 [-12.753, 1.520] \\
Int, 0.0, C & -0.113 [-0.176, -0.052] & -0.047 [-0.108, 0.011] & 0.044 [-0.125, 0.206] \\
Int, 0.0, F & -2.610 [-4.872, -0.707] & 0.056 [-2.878, 3.865] & -1.841 [-10.497, 4.837] \\
Int, 0.7, C & -0.109 [-0.191, -0.035] & -0.054 [-0.137, 0.024] & 0.012 [-0.205, 0.220] \\
Int, 0.7, F & -1.863 [-3.990, -0.101] & -1.561 [-3.235, -0.036] & -6.093 [-17.040, 2.163] \\
Remote, 0.0, C & 0.038 [0.033, 0.043] & 0.003 [0.001, 0.005] & 0.001 [0.000, 0.002] \\
Remote, 0.0, F & -0.260 [-1.972, 1.157] & -0.056 [-1.771, 1.735] & 2.676 [0.655, 4.469] \\
Remote, 0.7, C & 0.038 [0.033, 0.044] & 0.002 [0.000, 0.004] & 0.003 [0.003, 0.003] \\
Remote, 0.7, F & -0.851 [-2.866, 0.922] & -0.380 [-1.988, 1.080] & -1.840 [-9.791, 3.717] \\
\bottomrule
\end{tabular}
\space
\end{table}
The null/interleaved correlation-only control has zero oracle gain and realised mean penalty 0.000498; its conditional risk difference is $-0.000498$ (95\% interval $-0.000558$ to $-0.000436$). Expected information estimates 0.000385 in that case, but overstates the small remote-null mean penalty by a factor of 25.9. Table~\ref{tab:e7_calibration} shows both directions of error across the complete grid.

\endgroup \begingroup \begin{table}[htbp]\centering\scriptsize\setlength{\tabcolsep}{3pt}
\caption{{Experiment~II-C decisions at 2\% materiality. Each mode pools six configurations, with 500 or 30 outer datasets per configuration as labelled. Positive risk reduction favours selection over the mode-matched independent fit. FPR/FNR condition on non-beneficial/beneficial joint fits. Intervals use 2,000 paired outer-dataset resamples within configuration. Unavailable penalties fall back to validation. C fixes nuisance parameters; F estimates all covariance parameters.}}
\label{tab:e7_decisions}%

\begin{tabular}{llrrrr}
\toprule
Mode/subset & Selection & Reduction (\%) [95\% CI] & FPR (\%) & FNR (\%) & Fallback \\
\midrule
C/500 & Validate & 0.65 [0.57, 0.73] & 21.19 & 46.88 & 0/3000 \\
C/500 & Expected & 1.11 [1.02, 1.20] & 11.80 & 15.21 & 0/3000 \\
C/500 & Observed & 1.10 [1.01, 1.19] & 11.56 & 16.00 & 0/3000 \\
C/30 & Validate & 0.48 [0.13, 0.82] & 19.38 & 54.90 & 0/180 \\
C/30 & Expected & 1.07 [0.72, 1.43] & 13.18 & 15.69 & 0/180 \\
C/30 & Observed & 1.09 [0.73, 1.45] & 12.40 & 15.69 & 0/180 \\
C/30 & Bootstrap & 1.00 [0.64, 1.37] & 12.40 & 21.57 & 0/180 \\
F/500 & Validate & 1.39 [0.34, 2.44] & 44.87 & 54.77 & 0/3000 \\
F/500 & Expected & -0.31 [-1.42, 0.79] & 65.44 & 34.57 & 1113/3000 \\
F/500 & Observed & -1.46 [-2.51, -0.42] & 51.31 & 47.66 & 11/3000 \\
F/30 & Validate & 0.31 [-5.86, 6.15] & 41.24 & 57.83 & 0/180 \\
F/30 & Expected & -0.78 [-7.18, 5.11] & 65.98 & 36.14 & 69/180 \\
F/30 & Observed & -0.35 [-6.40, 5.19] & 47.42 & 51.81 & 0/180 \\
F/30 & Bootstrap & 1.30 [-5.62, 7.88] & 51.55 & 28.92 & 0/180 \\
\bottomrule
\end{tabular}
\space
\end{table}
Table~\ref{tab:e7_decisions} reports both fitting modes and the matched refit subset; main Table~\ref{tab:e7_selection_summary} summarises paired full-fitting contrasts with validation.
All outer fits and 72,000 refit draws completed without failed fits; 34,552 refit draws reached a joint or marginal parameter bound. Expected joint information was unavailable on 817 outer cases. Decision fallbacks in Table~\ref{tab:e7_decisions} additionally include unavailable marginal penalties. Full-fitting refits took 5.92--6.52 seconds per outer case, versus 0.109--0.126 seconds for four-fold validation, under matched single-threaded settings.

\endgroup 
\space
\space
\FloatBarrier
\section[Sharing parameters and screening response coupling]{{Sharing parameters and screening response coupling}}
\label{sec:task_screen_methods}
\begingroup 
Experiment~III evaluates main Section~\ref{sec:screening} using SEP, DIAG and JOINT under the common loss in Appendix~\ref{subsec:si_simulation}.

\endgroup \begingroup \subsection{Design, fitting and observation availability}
Target sites are \(X_1=\{i/(n-1):i=0,\ldots,n-1\}\),
\(n\in\{12,24\}\). Each auxiliary design has \(4n\) observations:
four independent noisy measurements at each target site (isotopic);
the \(4n\) midpoints of \([0,1]\) (infill);
the \(4n\) midpoints of \([0.5,1.5]\) (partial coverage);
or the \(4n\) midpoints of \([2,3]\) (remote).
Isotopic auxiliary replicates share one latent value per site and have independent measurement noise.
The prediction region is \([0,1]\), represented by 201 equally weighted midpoints.

The generating law has zero means, unit signal variances, a common
Mat\'ern 3/2 lengthscale of 0.15, correlation
\(\rho\in\{0,0.4,0.8\}\), target noise variance
\(\tau_1^2\in\{0.01,0.1\}\), and auxiliary noise variance 0.01.
The full factorial grid contains 48 scenarios. Each uses 100 calibration
datasets and 200 independent evaluation datasets, for 4,800 and 9,600 datasets.
All covariance parameters are estimated; only zero means are fixed.
The exact float64 likelihood, parameter bounds, jitter ladder and stopping
tolerances are specified in Appendix~\ref{subsec:si_common_fitting}. The diagonal fit minimises
the sum of the two marginal negative log likelihoods using separate block
factorisations and a common scale. Its four starts use the separate marginal
signal estimates and the observation-weighted geometric mean of their scales,
with multipliers \((0.5,1,1,2)\). The first retains marginal noise estimates;
the other three use noise equal to \((0.03,0.1,0.3)\) times the corresponding
signal estimate. The coupled model uses the diagonal solution to initialise
its four complete starts, retaining the correlation, scale-multiplier and
noise-start conventions of the exact fitter. Failed diagonal fits use separate
initialisation.

Unique target sites, identified using coordinates rounded to \(10^{-10}\), are balanced across four folds without using responses. All auxiliary observations remain
available. The separate auxiliary fit is therefore reused across folds in
every comparator. Target and diagonal parameters are refitted within each
fold; full-data target estimates never initialise fold fits. Selection and
failure fallbacks are those in main Section~\ref{subsec:screen_validation}.

\endgroup \begingroup \subsection{Exact conditional opportunity and its envelope}
\label{subsec:task_gain_envelope}
For one scale scenario, let \(K_{ab}\) denote the common unit-diagonal
kernel blocks, \(v_i\) the diagonal signal variances, and \(\eta_i\)
the observation-noise variances. Define
\[
A=K_{11}+\eta_1/v_1\,I,\qquad
D=K_{qq}+\eta_q/v_q\,I,\qquad
B=K_{q1}A^{-1}K_{1q},
\]
where \(k_i(x)\) is the vector of kernel covariances between design \(i\) and \(x\), and set \(r(x)=k_q(x)-K_{q1}A^{-1}k_1(x)\).
At squared correlation \(t=\rho^2\), the conditional auxiliary covariance is \(v_q(D-tB)\), and its conditional cross-covariance with the latent target is \(\rho\sqrt{v_1v_q}\,r(x)\). Main Theorem~\ref{thm:exact_oracle_gain} therefore gives
\begin{equation}
\Delta(x;t)=v_1t\,r(x)^\top(D-tB)^{-1}r(x).
\label{eq:task_exact_gain}
\end{equation}
The marginal variance is
\(V_D(x)=v_1\{1-k_1(x)^\top A^{-1}k_1(x)\}\).
All operations use Cholesky solves; no coupled likelihood optimisation is needed.

\begin{proposition}[Conditional gain envelope]
Suppose\label{prop:task_gain_envelope} the common kernel is positive semidefinite, the signal and noise
variances are positive, and \(0\leq t\leq t_{\max}<1\).
Then \(0\leq\Delta(x;t)\leq\Delta(x;t_{\max})\leq V_D(x)\).
The ratio of integrated gain at \(t_{\max}\) to integrated marginal variance
bounds the corresponding ratio for every smaller \(t\).
\end{proposition}
\begin{proof}
\leavevmode\phantomsection\label{proof:task_gain_envelope}%
The joint observation covariance is positive definite, so its auxiliary
Schur complement \(v_q(D-tB)\) is positive definite. Writing
\(S=D-tB\), differentiation gives
\[
\frac{\mathrm d}{\mathrm dt}\{t(D-tB)^{-1}\}
=S^{-1}+tS^{-1}BS^{-1}=S^{-1}DS^{-1}\succeq0.
\]
The gain is zero at independence and is monotone in \(t\).
The final inequality is the nonnegativity of the joint conditional variance.
Summing the pointwise inequalities and dividing by the same positive sum of
marginal variances proves the integrated statement.
\end{proof}

For fixed positive diagonal observation-noise matrices $\bmOmega_1,\bmOmega_q$, replace
$\eta_1 I/v_1$ and $\eta_q I/v_q$ by $\bmOmega_1/v_1$ and $\bmOmega_q/v_q$.
The same Schur-complement expression and proof apply because $D$ remains
positive definite. TQA uses this form with training-estimated simulation
variances plus fitted nuggets, conditional on fitted mean functions.
Use \(t_{\max}=0.995^2\), and take the largest integrated relative gain
over scales \(0.5\widehat{\boldsymbol\ell}\),
\(\widehat{\boldsymbol\ell}\), and
\(2\widehat{\boldsymbol\ell}\) from the corresponding diagonal fit.
The maximum covers these three working covariance scenarios; calibration
assesses its use with fitted parameters. Finite boundary solutions remain
included. A mean marginal variance below \(10^{-12}v_1\),
a failed fit or solve, non-finite score, escalated jitter, or fewer than two
distinct sites in either design returns \texttt{unresolved}, retaining coupling.
The loose-score control substitutes
\(\min\{V_D(x),v_1v_qt_{\max}\|r(x)\|^2/\eta_q\}\)
for the exact gain, using \(v_q(D-tB)\succeq\eta_q I\).
The denominator and covariance scenarios are unchanged.
It isolates the exclusion effect of the noise-based bound.

\subsection{Calibration and independent assessment}
\label{subsec:task_screen_calibration}
Search
\(t_U\in\{0,0.0025,0.005,0.01,0.02,0.05,0.1,0.2\}\).
For each threshold, evaluate the selected predictors and sum the measured costs
of the required fitting, scoring and prediction operations, accounting for
full-data exclusion and reuse of the auxiliary fit.
Choose the minimum mean application component time subject to one-sided
95\% bootstrap upper bounds of: 1\% excess selected risk versus direct
three-candidate validation overall; 1\% positive exclusion regret overall;
and 2\% excess selected risk versus direct within each geometry.
Each fraction uses that population's mean \texttt{SEP} risk.
Positive exclusion regret is
\[
\mathbb E\!\left[\mathds1\{\mathrm{exclude}\}
\max\{\min(L^*_{\rm SEP},L^*_{\rm DIAG})-L^*_{\rm JOINT},0\}\right]
\big/\mathbb E[L^*_{\rm SEP}].
\]
These are empirical pointwise percentile criteria. Resample independent
datasets within every scenario 2,000 times and break exact cost ties with
the smaller threshold. The exact and loose gates are calibrated separately;
all eight thresholds satisfy their respective limits, and both select 0.20,
the largest tested value. For the exact score, estimated mean component times
range from 0.339 to 0.341\,s across non-zero thresholds
(Table~\ref{tab:e6_task_calibration}), showing little cost variation within
the tested range. The selected threshold minimises estimated cost over the
prespecified grid; thresholds above 0.20 were not assessed.
The zero threshold uses direct
validation without score computation, and its calibration cost includes only
the required fitting and prediction operations. The separate retain-all
control computes scores while suppressing exclusions.

Calibration and independent assessment use the generating-law loss $L_m^*(Y)$ from Appendix~\ref{subsec:si_simulation}.
Coupling is materially beneficial if
\(L^*_{\rm JOINT}<0.98\min(L^*_{\rm SEP},L^*_{\rm DIAG})\).
Candidate false exclusion and final coupling FPR/FNR use this reference;
regret compares the final risk with the smallest of all three fitted risks.
Paired evaluation intervals use 2,000 within-scenario dataset resamples;
evaluation points are not independent replicates.
Paired timings use the first five evaluation datasets per scenario (240 datasets), alternating execution order between screened and direct three-candidate validation. \endgroup 
\space
\subsection{Evaluation, controls and computational value}
\label{subsec:task_screen_results}
\begingroup 
The screening score was successfully computed for every full training dataset.
Table~\ref{tab:e6_task_errors} separates candidate exclusions from final
coupling errors. Mean positive loss
attributable to exclusion is 0.13\% of SEP risk (95\% paired interval 0.10--0.18\%).

\begin{table}[htbp]\centering\scriptsize\setlength{\tabcolsep}{3pt}
\caption[Screening and final coupling errors]{{Experiment~III:
candidate exclusion and final decision rates (\%), with 95\% intervals from
2,000 paired within-scenario resamples of 9,600 independent datasets.
False exclusion and FNR use the 2,556 beneficial coupled fits; FPR uses the
7,044 non-beneficial fits. Materiality is a 2\% conditional latent-risk
improvement on both independent candidates. All rules use the same fitted three-model catalogue.}}
\label{tab:e6_task_errors}
{%
\begin{tabular}{@{}lrrrr@{}}
\toprule
Rule & Exclude (\%) & False exclude (\%) & FPR (\%) & FNR (\%) \\
\midrule
Direct validation (three) & 0.00 [0.00, 0.00] & 0.00 [0.00, 0.00] & 22.71 [21.87, 23.60] & 37.56 [35.72, 39.44] \\
Screened validation (three) & 33.01 [32.59, 33.44] & 4.26 [3.52, 4.99] & 18.55 [17.69, 19.39] & 40.38 [38.56, 42.21] \\
Loose bound & 32.04 [31.59, 32.50] & 3.87 [3.16, 4.58] & 18.61 [17.75, 19.45] & 40.22 [38.43, 42.06] \\
Retain boundaries & 17.58 [17.15, 18.05] & 0.82 [0.50, 1.19] & 22.52 [21.65, 23.39] & 37.87 [36.03, 39.75] \\
\bottomrule
\end{tabular}
}
\end{table}

\begin{table}[htbp]\centering\scriptsize\setlength{\tabcolsep}{3pt}
\caption[Risk and cost controls]{{Experiment~III rule controls on
all 9,600 datasets. Excess risk relative to direct three-candidate validation
and three-model regret are percentages of mean SEP risk; positive excess
means worse prediction. Regret uses the lowest of the three conditional fitted
risks. Component savings include every required fitted/score/prediction block
and auxiliary-fit reuse, with positive values favouring the rule.
Joint fits are mean application counts per dataset. Intervals use 2,000 paired
within-scenario dataset resamples. Actual elapsed time is measured separately
in Table~\ref{tab:e6_task_timing}.}}
\label{tab:e6_task_ablations}
{%
\begin{tabular}{@{}lrrrr@{}}
\toprule
Rule & Excess risk (\%) & Component saving (\%) & Joint fits & Regret (\%) \\
\midrule
Direct validation (three) & 0.00 [0.00, 0.00] & 0.00 [0.00, 0.00] & 4.33 & 14.15 [13.22, 15.13] \\
Screened validation (three) & -0.05 [-0.25, 0.16] & 10.91 [10.70, 11.12] & 2.96 & 14.10 [13.17, 15.08] \\
Loose bound & -0.03 [-0.23, 0.17] & 10.40 [10.19, 10.61] & 3.00 & 14.12 [13.18, 15.09] \\
Retain boundaries & -0.04 [-0.13, 0.02] & 5.38 [5.21, 5.56] & 3.55 & 14.11 [13.17, 15.10] \\
\bottomrule
\end{tabular}
}
\end{table}

The retain-all control computes the score but suppresses exclusions, unlike
operational \(t_U=0\), which bypasses scoring. It reproduces direct validation's
predictions and decisions while increasing component cost by 2.80\%
(95\% interval 2.79--2.81\%). Disabling only the full-data shortcut, while
retaining fold-specific screening, reproduces screened validation's predictions
and decisions; it requires 3.28 joint fits per dataset and saves 7.38\%
(7.30--7.46\%) in component cost relative to direct validation.
Dense marginal factorisations cost $O(N_1^3+N_q^3)$; conditional-gain solves
also depend on the prediction-grid size. Complete-procedure timings include
\texttt{DIAG} fits, scoring, restarts, validation and final prediction.
Loading, offline calibration and assessment-only fits are excluded.

\begin{table}[htbp]\centering\scriptsize\setlength{\tabcolsep}{3pt}
\caption[Independent threshold calibration]{{Experiment~III exact
score calibration on 4,800 datasets, all eight candidate thresholds.
Risk and positive-exclusion-regret columns are one-sided 95\% bootstrap upper
bounds, as percentages of mean SEP risk. The additional geometry-specific
risk constraints are described in Section~\ref{subsec:task_screen_calibration}.
Seconds are mean application
component times; the zero-threshold row uses direct validation without score
computation. The separate retain-all control includes scoring. Smaller
thresholds break exact ties. Both the exact and
separately calibrated loose-score rule select 0.20.}}
\label{tab:e6_task_calibration}
{%
\begin{tabular}{@{}rrrrrrc@{}}
\toprule
$t_U$ & Excluded & False excl. & Risk upper (\%) & Regret upper (\%) & Seconds & Admissible \\
\midrule
0 & 0 & 0 & 0.000 & 0.000 & 0.382 & Yes \\
0.0025 & 1573 & 44 & 0.262 & 0.124 & 0.341 & Yes \\
0.005 & 1575 & 44 & 0.262 & 0.124 & 0.341 & Yes \\
0.01 & 1577 & 45 & 0.262 & 0.140 & 0.341 & Yes \\
0.02 & 1578 & 45 & 0.262 & 0.140 & 0.341 & Yes \\
0.05 & 1582 & 48 & 0.263 & 0.141 & 0.340 & Yes \\
0.1 & 1585 & 49 & 0.263 & 0.141 & 0.340 & Yes \\
0.2 & 1595 & 50 & 0.261 & 0.142 & 0.339 & Yes \\
\bottomrule
\end{tabular}
}
\end{table}

\begin{table}[htbp]\centering\footnotesize\setlength{\tabcolsep}{4pt}
\caption[Actual complete-procedure timing]{{Experiment~III actual
runtime on 240 paired datasets, the first five evaluation replications in each
scenario. Screened and direct validation use the same three candidates;
execution order alternates by dataset. Values are
means [95\% paired within-scenario bootstrap intervals]. Positive saving
favours the screen. Runs use an Intel
Core Ultra 7 265K. Loading and assessment are excluded; all fitting, scores and
predictions are included.}}
\label{tab:e6_task_timing}
{%
\begin{tabular}{@{}lrrr@{}}
\toprule
Comparator & Direct seconds & Screen seconds & Screen saving (\%) \\
\midrule
Three-candidate direct & 0.441 [0.434, 0.448] & 0.379 [0.371, 0.387] & 14.01 [12.55, 15.67] \\
\bottomrule
\end{tabular}
}
\end{table}

\endgroup \begingroup For Experiment~III, the optimiser reported convergence with a finite objective value for all 153,600 primary full-data and fold fits.
There are 58,391 fits with an active parameter bound and 89,588 whose projected
gradient exceeds \(10^{-5}\) after objective-tolerance convergence. Every start
succeeds and no jitter escalation occurs.

\begin{table}[htbp]\centering\footnotesize\setlength{\tabcolsep}{4pt}
\caption[Predictive calibration after selection]{{Experiment~III
nominal 95\% response-interval coverage, mean width and mean Gaussian log
predictive density, with 95\% paired dataset-resampling intervals.
Values pool all 9,600 datasets across the 48 scenarios. Test measurement noise
is independent of the training data. Predictive variances use fitted covariance
parameters and omit their estimation uncertainty. All fits and boundary
solutions are included, including extreme log-score losses.}}
\label{tab:e6_task_predictive_calibration}
{%
\begin{tabular}{@{}lrrr@{}}
\toprule
Procedure & Coverage (\%) & Width & Mean log score \\
\midrule
Always SEP & 80.80 [80.39, 81.22] & 1.08 [1.07, 1.09] & -38.05 [-40.16, -35.99] \\
Always DIAG & 82.54 [82.14, 82.91] & 1.11 [1.09, 1.12] & -29.50 [-31.32, -27.65] \\
Always JOINT & 82.58 [82.20, 82.95] & 1.08 [1.07, 1.09] & -27.93 [-29.67, -26.19] \\
Direct validation (three) & 81.08 [80.70, 81.48] & 1.04 [1.03, 1.05] & -34.43 [-36.39, -32.43] \\
Screened validation (three) & 80.79 [80.40, 81.19] & 1.04 [1.03, 1.05] & -35.24 [-37.30, -33.27] \\
\bottomrule
\end{tabular}
}
\end{table}

\endgroup 
\space
\FloatBarrier
\Needspace{6\baselineskip}
\subsection{Geometry-specific results and scale sensitivity}
\label{subsec:screening_implementation}
\begingroup 
Gains and selection frequencies by geometry are reported in Table~\ref{tab:e6_task_geometry_summary}.

\begin{table}[htbp]\centering\scriptsize\setlength{\tabcolsep}{2pt}
\caption[Risk reductions and selections by geometry]{{Experiment~III:
conditional latent-risk reductions relative to \texttt{SEP} and final
selection frequencies for screened validation. Each geometry pools 2,400
datasets; the final row pools all 9,600. \texttt{DIAG} and \texttt{JOINT}
always use the indicated candidate; direct and screened validation select
among all three candidates. Reductions are ratios of mean losses, with
95\% intervals from 2,000 paired within-scenario resamples. Selection
frequencies use all datasets in the row.}}
\label{tab:e6_task_geometry_summary}
{%
\begin{tabular}{@{}lrrrrrrr@{}}
\toprule
& \multicolumn{4}{c}{Risk reduction vs \texttt{SEP} (\%)} & \multicolumn{3}{c}{Screened selection (\%)} \\
\cmidrule(lr){2-5}\cmidrule(l){6-8}
Geometry & \texttt{DIAG} & \texttt{JOINT} & Direct & Screened & \texttt{SEP} & \texttt{DIAG} & \texttt{JOINT} \\
\midrule
Isotopic & 7.21 [3.53, 10.63] & 5.66 [1.87, 9.19] & 8.36 [5.35, 11.35] & 8.53 [5.59, 11.45] & 34.3 & 34.8 & 30.9 \\
Infill & 5.43 [2.18, 8.58] & 7.26 [4.07, 10.29] & 7.99 [5.73, 10.09] & 7.99 [5.73, 10.09] & 28.5 & 25.0 & 46.4 \\
Partial coverage & 6.61 [3.54, 9.53] & 5.18 [1.81, 8.33] & 5.45 [2.94, 7.84] & 5.46 [2.95, 7.85] & 31.6 & 27.8 & 40.7 \\
Remote & 4.52 [1.01, 7.88] & 4.39 [0.84, 7.79] & 4.68 [2.58, 6.64] & 4.68 [2.58, 6.64] & 49.8 & 50.2 & 0.0 \\
\midrule
All geometries & 5.95 [4.27, 7.60] & 5.62 [3.93, 7.24] & 6.62 [5.33, 7.82] & 6.66 [5.37, 7.86] & 36.1 & 34.5 & 29.5 \\
\bottomrule
\end{tabular}
}
\end{table}

\endgroup \begingroup \begin{remark}[Anisotropic scale sensitivity]
\label{rem:anisotropic_scales}
Under Proposition~\ref{prop:anisotropic_mass_corners}, raw interaction mass
attains its extrema at the all-lower and all-upper corners of a scale rectangle.
Whitening and conditional residualisation also change with scale, so the result
does not establish extrema of correlation information or conditional opportunity.
Appendix~\ref{sec:e9_scales} reports finite-scenario checks; these do not certify
continuous-box extrema.
\end{remark}
\endgroup 
\FloatBarrier
\space
\FloatBarrier
\section[Borrowing beyond a common kernel]{{Borrowing beyond a common kernel}}
\label{sec:si_kernel_extensions}
\subsection[Richer-family borrowing: Experiment~IV-A]{{Richer-family borrowing: Experiment~IV-A}}
\label{sec:e8_richer}
\begingroup 
Experiment~IV-A uses isotopic, interleaved
and remote designs with 20 or 40 observations per output.
Target sites are equally spaced in $[0,1]$; auxiliary sites equal them
(isotopic), or form midpoints of $[0,1]$ (interleaved) or $[1.25,2.25]$
(remote). Assessment uses 201 fixed midpoints of $[0,1]$.

The positive-LMC loading rows are proportional to $(1,0.25)$ and $(0.2,1)$, each scaled to unit Euclidean norm. The cancelling rows are $(\sqrt{0.8},\sqrt{0.2})$ and $(\sqrt{0.2},-\sqrt{0.8})$. Independent component kernels have Mat\'ern 3/2 scales $0.06$ and $0.35$, giving unit marginal variances. In the cancelling case the cross-covariance is $0.4\{k_{0.06}(h)-k_{0.35}(h)\}$: zero at $h=0$, generally nonzero elsewhere. The fidelity target is $0.7u_L+\sqrt{0.51}u_\delta$ and the auxiliary is $u_L$, with independent unit-variance components of scales $0.30$ and $0.06$. Gaussian noise variances are $0.05$ at both outputs.

The richer-model comparisons use 200 independent datasets per
configuration, generated separately from Experiment~V-A. They compare joint
fitting with matching independent covariance families and direct two-candidate
validation; no separable screen is applied to these richer covariances.
Their four common-site folds follow Appendix~\ref{sec:revision_penalty_methods},
withholding all outputs assigned to each fold.

Richer LMC fits have two rank-one components, two free scales and separate noises; their independent comparator fits two-kernel univariate mixtures. The AR fit estimates its coefficient, low-fidelity variance, discrepancy variance, two scales and two noises; its target-independent comparator is a two-kernel mixture.

All base kernels have unit amplitude. Rich models use exact float64 likelihoods, analytic covariance derivatives, four starts, and the stopping tolerances and Cholesky ladder in Appendix~\ref{subsec:si_common_fitting}. LMC target loadings are positive in $[0.001,5]$, auxiliary loadings signed in $[-5,5]$; AR coefficients are in $[-2,2]$. Scales, noises, AR component variances and independent-mixture variances use the common bounds.

For starts, let $v_j$ be the training second moment clipped to $[0.05,5]$. The short/long scale pairs are $(0.04,0.2),(0.08,0.4),(0.15,0.8),(0.3,1.5)$. LMC loading magnitudes are $\sqrt{v_j/2}$, with auxiliary sign pairs $(-,+),(+,-),(+,+),(+,+)$. AR starts use low-fidelity variance $v_q$ and the long scale, discrepancy variance $v_1/2$ and the short scale, and coefficients $-0.7,-0.2,0.2,0.7$; independent mixtures split variance equally. Noise starts are $v_j$ times $0.03,0.1,0.3,0.1$, clipped to $[0.001,0.5]$. Complete starts are clipped inside parameter bounds. Covariance parameters are re-estimated within every validation fold, with zero means fixed. A failed joint fit uses its matching independent comparator.

\begin{table}[htbp]\centering\scriptsize\setlength{\tabcolsep}{3pt}
\caption{{Experiment~IV-A: richer-family validation at a 2\%
response-loss improvement threshold. Entries are percentage reductions in
mean conditional latent risk against the matching fitted independent covariance
model, with pointwise 95\% paired intervals. Each of the 18 configurations has
200 independent datasets; family summaries pool all six geometry/count
configurations (1,200 datasets). Intervals use 2,000 paired resamples within
configuration, stratified by configuration for family summaries. Two-candidate
validation uses held-out responses under common-site folds; the generating
Gaussian law supplies assessment risk only.}}
\label{tab:e8_richer_families}

\begin{tabular}{llcr}
\toprule
Family & Geometry & $N_1=N_q$ & Risk reduction (\%) [95\% interval] \\
\midrule
LMC positive & Isotopic & 20 & -3.39 [-9.50, 2.06] \\
LMC positive & Isotopic & 40 & -0.89 [-2.61, 0.41] \\
LMC positive & Interleaved & 20 & -5.83 [-14.03, 1.31] \\
LMC positive & Interleaved & 40 & 0.79 [-0.20, 1.71] \\
LMC positive & Remote & 20 & -5.26 [-12.07, -0.62] \\
LMC positive & Remote & 40 & -0.11 [-0.48, 0.20] \\
\midrule
LMC cancelling & Isotopic & 20 & 1.83 [-2.30, 6.18] \\
LMC cancelling & Isotopic & 40 & 2.26 [0.59, 3.76] \\
LMC cancelling & Interleaved & 20 & 5.74 [1.47, 9.93] \\
LMC cancelling & Interleaved & 40 & 5.42 [3.60, 7.10] \\
LMC cancelling & Remote & 20 & -7.18 [-15.11, -1.60] \\
LMC cancelling & Remote & 40 & -0.22 [-0.71, 0.17] \\
\midrule
AR fidelity & Isotopic & 20 & -4.86 [-9.87, -0.46] \\
AR fidelity & Isotopic & 40 & 0.52 [-0.33, 1.29] \\
AR fidelity & Interleaved & 20 & -2.88 [-6.43, 0.35] \\
AR fidelity & Interleaved & 40 & -0.34 [-1.65, 0.75] \\
AR fidelity & Remote & 20 & -1.53 [-3.29, -0.10] \\
AR fidelity & Remote & 40 & -0.95 [-2.91, 0.13] \\
\midrule
\multicolumn{4}{l}{\textit{Family summaries: all geometries and counts}} \\
\multicolumn{3}{l}{LMC positive} & -3.54 [-6.33, -0.87] \\
\multicolumn{3}{l}{LMC cancelling} & 0.73 [-1.70, 2.81] \\
\multicolumn{3}{l}{AR fidelity} & -2.24 [-3.73, -0.82] \\
\bottomrule
\end{tabular}
\space
\end{table}
\endgroup 
\Needspace{5\baselineskip}
\space
\subsection[Target scarcity and auxiliary allocation: Experiment~IV-B]{{Target scarcity and auxiliary allocation: Experiment~IV-B}}
\label{subsec:e8_opportunity}
\begingroup 
Experiment~IV-B uses IV-A's positive nonproportional LMC and
autoregressive generating covariances and noise variance 0.05
(Appendix~\ref{sec:e8_richer}). It crosses
$N_1\in\{8,16\}$, $N_q/N_1\in\{1,4\}$ and interleaved or remote geometry,
using equally spaced target sites in $[0,1]$ and auxiliary midpoints in
$[0,1]$ or $[1.25,2.25]$, respectively. The prediction grid has the same 201 midpoints in
$[0,1]$. Evaluation uses 500 independent Gaussian draws in each of the 16
configurations (8,000 datasets), disjoint from pilot simulations.
The richer joint models and their matching independent two-kernel mixtures
use the covariance parameterisations, four complete starts, bounds and
tolerances in Appendix~\ref{sec:e8_richer}. Means are fixed at zero.
All fitted predictions
are assessed against the generating conditional Gaussian law; true covariance
parameters enter assessment only.

\begin{table}[htbp]\centering\scriptsize\setlength{\tabcolsep}{3pt}
\caption{{Experiment~IV-B. Oracle gain is the percentage reduction
in mean latent variance relative to the true marginal predictor. Fitted
reductions use mean conditional latent risk relative to each model's matching
fitted independent family; brackets give pointwise 95\% intervals from 2,000
paired dataset resamples within cell. Oracle and fitted percentages have
different denominators.}}
\label{tab:e8_opportunity_full}

\begin{tabular}{lrrlrr}
\toprule
Family & $N_1$ & $N_q$ & Geometry & Oracle (\%) & Fitted / matching independent \\
\midrule
LMC & 8 & 8 & Interleaved & 7.783 & 13.01 [10.36, 15.54] \\
LMC & 8 & 8 & Remote & 0.019 & -7.97 [-9.86, -6.21] \\
LMC & 8 & 32 & Interleaved & 24.614 & 10.60 [7.12, 14.00] \\
LMC & 8 & 32 & Remote & 0.027 & -8.52 [-10.80, -6.41] \\
LMC & 16 & 16 & Interleaved & 4.894 & -5.55 [-11.90, 0.29] \\
LMC & 16 & 16 & Remote & 0.002 & -15.99 [-20.40, -11.60] \\
LMC & 16 & 64 & Interleaved & 15.046 & 12.32 [6.23, 18.15] \\
LMC & 16 & 64 & Remote & 0.002 & -19.26 [-25.45, -13.73] \\
AR & 8 & 8 & Interleaved & 5.489 & 2.23 [-0.20, 4.80] \\
AR & 8 & 8 & Remote & 0.016 & -1.28 [-2.40, -0.11] \\
AR & 8 & 32 & Interleaved & 6.419 & 6.11 [3.55, 8.80] \\
AR & 8 & 32 & Remote & 0.027 & 2.49 [1.31, 3.67] \\
AR & 16 & 16 & Interleaved & 0.760 & -13.71 [-18.67, -9.09] \\
AR & 16 & 16 & Remote & 0.003 & -5.60 [-9.13, -2.81] \\
AR & 16 & 64 & Interleaved & 0.899 & -6.61 [-10.37, -2.88] \\
AR & 16 & 64 & Remote & 0.004 & -6.76 [-10.04, -3.73] \\
\bottomrule
\end{tabular}
\space
\end{table}

At $N_1=8,N_q=32$, remote AR data reduce fitted risk by
2.49\% (1.31--3.67\%) despite an oracle gain of only 0.027\%. The joint fit has a smaller
mean estimation penalty than its matching independent fit, consistent with
the parameter-sharing mechanism.
All richer and matching independent fits return solver success; finite parameter-bound and above-gradient-tolerance outcomes remain in every risk summary. No start fails and no jitter escalates.

\endgroup 
\space
\space
\FloatBarrier
\section[Robustness and sensitivity]{{Robustness and sensitivity}}
\label{sec:si_robustness}
\subsection{Transfer protocol and computational cost}
\label{subsec:section7_transfer}
\begingroup 
Experiments~V-A/V-B apply main Section~\ref{sec:screening} with Experiment~III's calibrated threshold $t_U=0.20$. They retain SEP/DIAG/JOINT, correlation envelope 0.995, scale multipliers $0.5,1,2$, four target-response folds retaining auxiliaries, and $\epsilon_{\rm val}=\epsilon_{\rm ass}=0.02$. Each fold refits target and DIAG parameters. Full-data exclusion skips coupled validation fits; the auxiliary marginal is reused. Numerical retention, starts and fallbacks follow Appendix~\ref{sec:task_screen_methods}.

Independent evaluation streams are separate from pilots and calibration. Gaussian cases use full-data conditional latent risk; Student-$t$ cases use realised squared latent error. Labels compare JOINT with both independent candidates; intervals use 2,000 paired resamples within configuration.
Complete executions use the first two evaluation datasets per configuration and kernel family, alternating screen/direct order on an Intel Core Ultra 7 265K. Timing includes every fit, score, validation and final prediction, excluding loading and assessment. TQA's calibration, noise-aware fits and independent timings are specified in Appendix~\ref{subsec:q2_calibrated}.
\begin{table}[htbp]\centering\scriptsize\setlength{\tabcolsep}{3pt}
\caption{{Section~7 transfer: paired complete executions on the first
two evaluation datasets per configuration and kernel family. Saving is the percentage
reduction in mean elapsed time relative to direct three-candidate validation;
negative values mean additional cost. Brackets give 95\% paired bootstrap
intervals.}}
\label{tab:section7_transfer_timing}

\begin{tabular}{@{}lrrrr@{}}
\toprule
Study / kernel family & Pairs & Screen (s) & Direct (s) & Saving (\%) \\
\midrule
V-A & 72 & 0.164 & 0.206 & 20.41 [17.74, 23.12] \\
V-B / ARD & 72 & 0.944 & 1.094 & 13.75 [10.70, 16.72] \\
V-B / isotropic & 72 & 0.672 & 0.769 & 12.55 [11.15, 13.99] \\
\bottomrule
\end{tabular}
\space
\end{table}

The 216 paired executions cover 144 datasets.

\begin{table}[htbp]\centering\scriptsize\setlength{\tabcolsep}{3pt}
\caption{{Full-evaluation numerical diagnostics. Fits include every
candidate needed for assessment, including those a screened execution skips.
``Bound'' counts finite parameter-bound solutions; the next column counts
projected gradients above the stopping tolerance after objective convergence.
``Unresolved'' counts unusable full-data screening scores; in these cases,
the coupled model remains eligible for validation.
All such outcomes remain in the risk and decision summaries.}}
\label{tab:section7_transfer_numerics}

\begin{tabular}{@{}lrrrrr@{}}
\toprule
Study / kernel family & Fits & Failed & Bound & $g>10^{-5}$ & Unresolved \\
\midrule
V-A & 115,200 & 0 & 30,669 & 72,366 & 0 \\
V-B / ARD & 57,600 & 0 & 44,778 & 54,172 & 0 \\
V-B / isotropic & 57,600 & 0 & 30,492 & 47,052 & 0 \\
\bottomrule
\end{tabular}
\space
\end{table}
\space
\endgroup 
\space
\Needspace{5\baselineskip}
\subsection[Working-model misspecification: Experiment~V-A]{{Working-model misspecification: Experiment~V-A}}
\label{sec:e8_robustness}
\begingroup 
Experiment~V-A applies the Section~7 procedure under working-model
misspecification. Its 36 configurations cross six families, three geometries
and $N_1=N_q\in\{20,40\}$, with 200 independent evaluation datasets per configuration.
The design coordinates, 201-point prediction grid and positive-LMC,
cancelling-LMC and AR generating covariances and noise variances follow
Appendix~\ref{sec:e8_richer}.
The transferred threshold, target-only folds and full nuisance fitting are
specified in Appendix~\ref{subsec:section7_transfer}. Working means are fixed at
zero and the catalogue uses one-kernel separable covariance models.

The trend family has independent unit latent GPs of scale $0.15$, with means $2x$ and $-1.5x$ and Gaussian noise variance $0.05$ at both outputs. The varying-noise family has correlation $0.7$, scale $0.15$, target noise variance $0.01+0.19x^2$ on the entire domain and auxiliary variance $0.02$. The Student-$t$ family has correlation $0.7$, scale $0.15$ and independent errors $\sqrt{0.05(3/5)}\,t_5$ at both outputs; test errors follow the same law independently.

\begin{table}[htbp]\centering\scriptsize\setlength{\tabcolsep}{3pt}
\caption{{Section~7 catalogue under all six V-A generators:
risk reductions (\%) [95\% paired interval] against SEP. Each family pools
1,200 datasets across its six geometry/count configurations. Gaussian cases
use conditional latent risk; Student-$t$ uses realised squared latent error.}}
\label{tab:e8_families}

\begin{tabular}{@{}lrrr@{}}
\toprule
Family & DIAG & JOINT & Screen \\
\midrule
Positive LMC & -84.12 [-94.74, -73.93] & -81.61 [-93.27, -71.15] & -19.49 [-24.70, -14.86] \\
Cancelling LMC & -43.23 [-51.39, -35.53] & -28.23 [-36.43, -20.53] & -7.12 [-11.70, -2.63] \\
AR fidelity & -37.70 [-44.53, -32.00] & -60.87 [-68.91, -53.59] & -21.23 [-26.02, -17.02] \\
Linear trends & -0.08 [-2.39, 1.96] & -2.46 [-5.15, -0.05] & -0.88 [-2.53, 0.65] \\
Varying noise & 5.07 [3.44, 6.64] & 8.43 [6.58, 10.18] & 5.22 [3.81, 6.57] \\
Student-$t$ & 3.36 [1.28, 5.25] & 6.06 [4.06, 7.96] & 4.10 [2.52, 5.62] \\
\bottomrule
\end{tabular}
\space
\end{table}
\endgroup \begingroup \begin{table}[htbp]\centering\scriptsize\setlength{\tabcolsep}{3pt}
\caption{{V-A direct three-candidate validation and screening
decisions. Direct gives risk reduction against SEP; exclusion and
false-exclusion rates are percentages. Paired intervals accompany direct
risk and false-exclusion rates. The latter use joint fits improving on both
independent candidates by more than 2\%, under each family's assessment loss.}}
\label{tab:e8_screen_decisions}

\begin{tabular}{@{}lrrr@{}}
\toprule
Family & Direct & Exclude (\%) & False exclude (\%) \\
\midrule
Positive LMC & -20.41 [-25.78, -15.57] & 49.25 & 31.91 [23.08, 41.38] \\
Cancelling LMC & -7.04 [-11.62, -2.51] & 46.42 & 25.42 [20.89, 30.08] \\
AR fidelity & -21.61 [-26.46, -17.36] & 47.33 & 19.51 [7.89, 32.43] \\
Linear trends & -1.03 [-2.66, 0.50] & 51.67 & 44.59 [33.33, 55.70] \\
Varying noise & 5.61 [4.20, 6.95] & 38.83 & 6.04 [4.21, 8.17] \\
Student-$t$ & 4.38 [2.79, 5.91] & 48.08 & 20.33 [16.67, 24.22] \\
\bottomrule
\end{tabular}
\space
\end{table}
False-exclusion rates range from 6.04\% in the varying-noise family to 44.59\% in
the trend family. The screen's excess risk over direct validation is small
relative to the losses from covariance misspecification: between $-0.92$ and
$0.39$ percentage points of SEP risk across the six families.
All 115,200 full/fold fits return solver success; Table~\ref{tab:section7_transfer_numerics}
reports parameter-bound and convergence outcomes.
\endgroup 
\space
\subsection[Dimension and scale sensitivity: Experiment~V-B]{{Dimension and scale sensitivity: Experiment~V-B}}
\label{sec:e9_scales}
\begingroup 
The grid is $T\in\{2,5,10\}$, $N_1\in\{20,50\}$, $N_2=2N_1$, three auxiliary geometries and anisotropy ratio $a\in\{1,5\}$. Scrambled Latin hypercube designs lie in $[0,1]^T$. The within-region design uses the full cube; the half-region design multiplies auxiliary coordinate 1 by $0.5$; the remote design adds $1.25$ to it. Target inputs and the underlying auxiliary cube design are identical across geometry/anisotropy variants at fixed $(T,N_1)$. Evaluation uses 1,024 $T$-dimensional scrambled Sobol points, common across variants and fixed independently of responses. Response replications are independent across configurations.

The true unit-variance Mat\'ern 3/2 kernel has scales
$\ell_t=0.20\sqrt{T}$, except $\ell_1=0.20\sqrt{T}/a$, correlation $0.7$ and
noise variance $0.05$ at both outputs. Dimension scaling moderates the growth
of raw Euclidean distances. Each of the 36 configurations has 100 independent
evaluation datasets, separate from the feasibility pilot. Both ARD and isotropic fits apply the complete Section~7
procedure to identical data and target folds, retaining all auxiliaries.
All covariance parameters are estimated; only zero means are fixed.
Each fitted catalogue contains SEP, common-scale DIAG and separable JOINT.
For ARD, DIAG estimates one shared coordinatewise scale vector; for isotropic
fits it estimates one scalar. The primary score multiplies that diagonal
estimate by $0.5,1,2$, holding its signal/noise estimates fixed. Both catalogues
otherwise follow the transfer protocol in Appendix~\ref{subsec:section7_transfer}.

\begin{table}[htbp]\centering\scriptsize\setlength{\tabcolsep}{3pt}
\caption{{Section~7 in Experiment~V-B: mean conditional latent
risk reduction (\%) [95\% paired interval] against the same scale family's SEP
fit. Each dimension pools 1,200 datasets across count, geometry and generating
anisotropy. ARD and isotropic methods use identical observations and folds,
but have separately fitted independent baselines.}}
\label{tab:e9_dimensions}

\begin{tabular}{@{}lrrr@{}}
\toprule
Fit / dimension & DIAG & JOINT & Screen \\
\midrule
ARD / 2 & 8.55 [7.38, 9.68] & 16.41 [15.28, 17.55] & 12.07 [11.07, 13.07] \\
ARD / 5 & 13.20 [12.34, 14.03] & 15.35 [14.43, 16.28] & 11.82 [10.91, 12.72] \\
ARD / 10 & 14.28 [13.66, 14.87] & 13.16 [12.49, 13.83] & 10.88 [10.21, 11.52] \\
Isotropic / 2 & 1.06 [0.42, 1.69] & 8.34 [7.58, 9.15] & 6.62 [5.93, 7.36] \\
Isotropic / 5 & 0.53 [0.30, 0.75] & 4.34 [3.97, 4.69] & 2.70 [2.40, 3.00] \\
Isotropic / 10 & 0.21 [0.05, 0.37] & 2.09 [1.81, 2.36] & 1.13 [0.93, 1.32] \\
\bottomrule
\end{tabular}
\space
\end{table}
ARD screening's excess risk over direct validation is 0.04, 0.07 and 0.14
percentage points of ARD SEP risk in 2, 5 and 10 dimensions. To compare scale
families directly, mean screened risks for ARD and isotropic fits are,
respectively, 0.3155 and 0.3810 at $T=2$, 0.7370 and 0.7071 at $T=5$,
and 0.9280 and 0.8324 at $T=10$.

The screening score was successfully computed for every full training dataset. ARD exclusion frequencies are 33.42\%,
30.17\% and 28.42\%; false-exclusion rates are 0.99\% (0.29--1.71\%),
4.52\% (3.05--6.15\%) and 16.15\% (12.26--20.13\%), using 704, 575 and 322 useful
joint fits as denominators.
\space

\endgroup \begingroup \paragraph*{Coordinatewise covariance scenarios}
For the first ten datasets per configuration, full-data ARD DIAG estimates
also define $\prod_t[0.5\widehat\ell_t,2\widehat\ell_t]$. Evaluate conditional
opportunity at 32 scrambled Sobol vectors in log-scale space plus the three
global vectors. The sampled set therefore includes both opposite corners and
all primary scenarios. These 360 diagnostic comparisons hold fitted marginal
signal/noise estimates fixed and leave the primary gate and threshold unchanged.

For the first two datasets in each two-dimensional configuration, a
$21\times21$ log-scale grid provides an additional finite-resolution reference,
24 datasets in total. This grid and the Sobol set are not nested, so differences
can have either sign.

\begin{table}[htbp]\centering\scriptsize\setlength{\tabcolsep}{2pt}
\caption{{Experiment~V-B covariance-scenario sensitivity. ``Point
decisions vary'' counts datasets whose individual scenario opportunities
straddle $0.20$. ``Global exclude / box retain'' counts changes when enlarging
the three global scenarios to the sampled box. The final column is the largest
increase in the opportunity maximum. These full-data diagnostics do not change
the primary fitted predictor.}}
\label{tab:e9_conditional_scales}

\begin{tabular}{@{}lrrrr@{}}
\toprule
$T$ & Datasets & Point decisions vary & Global exclude / box retain & Max score increase \\
\midrule
2 & 120 & 2 & 0 & 0.092 \\
5 & 120 & 65 & 1 & 0.066 \\
10 & 120 & 83 & 1 & 0.026 \\
\bottomrule
\end{tabular}
\space
\end{table}
Across the 24 dense two-dimensional comparisons, no exclusion decision
changes; the dense maximum exceeds the sampled-box maximum by at most 0.0156
on the $[0,1]$ opportunity scale. These finite checks establish neither
continuous-box extrema nor simultaneous uncertainty coverage for the fitted scale limits.
\space
\endgroup 
\space
\space
\FloatBarrier
\section[Queueing and monitoring applications]{{Queueing and monitoring applications}}
\label{sec:si_applications}
\subsection[Tandem queueing application (TQA)]{{Tandem queueing application (TQA)}}
\label{subsec:q2_calibrated}
\begingroup 
TQA predicts mean total sojourn time in a two-station FCFS queue using independently simulated single-station sojourn time as the auxiliary. Arrivals are Poisson with rate $\lambda$; independent lognormal services have means 1 and $1/1.1$ and common squared coefficient of variation $v$. Departures satisfy
\[
D_{1i}=\max(A_i,D_{1,i-1})+S_{1i},\qquad
D_{2i}=\max(D_{1i},D_{2,i-1})+S_{2i}.
\]
For service mean $m$, $\sigma_{\log}^2=\log(1+v)$ and $\mu_{\log}=\log m-\sigma_{\log}^2/2$. Target observations average eight runs of 20,000 retained customers after 2,000 warm-up customers; auxiliary observations average four runs of 2,000 after 200 warm-up customers. The sample variance of the run means, divided by the number of runs, estimates each observation's simulation variance.

The nine configurations cross 8, 12 or 20 target sites with 60 auxiliaries in within-region, half-region or remote designs. Target and auxiliary Latin hypercubes are regenerated for every macro-dataset. The target rectangle is $[0.35,0.75]\times[0.25,1.5]$ in $(\lambda,v)$; kernel coordinates are $x_1=(\lambda-0.35)/0.4$ and $x_2=(v-0.25)/1.25$. Within-region auxiliaries use the same rectangle; half-region designs multiply cube coordinate 1 by 0.5, and remote designs replace it by $-0.75+0.5x_1$, giving $\lambda\in[0.05,0.25]$. Target data/folds and the underlying auxiliary cube are shared across geometries within each target count. All designs have zero exact overlap.

\endgroup \begingroup \paragraph*{Calibration and independent evaluation}
Independent calibration and evaluation use 100 and 200 macro-datasets per configuration (900 and 1,800 case datasets), with separate pilot and reference samples. Centred and linear means use identical observations and folds. Calibration selects the mean with lower pooled physical MSE under direct three-candidate validation; ties favour centring. Both means are evaluated as prespecified.

For each mean, evaluate thresholds $0,0.0025,0.005,0.01,0.02,0.05,0.1,0.2$.
Eligibility requires one-sided 95\% bootstrap upper limits of at most 1\%
pooled excess risk over direct validation, 1\% pooled positive exclusion
regret, and 2\% excess risk in every target-count/geometry configuration.
Each fraction uses the same-mean SEP risk of its comparison population.
The configuration-specific risk constraint refines Experiment~III's geometry constraint to include target count.
Positive exclusion regret is the loss from excluding a joint fit that improves
on the better independent candidate. Independent reference uncertainty enters
the risk and positive-regret calculations. Among eligible thresholds, choose
the smallest estimated complete component time, breaking ties toward the
smaller threshold. Zero executes direct validation without score computation.
Both thresholds and the primary mean are fixed before evaluation data or its
reference are generated. The full calibration grid appears below.

\begin{table}[htbp]\centering\scriptsize\setlength{\tabcolsep}{3pt}
\caption{{Application calibration on 900 case datasets per mean.
Entries are one-sided 95\% bootstrap upper limits in percent of the matching
SEP risk; eligibility limits are 1, 1 and 2\%, respectively. Every candidate
threshold is reported. Calibration resampling includes reference uncertainty.}}
\label{tab:q2_calibration}

\begin{tabular}{@{}lrlrrr@{}}
\toprule
Mean & $t_U$ & Eligible & Pooled excess & Exclusion regret & Max. cell excess \\
\midrule
Centred & 0 & Yes & 0.00 & 0.00 & 0.00 \\
Centred & 0.0025 & Yes & 0.00 & 0.00 & 0.00 \\
Centred & 0.005 & Yes & 0.00 & 0.00 & 0.00 \\
Centred & 0.01 & Yes & 0.00 & 0.00 & 0.00 \\
Centred & 0.02 & Yes & 0.00 & 0.00 & 0.00 \\
Centred & 0.05 & Yes & 0.00 & 0.00 & 0.00 \\
Centred & 0.1 & Yes & -0.01 & 0.07 & 0.00 \\
Centred & 0.2 & Yes & -0.09 & 0.49 & 0.02 \\
Linear & 0 & Yes & 0.00 & 0.00 & 0.00 \\
Linear & 0.0025 & Yes & 0.06 & 0.09 & 0.71 \\
Linear & 0.005 & Yes & 0.04 & 0.11 & 0.71 \\
Linear & 0.01 & Yes & 0.07 & 0.19 & 0.71 \\
Linear & 0.02 & No & 0.41 & 0.50 & 7.15 \\
Linear & 0.05 & No & 1.15 & 1.52 & 7.21 \\
Linear & 0.1 & No & 1.38 & 1.94 & 8.02 \\
Linear & 0.2 & No & 1.40 & 2.14 & 8.06 \\
\bottomrule
\end{tabular}
\space
\end{table}
Calibration selects the centred mean: its direct-validation physical MSE is
0.01815, compared with 0.03700 for the linear mean. The selected thresholds
are 0.20 for centred and 0.01 for linear means. In the linear mean specification, 0.02 and
larger thresholds fail the calibration limits. The calibrated centred rule
therefore coincides with fixed-0.20 screening, while the linear sensitivity
assesses the effect of a lower threshold on the same evaluation observations.
\space

\paragraph*{Matched fitting and assessment}
SEP, shared-scale DIAG and separable JOINT use stage-specific output centring/scaling and training-estimated simulation variances plus fitted nuggets. Linear means profile an intercept and both input coordinates. All covariance parameters and mean coefficients are re-estimated within four target-response folds retaining auxiliaries. The common exact fitting settings apply, with nugget bounds $[10^{-8},2]$. Scale multipliers $0.5,1,2$, correlation envelope 0.995 and validation/assessment fractions 0.02 match main Experiment~III (Section~\ref{subsec:revision_e6_results}). The score uses the fitted diagonal noise matrices and conditions on fitted means; linear-mean predictive variance includes the profiled-mean variance term.

Scores use the 49-point prediction grid. Main Table~\ref{tab:q2_opportunity} reports median full-data scores and exclusion/selection frequencies over all 200 datasets per configuration. Direct, fixed-0.20 and calibrated procedures share fitted candidates and fold losses; execution skips only operations permitted by each rule.

\begin{table}[htbp]\centering\scriptsize\setlength{\tabcolsep}{2pt}
\caption{{Independent TQA evaluation with centred means.
Panel (a): SEP RMSE in physical time units and procedure MSE reductions (\%)
against SEP. Panel (b): DIAG and JOINT MSE reductions against SEP and final
selection frequencies (\%) for the calibrated screen. Brackets give 95\%
paired intervals, propagating macro-dataset and reference uncertainty.
Each row has 200 macro-datasets and 60 auxiliaries; selection frequencies
use all 200 datasets. Centred calibration selects $t_U=0.20$, so fixed-threshold
and calibrated screening coincide and are reported once.
Calibration and evaluation data and reference runs are disjoint.}}
\label{tab:q2_calibrated_centred}

\begin{tabular}{@{}rlrrr@{}}
\toprule
\multicolumn{5}{l}{(a) Prediction procedures} \\

$N_1$ & Auxiliary region & SEP RMSE & Direct & Calibrated \\
\midrule
8 & Within-region & 0.2012 & 49.61 [43.70, 55.04] & 49.61 [43.70, 55.04] \\
8 & Half-region & 0.2012 & 1.17 [-3.06, 4.93] & 1.17 [-3.06, 4.93] \\
8 & Remote & 0.2012 & 4.45 [2.24, 6.62] & 4.98 [3.16, 6.92] \\
12 & Within-region & 0.1265 & 12.12 [3.49, 20.13] & 12.12 [3.49, 20.13] \\
12 & Half-region & 0.1265 & -9.14 [-13.26, -5.64] & -9.14 [-13.26, -5.64] \\
12 & Remote & 0.1265 & -0.66 [-1.20, -0.22] & -0.15 [-0.45, 0.18] \\
20 & Within-region & 0.0832 & 4.45 [-0.60, 9.30] & 4.45 [-0.60, 9.30] \\
20 & Half-region & 0.0832 & -2.93 [-4.20, -1.74] & -2.54 [-3.72, -1.46] \\
20 & Remote & 0.0832 & -0.27 [-0.46, -0.11] & 0.00 [0.00, 0.00] \\
\bottomrule
\end{tabular}
\par\smallskip
\begin{tabular}{@{}rlrrrrr@{}}
\toprule
\multicolumn{2}{l}{(b) Candidates and final selection} & \multicolumn{2}{c}{MSE reduction vs SEP (\%)} & \multicolumn{3}{c}{Calibrated selection (\%)} \\
\cmidrule(lr){3-4}\cmidrule(l){5-7}
$N_1$ & Auxiliary region & DIAG & JOINT & SEP & DIAG & JOINT \\
\midrule
8 & Within-region & 10.15 [8.15, 12.18] & 52.84 [46.81, 58.17] & 0.5 & 8.5 & 91.0 \\
8 & Half-region & 6.28 [3.61, 8.87] & -6.14 [-13.23, 0.27] & 7.0 & 30.5 & 62.5 \\
8 & Remote & 8.99 [7.08, 10.95] & 1.41 [-1.92, 4.73] & 17.5 & 77.5 & 5.0 \\
12 & Within-region & 0.08 [-0.18, 0.34] & 21.68 [12.39, 29.59] & 8.5 & 9.5 & 82.0 \\
12 & Half-region & -1.39 [-2.41, -0.54] & -16.26 [-22.65, -10.61] & 35.5 & 21.0 & 43.5 \\
12 & Remote & -1.30 [-2.24, -0.54] & -5.19 [-7.19, -3.56] & 57.5 & 42.0 & 0.5 \\
20 & Within-region & 0.00 [0.00, 0.00] & 2.29 [-6.86, 10.18] & 30.0 & 0.5 & 69.5 \\
20 & Half-region & -0.41 [-0.81, -0.10] & -14.24 [-16.68, -11.80] & 84.0 & 0.5 & 15.5 \\
20 & Remote & -0.37 [-0.62, -0.16] & -2.50 [-3.09, -1.95] & 99.5 & 0.5 & 0.0 \\
\bottomrule
\end{tabular}
\space
\end{table}
\begin{table}[htbp]\centering\scriptsize\setlength{\tabcolsep}{2pt}
\caption{{Independent TQA linear-mean sensitivity on the same evaluation
datasets and folds, with panels as in Table~\ref{tab:q2_calibrated_centred}.
Reductions use the linear-mean SEP baseline; 95\% paired intervals propagate
both macro-dataset and reference uncertainty. Selection frequencies use all
200 datasets per configuration under the separately calibrated linear-mean screen.}}
\label{tab:q2_calibrated_linear}

\begin{tabular}{@{}rlrrrr@{}}
\toprule
\multicolumn{6}{l}{(a) Prediction procedures} \\

$N_1$ & Auxiliary region & SEP RMSE & Direct & Fixed $t_U=0.20$ & Calibrated \\
\midrule
8 & Within-region & 0.3091 & 66.68 [62.08, 71.04] & 66.68 [62.08, 71.04] & 66.68 [62.08, 71.04] \\
8 & Half-region & 0.3091 & 2.95 [-0.58, 6.44] & 2.31 [-1.17, 5.70] & 2.67 [-0.77, 6.10] \\
8 & Remote & 0.3091 & 2.65 [1.42, 3.97] & 0.82 [0.14, 1.63] & 2.87 [1.75, 4.17] \\
12 & Within-region & 0.1891 & 56.42 [50.84, 61.09] & 56.42 [50.84, 61.09] & 56.42 [50.84, 61.09] \\
12 & Half-region & 0.1891 & -4.76 [-11.52, 1.66] & -2.07 [-8.23, 4.07] & -4.76 [-11.52, 1.66] \\
12 & Remote & 0.1891 & 0.12 [-4.16, 4.08] & -5.83 [-9.53, -3.01] & -0.01 [-4.06, 3.87] \\
20 & Within-region & 0.0887 & 5.20 [-0.93, 10.60] & 5.20 [-0.93, 10.60] & 5.20 [-0.93, 10.60] \\
20 & Half-region & 0.0887 & -6.40 [-8.91, -4.22] & -4.74 [-7.21, -2.58] & -5.11 [-7.60, -2.98] \\
20 & Remote & 0.0887 & -0.12 [-3.03, 2.10] & -0.95 [-2.89, 0.00] & 0.08 [-2.84, 2.30] \\
\bottomrule
\end{tabular}
\par\smallskip
\begin{tabular}{@{}rlrrrrr@{}}
\toprule
\multicolumn{2}{l}{(b) Candidates and final selection} & \multicolumn{2}{c}{MSE reduction vs SEP (\%)} & \multicolumn{3}{c}{Calibrated selection (\%)} \\
\cmidrule(lr){3-4}\cmidrule(l){5-7}
$N_1$ & Auxiliary region & DIAG & JOINT & SEP & DIAG & JOINT \\
\midrule
8 & Within-region & -12.35 [-16.83, -8.24] & 73.14 [69.95, 76.03] & 6.0 & 2.5 & 91.5 \\
8 & Half-region & -13.20 [-17.70, -9.13] & 1.31 [-4.75, 6.73] & 55.0 & 6.5 & 38.5 \\
8 & Remote & -5.33 [-8.88, -2.23] & 0.12 [-3.50, 3.40] & 66.0 & 12.5 & 21.5 \\
12 & Within-region & 6.14 [1.05, 11.44] & 57.17 [52.06, 61.75] & 0.5 & 0.5 & 99.0 \\
12 & Half-region & -58.67 [-75.67, -44.43] & -22.36 [-38.24, -8.33] & 52.5 & 2.5 & 45.0 \\
12 & Remote & -31.93 [-41.82, -23.51] & -11.27 [-20.38, -2.84] & 49.5 & 12.0 & 38.5 \\
20 & Within-region & -14.07 [-16.75, -11.73] & 3.08 [-4.62, 9.69] & 13.5 & 3.5 & 83.0 \\
20 & Half-region & -54.52 [-87.65, -28.32] & -35.71 [-64.59, -13.81] & 73.0 & 0.5 & 26.5 \\
20 & Remote & -28.88 [-33.04, -25.18] & -0.95 [-4.32, 1.68] & 66.5 & 0.5 & 33.0 \\
\bottomrule
\end{tabular}
\space
\end{table}
Centred screening's excess risk over direct validation is $-0.18$ ($-0.43$ to 0.10) percentage points of SEP risk. Nominal 95\% mean-response intervals attain 94.48\% coverage (94.12--94.82\%).
In the linear sensitivity, calibration lowers the threshold from 0.20 to 0.01. On the same evaluation datasets, false exclusion falls from 31.39\% to 4.48\%, and excess risk over direct validation falls from 0.83 (0.45--1.24) to $-0.002$ ($-0.16$ to 0.17) percentage points. The centred procedure has lower absolute MSE in every configuration.
\space

\begin{table}[htbp]\centering\scriptsize\setlength{\tabcolsep}{2pt}
\caption{{Independent TQA aggregate prediction and exclusions over all
1,800 case datasets per mean. MSE reductions use same-mean SEP as the baseline;
excess risk is in percentage points of that SEP risk. Both use ratios of
pooled mean losses. False exclusion is conditional on joint fits improving physical MSE
by more than 2\% on both independent candidates. Brackets give paired 95\%
intervals; no cases are removed after evaluation. Centred fixed-0.20 and
calibrated screening coincide and are reported once.}}
\label{tab:q2_calibrated_decisions}

\begin{tabular}{@{}lrrrrr@{}}
\toprule
Mean / procedure & MSE & Reduction vs SEP (\%) & Excess vs direct (pp) & Exclude (\%) & False exclusion (\%) \\
\midrule
Centred / Direct & 0.018598 & 12.00 [9.68, 14.10] & 0.00 [0.00, 0.00] & 0.00 & 0.00 [0.00, 0.00] \\
Centred / Calibrated & 0.018560 & 12.18 [9.87, 14.31] & -0.18 [-0.43, 0.10] & 28.06 & 14.92 [12.48, 17.44] \\
Linear / Direct & 0.036672 & 20.95 [19.31, 22.54] & 0.00 [0.00, 0.00] & 0.00 & 0.00 [0.00, 0.00] \\
Linear / Fixed 0.20 & 0.037056 & 20.12 [18.56, 21.69] & 0.83 [0.45, 1.24] & 46.06 & 31.39 [29.01, 33.90] \\
Linear / Calibrated & 0.036671 & 20.95 [19.36, 22.51] & 0.00 [-0.16, 0.17] & 12.39 & 4.48 [3.24, 5.77] \\
\bottomrule
\end{tabular}
\space
\end{table}

\endgroup \begingroup \paragraph*{Reference uncertainty, convergence and cost}
Calibration and evaluation each have an independent 49-point interior-grid
reference, with 128 runs per point, 100,000 retained customers and 10,000 warm-up.
Each of 2,000 bootstrap draws resamples complete macro-datasets within target
count, sharing weights across geometry and methods, and independently resamples
reference runs within each point. The reference draw is common to all methods.
MSE, relative gains and paired differences are recalculated. Predictive-interval
coverage refers to the simulator mean response.
All 57,600 evaluation full/fold fits succeed; 50,895 reach bounds and 27,793 stop above the gradient tolerance.

Calibration uses 28,800 full/fold fits across both means. Parallel calibration data generation and fitting take 90.62 seconds on 16 workers; evaluation takes 178.06 seconds, excluding references and bootstrap analysis. The calibration reference consumes 49.35 CPU seconds; its mean and maximum Monte Carlo standard errors are 0.00390 and 0.01296. The evaluation reference consumes 47.46 CPU seconds, with corresponding errors of 0.00394 and 0.01417. Table~\ref{tab:q2_calibrated_timing} gives complete application timings after calibration.
\space

\endgroup \begingroup \begin{table}[htbp]\centering\scriptsize\setlength{\tabcolsep}{3pt}
\caption{{Paired complete executions for the calibration-selected
mean on the first five independent macro-datasets per configuration (45 case datasets).
Execution order rotates. Times include all fitting, scoring, validation and prediction operations;
loading, simulation, reference generation and offline calibration are separate.
Brackets give 95\% paired bootstrap intervals. Fixed-0.20 and calibrated
screening use the same centred rule; their timings come from separate executions.}}
\label{tab:q2_calibrated_timing}

\begin{tabular}{@{}lrr@{}}
\toprule
Procedure & Time (s) & Saving vs direct (\%) \\
\midrule
Direct validation & 0.4930 & 0.00 [0.00, 0.00] \\
Fixed 0.20 & 0.4495 & 8.83 [6.66, 10.82] \\
Calibrated screen & 0.4506 & 8.61 [6.51, 10.52] \\
\bottomrule
\end{tabular}
\space
\end{table}
\endgroup 
\space

\FloatBarrier
\Needspace{10\baselineskip}
\subsection[Task-specific AQS analysis]{{Task-specific AQS analysis}}
\label{sec:aqs_revised}

\begingroup 
\paragraph*{Annual response and preprocessing}
The 2024 national annual-summary file contains 67,923 rows; selecting states 06, 08 and 48 and pollutants 88101, 44201 and 42602 leaves 4,164 before defining annual measures. PM$_{2.5}$ uses ``PM25 Annual 2024'', quarterly means of daily means from 24-hour observations, in micrograms/m$^3$ at local conditions. Ozone uses ``Ozone 8-hour 2015'', annual arithmetic summaries of daily maximum eight-hour running averages with beginning-hour duration, in ppm. NO$_2$ uses ``NO2 Annual 1971'', annual summaries of one-hour observations, in ppb. Other regulatory summaries of the same pollutant are excluded.

No-events/events-included records with finite values and coordinates, positive observation count, completeness Y and no QA-issue certification flag are eligible; pending certification is retained. Instruments are ranked without concentrations, using certification, valid-day count, observation percentage/count, duration order, POC and original row order. One instrument is primary; averaging all eligible collocated instruments is a separate response-definition sensitivity (39 multiple-instrument site/pollutant cases). The selected site counts and available auxiliaries are in Table~\ref{tab:aqs_revised_counts}.

\begin{table}[htbp]\centering\scriptsize\setlength{\tabcolsep}{3pt}
\caption{{AQS site counts under the defined annual measures. Task A auxiliary availability is the percentage of target sites at which at least one auxiliary pollutant remains observed. Each state is one observed network; these counts have no sampling interval.}}
\label{tab:aqs_revised_counts}%

\begin{tabular}{llrr}
\toprule
State & Pollutant & Sites & Task A auxiliary availability (\%) \\
\midrule
CA & PM$_{2.5}$ & 111 & 63.1 \\
CA & O$_3$ & 124 & 71.0 \\
CA & NO$_2$ & 73 & 95.9 \\
CO & PM$_{2.5}$ & 17 & 41.2 \\
CO & O$_3$ & 38 & 26.3 \\
CO & NO$_2$ & 11 & 100.0 \\
TX & PM$_{2.5}$ & 43 & 51.2 \\
TX & O$_3$ & 71 & 62.0 \\
TX & NO$_2$ & 48 & 91.7 \\
\bottomrule
\end{tabular}
\space
\end{table}

\paragraph*{Spatial metric and information sets}
Site IDs are state/county/site number. NAD83/WGS84 coordinates are transformed to WGS84 without ballpark transformations, then projected to a state-centred azimuthal equidistant plane in kilometres. Both projected coordinates are divided by the larger coordinate range of the state inventory. Maximum pairwise relative errors against ellipsoidal geodesics are 0.165\%, 0.039\% and 0.316\% for California, Colorado and Texas. Projection and partition rules use coordinates and pollutant availability, without concentrations.

Task A withholds only the target pollutant, permitting observed auxiliaries at its site; B withholds all outputs at a site; C withholds all outputs in a contiguous nearest-centroid region. Ten seeded partitions use five folds, or three in Colorado because NO$_2$ has fewer than 15 sites. Task C uses $k$-means++ initialisation and 100 $k$-means iterations on site coordinates. Its exclusion buffers are zero, one or two times the training median nearest-neighbour distance; A/B use zero buffer. B/C have zero shared training/test sites; every pollutant retains at least three training sites. Each site receives ten out-of-fold predictions per task/buffer/response definition.

For each held-out point, $\delta_j$ is the distance to the nearest retained
observation of its pollutant, and $\delta_{-j}$ is the distance to the nearest
retained observation across the other two pollutants. We average distances
over ten partitions within each site, sites within each pollutant, and the
three pollutants equally. For primary instruments and zero buffer, the mean
nearest-target/nearest-auxiliary distance pairs
$(\overline\delta_j,\overline\delta_{-j})$, in kilometres, are:
\begin{center}\small\setlength{\tabcolsep}{7pt}

\begin{tabular}{@{}lrrr@{}}
\toprule
State & Task A & Task B & Task C \\
\midrule
CA & (26.9, 8.0) & (27.1, 25.0) & (110.6, 107.0) \\
CO & (42.4, 23.4) & (45.5, 43.6) & (210.2, 180.1) \\
TX & (43.8, 14.5) & (43.8, 35.7) & (244.4, 225.5) \\
\bottomrule
\end{tabular}
\space
\end{center}

\paragraph*{Models, fitting and uncertainty}
Each fold estimates output means and sample standard deviations using training responses only; predictions return to physical units. Exact float64 Mat\'ern-$3/2$ fits use the common variance/scale bounds, four-start schedule, tolerances and jitter ladder in Appendix~\ref{subsec:si_common_fitting}. \texttt{SEP} fits each marginal independently; \texttt{DIAG} fits a common scale with diagonal output covariance; \texttt{JOINT} fits a full positive-definite output covariance. Training mean, nearest neighbour, five-neighbour inverse-distance weighting (power 2) and an OLS spatial plane provide point baselines. All comparisons are prespecified; outer test outcomes do not tune or select models. Failures of \texttt{JOINT} or \texttt{DIAG} use the corresponding \texttt{SEP} fit; marginal failures use its training mean.

All 9,100 primary/instrument-sensitivity fits over 1,820 splits return solver success; 1,993 reach bounds and 5,561 stop above the gradient tolerance. Each site's squared errors and predictive scores are averaged over its ten out-of-fold predictions before 2,000 paired resamples of ten coordinate-only spatial blocks; all methods and tasks share block weights. Empty-denominator draws are counted and omitted (at least 1,985 valid draws per summary). These intervals condition on fitted predictions and omit refitting uncertainty. Colorado NO$_2$ occupies only four blocks. Table~\ref{tab:aqs_baselines_all} shows all primary baselines; Table~\ref{tab:aqs_borrowing_all} gives paired borrowing intervals.

\begin{table}[htbp]\centering\scriptsize\setlength{\tabcolsep}{2pt}
\caption{{AQS primary-instrument, zero-buffer RMSE in pollutant-specific physical units: micrograms/m$^3$, ppm and ppb. Lower is better. Means average site losses before taking square roots; repeated folds are dependent. GP covariance and training transformations, or spatial-plane coefficients, are fitted within folds. \texttt{MIX} is the matching independent two-kernel mixture for the two-component LMC. All methods use the same held-out sites and partitions. CA: California; CO: Colorado; TX: Texas.}}
\label{tab:aqs_baselines_all}%

\begin{tabular}{@{}lllrrrrrrrrr@{}}
\toprule
State & Task & Pollutant & Mean & Nearest & IDW & Trend & \texttt{SEP} & \texttt{DIAG} & \texttt{JOINT} & \texttt{MIX} & LMC \\
\midrule
CA & A & PM$_{2.5}$ & 3.040 & 2.415 & 2.234 & 2.898 & 2.107 & 2.132 & 2.086 & 2.075 & 1.986 \\
CA & A & O$_3$ & 0.00733 & 0.00614 & 0.00503 & 0.00587 & 0.00474 & 0.00475 & 0.00475 & 0.00472 & 0.00464 \\
CA & A & NO$_2$ & 5.692 & 4.861 & 4.355 & 5.126 & 4.174 & 4.125 & 3.851 & 4.112 & 3.466 \\
\addlinespace[2pt]
CA & B & PM$_{2.5}$ & 3.032 & 2.383 & 2.160 & 2.868 & 2.024 & 2.031 & 2.040 & 1.985 & 2.005 \\
CA & B & O$_3$ & 0.00734 & 0.00623 & 0.00514 & 0.00586 & 0.00483 & 0.00485 & 0.00490 & 0.00479 & 0.00475 \\
CA & B & NO$_2$ & 5.671 & 4.959 & 4.457 & 5.116 & 4.208 & 4.160 & 4.053 & 4.140 & 3.721 \\
\addlinespace[2pt]
CA & C & PM$_{2.5}$ & 3.193 & 4.012 & 3.693 & 3.042 & 3.142 & 3.173 & 3.118 & 3.097 & 2.986 \\
CA & C & O$_3$ & 0.00787 & 0.00718 & 0.00669 & 0.00644 & 0.00647 & 0.00716 & 0.00724 & 0.00640 & 0.00643 \\
CA & C & NO$_2$ & 6.725 & 6.341 & 6.438 & 6.256 & 6.788 & 6.424 & 6.406 & 6.619 & 6.069 \\
\midrule
CO & A & PM$_{2.5}$ & 1.166 & 1.334 & 1.128 & 1.073 & 1.132 & 1.143 & 1.130 & 1.121 & 1.057 \\
CO & A & O$_3$ & 0.00351 & 0.00423 & 0.00365 & 0.00342 & 0.00352 & 0.00367 & 0.00374 & 0.00353 & 0.00365 \\
CO & A & NO$_2$ & 8.371 & 7.116 & 7.252 & 7.623 & 7.107 & 6.893 & 7.082 & 7.076 & 7.235 \\
\addlinespace[2pt]
CO & B & PM$_{2.5}$ & 1.193 & 1.464 & 1.223 & 1.071 & 1.224 & 1.235 & 1.208 & 1.223 & 1.175 \\
CO & B & O$_3$ & 0.00354 & 0.00410 & 0.00358 & 0.00351 & 0.00347 & 0.00354 & 0.00357 & 0.00343 & 0.00345 \\
CO & B & NO$_2$ & 8.867 & 7.494 & 7.816 & 23.117 & 7.462 & 7.532 & 7.552 & 7.398 & 7.309 \\
\addlinespace[2pt]
CO & C & PM$_{2.5}$ & 1.498 & 1.634 & 1.472 & 1.331 & 1.498 & 1.497 & 1.500 & 1.498 & 1.490 \\
CO & C & O$_3$ & 0.00440 & 0.00428 & 0.00422 & 0.00593 & 0.00440 & 0.00440 & 0.00439 & 0.00440 & 0.00438 \\
CO & C & NO$_2$ & 12.325 & 10.238 & 11.758 & 22.861 & 12.325 & 12.325 & 12.325 & 12.320 & 12.325 \\
\midrule
TX & A & PM$_{2.5}$ & 2.442 & 2.639 & 2.304 & 2.178 & 2.535 & 2.438 & 2.402 & 2.400 & 2.207 \\
TX & A & O$_3$ & 0.00441 & 0.00351 & 0.00272 & 0.00289 & 0.00251 & 0.00274 & 0.00280 & 0.00252 & 0.00248 \\
TX & A & NO$_2$ & 3.589 & 4.352 & 3.773 & 3.764 & 3.392 & 3.653 & 3.596 & 3.254 & 3.197 \\
\addlinespace[2pt]
TX & B & PM$_{2.5}$ & 2.419 & 2.643 & 2.311 & 2.109 & 2.512 & 2.456 & 2.439 & 2.354 & 2.189 \\
TX & B & O$_3$ & 0.00439 & 0.00340 & 0.00272 & 0.00284 & 0.00247 & 0.00265 & 0.00279 & 0.00250 & 0.00251 \\
TX & B & NO$_2$ & 3.598 & 4.272 & 3.701 & 3.676 & 3.358 & 3.591 & 3.429 & 3.196 & 3.209 \\
\addlinespace[2pt]
TX & C & PM$_{2.5}$ & 2.507 & 2.843 & 2.505 & 2.464 & 2.505 & 2.497 & 2.706 & 2.526 & 2.539 \\
TX & C & O$_3$ & 0.00555 & 0.00420 & 0.00375 & 0.00477 & 0.00278 & 0.00396 & 0.00393 & 0.00278 & 0.00275 \\
TX & C & NO$_2$ & 3.643 & 4.341 & 3.952 & 5.706 & 3.697 & 4.510 & 5.902 & 4.101 & 5.363 \\
\bottomrule
\end{tabular}
\space
\end{table}
\FloatBarrier

\paragraph*{Richer covariance and sensitivities}
A two-component LMC and its matching independent two-kernel mixture use all 650 unbuffered primary A/B/C splits. Each component has half the separable signal-variance bounds; the same noise/scale bounds, four-start budgets and tolerances apply. The optimiser reported convergence for all 2,600 fits; 1,089 reached parameter bounds and 2,297 stopped above the gradient tolerance. Table~\ref{tab:aqs_borrowing_all} keeps this own-baseline comparison distinct from improvement due to marginal flexibility.
{\scriptsize\setlength{\tabcolsep}{3pt}%

\begin{longtable}{lllrrr}
\caption[Separable and richer borrowing by task]{{AQS pollutant-specific MSE reduction (\%) [95\% paired spatial-block interval], primary instruments and zero buffer. \texttt{JOINT} and LMC use their matching single- and two-kernel independent baselines. Ten out-of-fold losses per site are averaged before 2,000 paired block resamples; intervals condition on fitted predictions and omit refitting uncertainty. Blocks are occupied spatial blocks. Positive values favour joint prediction.}}\label{tab:aqs_borrowing_all}\\
\toprule
State & Task & Pollutant & \texttt{JOINT} vs \texttt{SEP} & LMC vs mixture independent & Blocks \\
\midrule
\endfirsthead
\toprule
State & Task & Pollutant & \texttt{JOINT} vs \texttt{SEP} & LMC vs mixture independent & Blocks \\
\midrule
\endhead
\bottomrule\endfoot
CA & A & PM$_{2.5}$ & 1.93 [-7.14, 10.66] & 8.42 [-0.36, 16.97] & 10 \\
CA & A & O$_3$ & -0.40 [-6.24, 4.45] & 3.36 [-0.98, 11.10] & 10 \\
CA & A & NO$_2$ & 14.91 [5.00, 30.25] & 28.95 [21.38, 52.45] & 8 \\
CA & B & PM$_{2.5}$ & -1.59 [-8.45, 4.99] & -2.07 [-13.47, 7.12] & 10 \\
CA & B & O$_3$ & -2.71 [-9.27, 1.91] & 1.72 [-1.78, 7.04] & 10 \\
CA & B & NO$_2$ & 7.22 [-1.78, 14.97] & 19.24 [9.90, 36.45] & 8 \\
CA & C & PM$_{2.5}$ & 1.50 [-4.83, 5.92] & 7.01 [-7.53, 14.26] & 10 \\
CA & C & O$_3$ & -25.30 [-57.89, -5.28] & -0.90 [-4.88, 4.50] & 10 \\
CA & C & NO$_2$ & 10.94 [7.17, 23.74] & 15.92 [3.51, 47.22] & 8 \\
\midrule
CO & A & PM$_{2.5}$ & 0.42 [-10.01, 26.13] & 11.07 [-18.69, 36.54] & 6 \\
CO & A & O$_3$ & -12.93 [-66.84, 5.46] & -6.92 [-38.19, 4.17] & 9 \\
CO & A & NO$_2$ & 0.71 [-23.58, 29.87] & -4.53 [-35.07, 9.19] & 4 \\
CO & B & PM$_{2.5}$ & 2.61 [-5.12, 4.97] & 7.69 [4.80, 24.81] & 6 \\
CO & B & O$_3$ & -6.30 [-21.38, 0.46] & -0.78 [-19.38, 11.33] & 9 \\
CO & B & NO$_2$ & -2.41 [-8.97, 3.09] & 2.39 [-14.79, 30.78] & 4 \\
CO & C & PM$_{2.5}$ & -0.26 [-2.40, $4\times10^{-5}$] & 0.99 [-0.01, 9.29] & 6 \\
CO & C & O$_3$ & 0.53 [-0.70, 2.13] & 1.19 [-0.79, 2.55] & 9 \\
CO & C & NO$_2$ & $4.5\times10^{-7}$ [$-1.5\times10^{-9}$, $2.2\times10^{-6}$] & -0.08 [-0.40, $-8\times10^{-4}$] & 4 \\
\midrule
TX & A & PM$_{2.5}$ & 10.16 [-34.45, 36.86] & 15.43 [-5.67, 46.04] & 9 \\
TX & A & O$_3$ & -24.52 [-90.69, -7.01] & 3.37 [-5.40, 15.94] & 10 \\
TX & A & NO$_2$ & -12.43 [-32.97, 6.17] & 3.49 [-7.72, 18.03] & 6 \\
TX & B & PM$_{2.5}$ & 5.71 [-43.18, 26.43] & 13.51 [-0.95, 41.88] & 9 \\
TX & B & O$_3$ & -27.77 [-106.66, -5.95] & -0.86 [-5.83, 8.66] & 10 \\
TX & B & NO$_2$ & -4.31 [-28.81, 20.09] & -0.86 [-9.03, 10.19] & 6 \\
TX & C & PM$_{2.5}$ & -16.71 [-97.00, -3.33] & -1.01 [-86.69, 30.07] & 9 \\
TX & C & O$_3$ & -100.86 [-193.48, -16.21] & 2.00 [-13.11, 22.72] & 10 \\
TX & C & NO$_2$ & -154.79 [-264.40, -3.85] & -71.04 [-108.34, -1.81] & 6 \\
\end{longtable}
}\par\smallskip
\paragraph*{Equal-output MSE reduction}
For a fixed state, task and covariance family, let \(L_{mji}\) be method \(m\)'s squared prediction error for pollutant \(j\) at site \(i\), averaged over its ten out-of-fold predictions. With \(n_j\) sites for that pollutant, define \(M_{mj}=n_j^{-1}\sum_{i=1}^{n_j}L_{mji}\). Main Table~\ref{tab:aqs_summary} reports
\[
R_m=100\biggl(1-\frac{1}{3}\sum_{j=1}^{3}\frac{M_{mj}}{M_{bj}}\biggr),
\]
where \(b\) is \texttt{SEP} for the separable comparison and the independent two-kernel mixture for the LMC comparison. Thus each pollutant has equal weight after normalisation by its matching independent MSE. In each spatial-block bootstrap draw, the same block multiplicities weight sites for every pollutant and method; we recompute each pollutant's two MSEs, their ratio and then \(R_m\). The interval is the 2.5th--97.5th percentile range of finite aggregate draws, omitting a draw if any pollutant's ratio is undefined. Table~\ref{tab:aqs_borrowing_all} gives pollutant-specific comparisons. Intervals are pointwise and unadjusted for multiple comparisons.

For the separable family, instrument-definition and Task C buffer sensitivities give inconclusive equal-output JOINT-versus-SEP comparisons except in Texas Task C, where reductions range from $-91.34\%$ to $-90.64\%$, with all 95\% intervals below zero. Noisy-response MLPD and 95\% intervals use latent variance plus fitted observation noise. Table~\ref{tab:aqs_calibration_all} reports these summaries for \texttt{SEP}, \texttt{DIAG}, \texttt{JOINT}, the independent two-kernel mixture and LMC.
\begin{table}[htbp]\centering\scriptsize\setlength{\tabcolsep}{1pt}
\caption{{AQS response calibration for separable and two-component families: primary instruments, zero buffer. Cells give 95\% response-interval coverage (\%) / mean width / MLPD in each pollutant\textquotesingle s physical units. Scores average ten out-of-fold predictions per site, then sites. All transforms and covariance parameters are fitted within training folds; these are point summaries. \texttt{MIX} is the matching independent two-kernel mixture for LMC. CA: California; CO: Colorado; TX: Texas.}}
\label{tab:aqs_calibration_all}

\begin{tabular}{@{}lllrrrrr@{}}
\toprule
State & Task & Pollutant & \texttt{SEP} & \texttt{DIAG} & \texttt{JOINT} & \texttt{MIX} & LMC \\
\midrule
CA & A & PM$_{2.5}$ & 93.5/7.88/-2.22 & 93.2/7.79/-2.21 & 92.0/7.52/-2.20 & 93.4/7.85/-2.22 & 92.3/7.38/-2.19 \\
CA & A & Ozone & 92.6/0.0188/3.89 & 92.2/0.019/3.89 & 92.5/0.0189/3.88 & 91.7/0.0184/3.85 & 91.5/0.0182/3.87 \\
CA & A & NO$_2$ & 93.7/16.8/-2.88 & 93.6/16.9/-2.87 & 89.7/15.2/-2.80 & 92.7/16.5/-2.88 & 90.3/13.7/-2.70 \\
\addlinespace[2pt]
CA & B & PM$_{2.5}$ & 94.5/7.97/-2.16 & 94.2/7.87/-2.15 & 93.7/7.75/-2.17 & 94.8/7.93/-2.15 & 94.7/7.87/-2.18 \\
CA & B & Ozone & 91.7/0.0186/3.81 & 91.3/0.0188/3.80 & 91.1/0.0189/3.81 & 91.8/0.0183/3.80 & 91.7/0.0182/3.81 \\
CA & B & NO$_2$ & 93.4/16.7/-2.90 & 94.1/16.9/-2.89 & 91.1/15.5/-2.86 & 92.2/16.5/-2.90 & 91.0/14.5/-2.81 \\
\addlinespace[2pt]
CA & C & PM$_{2.5}$ & 90.5/10.8/-2.62 & 90.3/11.4/-2.63 & 90.0/11.3/-2.61 & 88.8/10.6/-2.61 & 89.4/10.1/-2.58 \\
CA & C & Ozone & 93.7/0.0257/3.60 & 90.8/0.0251/3.49 & 90.3/0.0249/3.47 & 94.1/0.0251/3.61 & 93.4/0.0247/3.60 \\
CA & C & NO$_2$ & 83.0/17.8/-3.88 & 85.2/17.7/-3.74 & 83.8/16.9/-3.95 & 80.5/17.4/-3.84 & 83.4/16.4/-3.96 \\
\midrule
CO & A & PM$_{2.5}$ & 87.6/3.77/-2.91 & 89.4/3.88/-3.11 & 82.9/3.62/-3.14 & 88.2/3.75/-2.94 & 72.4/3.3/-397.47 \\
CO & A & Ozone & 91.8/0.0117/4.04 & 90.3/0.0117/4.00 & 88.7/0.0115/-63.99 & 91.3/0.0116/4.05 & 86.6/0.0105/-70.91 \\
CO & A & NO$_2$ & 81.8/23.7/-5.89 & 81.8/23.7/-5.06 & 79.1/20.7/-16.83 & 81.8/23.6/-6.07 & 66.4/15.5/-7.36 \\
\addlinespace[2pt]
CO & B & PM$_{2.5}$ & 84.7/3.66/-5.08 & 87.1/3.78/-6.29 & 82.9/3.7/-5.16 & 84.1/3.66/-4.98 & 82.9/4.02/-4.03 \\
CO & B & Ozone & 91.8/0.0121/4.01 & 92.4/0.0123/4.02 & 91.8/0.0121/4.00 & 92.1/0.012/4.03 & 91.8/0.0114/4.06 \\
CO & B & NO$_2$ & 77.3/22.8/-5.54 & 82.7/23.7/-6.13 & 78.2/23.7/-6.45 & 76.4/22.8/-5.61 & 78.2/26/-4.95 \\
\addlinespace[2pt]
CO & C & PM$_{2.5}$ & 61.2/2.72/-4.04 & 61.2/2.73/-4.04 & 61.2/2.83/-3.92 & 61.2/2.72/-4.04 & 75.3/3.96/-2.85 \\
CO & C & Ozone & 94.7/0.0133/3.68 & 94.7/0.0134/3.66 & 93.7/0.0131/3.68 & 94.7/0.0134/3.69 & 93.2/0.013/3.75 \\
CO & C & NO$_2$ & 45.5/10.6/-82.49 & 45.5/10.7/-82.49 & 45.5/15.1/-23.31 & 45.5/10.6/-82.49 & 50.9/17.2/-23.71 \\
\midrule
TX & A & PM$_{2.5}$ & 89.3/7.34/-2.64 & 94.0/7.88/-2.67 & 93.3/7.17/-2.64 & 90.5/7.04/-2.53 & 90.9/5.98/-2.44 \\
TX & A & Ozone & 93.7/0.0098/4.55 & 93.2/0.0101/4.49 & 93.1/0.0103/4.47 & 93.7/0.00978/4.54 & 94.2/0.00943/4.55 \\
TX & A & NO$_2$ & 93.5/12.4/-2.76 & 95.0/13.6/-2.75 & 93.1/12.4/-2.78 & 92.3/12.1/-2.76 & 92.3/11.8/-2.81 \\
\addlinespace[2pt]
TX & B & PM$_{2.5}$ & 88.8/7.1/-2.63 & 94.0/7.76/-2.65 & 93.3/7.23/-2.67 & 91.2/6.88/-2.52 & 90.7/6.2/-2.45 \\
TX & B & Ozone & 93.7/0.00981/4.56 & 94.4/0.0101/4.52 & 93.2/0.0102/4.48 & 93.4/0.00979/4.54 & 94.2/0.00957/4.53 \\
TX & B & NO$_2$ & 95.6/12.6/-2.68 & 95.4/13.3/-2.73 & 93.1/12.6/-2.69 & 93.3/12.3/-2.65 & 92.5/12.3/-2.69 \\
\addlinespace[2pt]
TX & C & PM$_{2.5}$ & 93.0/10.5/-2.73 & 93.0/9.96/-2.52 & 93.3/9.87/-2.53 & 93.0/9.83/-2.66 & 95.3/8.51/-2.71 \\
TX & C & Ozone & 100.0/0.0134/4.45 & 91.1/0.0141/4.16 & 95.8/0.014/4.21 & 100.0/0.0135/4.45 & 99.7/0.013/4.45 \\
TX & C & NO$_2$ & 95.2/14.5/-2.74 & 87.5/14.7/-3.08 & 75.4/16.5/-3.84 & 89.6/14.2/-2.90 & 77.3/13.1/-3.63 \\
\bottomrule
\end{tabular}
\end{table}
\space
\endgroup 
\space
\space

\FloatBarrier
\section*{{Code availability}}
{Code, experiment configurations, random seeds and reproduction instructions are available at \url{https://github.com/Magica-Chen/heterotpic-gp}.}

\clearpage
\phantomsection
\addcontentsline{toc}{section}{References}
\bibliography{references}

\end{document}